\documentclass[11pt]{article}
\usepackage{amsmath,amsthm,amssymb,amsfonts}
\usepackage{mathtools}
\usepackage{mathrsfs}
\usepackage{multirow}
\usepackage[inline]{enumitem}
\usepackage{hyperref}
\usepackage{bm}
\usepackage[margin=1in]{geometry}
\usepackage{colortbl}
\usepackage{booktabs}
\usepackage[mathscr]{euscript}
\usepackage{microtype}

\usepackage{sectsty}
\usepackage{setspace}
\usepackage{caption}
\usepackage{todonotes}

\usepackage{kbordermatrix}
\renewcommand{\kbldelim}{(}
\renewcommand{\kbrdelim}{)}

\newtheorem{assumption}{Assumption}
\newtheorem{theorem}{Theorem}
\newtheorem{corollary}{Corollary}
\newtheorem{lemma}{Lemma}
\newtheorem{proposition}{Proposition}
\theoremstyle{remark}
\newtheorem{remark}{Remark}
\newtheorem{example}{Example}

\usepackage[round]{natbib}
 \bibpunct[, ]{(}{)}{,}{a}{}{,}%

\input{latexmacro}

\newcommand{\dif}{\mathop{}\!\mathrm{d}}
\newcommand{\SFM}{\mathsf M}

\newcommand{\SFZ}{\mathsf Z}
\newcommand{\CalO}{\mathcal O}
\newcommand{\CalT}{\mathcal T}
\newcommand{\CalX}{\mathcal X}

\begin{document}

\title{Scalable Stochastic Kriging with Markovian Covariances}

\author{Liang Ding and Xiaowei Zhang\thanks{Corresponding author. Email: \href{mailto:xiaoweiz@ust.hk}{xiaoweiz@ust.hk}}}

\date{\small{Department of Industrial Engineering and Decision Analytics \\ The Hong Kong University of Science and Technology, Clear Water Bay, Hong Kong}}

\maketitle

\begin{abstract}
Stochastic kriging is a popular technique for simulation metamodeling due to its flexibility and analytical tractability. Its computational bottleneck is the inversion of a covariance matrix, which takes $\CalO(n^3)$ time in general and becomes prohibitive for large $n$, where $n$ is the number of design points. Moreover, the covariance matrix is often ill-conditioned for large $n$, and thus the inversion is prone to numerical instability, resulting in erroneous parameter estimation and prediction. These two numerical issues preclude the use of stochastic kriging at a large scale. This paper presents a novel approach to address them. We construct a class of covariance functions, called Markovian covariance functions (MCFs), which have two properties: (i) the associated covariance matrices  can be inverted analytically, and (ii) the inverse matrices are sparse. With the use of MCFs, the inversion-related computational time is reduced to $\CalO(n^2)$ in general, and can be further reduced by orders of magnitude with additional assumptions on the simulation errors and design points. The analytical invertibility also enhance the numerical stability dramatically. The key  in our approach is that we identify a general functional form of covariance functions that can induce sparsity in the corresponding inverse matrices. We also establish a connection between MCFs and linear ordinary differential equations. Such a connection provides a  flexible, principled approach to constructing a wide class of MCFs. Extensive numerical experiments demonstrate that stochastic kriging with MCFs can handle large-scale problems in an both computationally efficient and numerically stable manner. 

\end{abstract}


\section{Introduction}

Simulation is used extensively to facilitate decision-making processes related to complex systems. The popularity stems from its flexibility, allowing users to incorporate arbitrarily fine details of the system and estimate virtually any performance measure of interest. However, simulation models are often computationally expensive to execute, severely restricting the usefulness of simulation in settings such as real-time decision making and system optimization. In order to alleviate this computational inefficiency, metamodeling has been developed actively in the simulation community \citep{BartonMeckesheimer06}. The basic idea is that the user only executes the simulation model at some carefully selected design points. A metamodel, which runs much faster than the simulation model in general, is then built to approximate the true response surface -- the performance measure of the simulation model -- as a function of the design variables, by interpolating the simulation outputs properly. The responses at other locations are predicted by the metamodel without running additional simulation, thereby reducing the computational cost substantially. 

Stochastic kriging (SK), proposed by \cite{AnkenmanNelsonStaum10}, is a particularly popular metamodel, thanks to its analytical tractability, ease of use, and capability of providing good global fit. It has been used successfully for quantifying input uncertainty in stochastic simulation \citep{BartonNelsonXie14,XieNelsonBarton14} and for optimizing expensive functions with noisy observations \citep{sun2014}.  SK represents the response surface as a Gaussian process, which is fully characterized by its covariance function, and leverages the spatial correlations between the responses to provide prediction. However, one often encounters two numerical issues when implementing SK in practice, both of which are related to matrix inversion. Indeed, the inverse of the covariance matrix of the simulation outputs is essential for computing various quantities in SK, including the optimal predictor, the mean squared error of prediction, and the likelihood function.

An immediate issue regarding the inversion of a $n\times n$ matrix is that it typically requires $\CalO(n^3)$ computational time, which is prohibitive for large $n$, where $n$ is the number of the design points. For instance, it is reported in \cite{huang2006} that a major limitation of SK-based methods for simulation optimization is the high computational cost of fitting the SK metamodel, which,  as the number of samples increases,  eventually becomes even more expensive than running the original simulation model. 
 
A second issue is that the covariance matrix involved in SK may become ill-conditioned (i.e., nearly singular), in which case the inversion is numerically unstable, resulting in inaccurate parameter estimation or prediction. This often occurs when $n$ is large, because then there are fairly likely two design points that are spatially close to each other, and thus the two corresponding columns in the covariance matrix are ``close to'' being linearly dependent. 

These two numerical issues preclude the use of SK at a large scale, especially for problems with a high-dimensional design space. In geostatistics literature, inverting large covariance matrices that arise from Gaussian processes is a well-known numerical challenge and is sometimes referred to  as ``the big $n$ problem'' informally. Typical solutions to this problem are based on approximations, that is, use another matrix that is easier to invert to approximate the covariance matrix; see \S \ref{sec:lit} for more details. This paper presents a new perspective. \emph{Instead of seeking good approximations for covariance matrices induced by an arbitrary covariance function, we will construct a specific class of covariance functions that induce computationally tractable covariance matrices.} In particular, the computational tractability stems from the following two properties of the covariance matrices induced by this class of covariance functions: (i) they can be inverted \emph{analytically}, and (ii) the inverse matrices are \emph{sparse}. Our novel approach will effectively reduce the computational complexity of SK to $\CalO(n^2)$, without resorting to approximation schemes. In situations where the simulation errors are negligible, our approach obviates the need of numerical inversion and further reduces the complexity to $\CalO(n)$.

We refer to this specific class of covariance functions as \emph{Markovian covariance functions} (MCFs), because the Gaussian processes equipped with them exhibit certain Markovian structure.  Albeit seemingly restrictive, MCFs actually represent a broad class of covariance functions and can be constructed in a flexible, convenient fashion.

\subsection{Main Contributions}
First and foremost, we identify a simple but general functional form with which the covariance function of a 1-dimensional Gaussian process yields \emph{tridiagonal} precision matrices (i.e., the inverse of the covariance matrices), which are obviously sparse. In addition, the nonzero entries of the precision matrices can be expressed in terms of the covariance function in closed-form. To the best of our knowledge, there is no prior result establishing this kind of explicit connection between the form of covariance functions and sparsity in the corresponding precision matrices.

Second, we link MCFs to Sturm-Liouville (S-L) differential equations. Specifically, we show that the Green's function of an S-L equation has exactly the same form as MCFs. Not only does this connection provide a convenient tool to construct MCFs, but also implies that the number of MCFs having an analytical expression is potentially enormous, since any second-order linear ordinary differential equation can be recast in the form of an S-L equation.

Third, we extend MCFs to multidimensional design spaces in a ``composite'' manner, namely, defining the multidimensional covariance to be the product of 1-dimensional covariances along each dimension. This way of construction allows use of tensor algebra to preserve the sparsity in the precision matrices, provided that the design points form a regular lattice.


Last but not least, we demonstrate through extensive numerical experiments that MCFs can significantly outperform those that are commonly used such as the squared exponential covariance function in terms of accuracy in prediction of response surfaces. The improved accuracy can be attributed to two reasons: (i) the numerical stability of matrix inversion is enhanced greatly; (ii) the reduced computational complexity allows us to use more data.

\subsection{Related Literature}\label{sec:lit}

A great variety of techniques have been proposed to address the big $n$ problem in both geostatistics and machine learning literature, where Gaussian processes are widely used. Most of them focus on developing approximations of the covariance matrix that are computationally cheaper. Representative approximation schemes include reduced-rank approximation and sparse approximation. The former approximates the covariance matrix by a matrix having a much lower rank. The latter can be achieved by a method called covariance tapering. It forces the covariance to zero if the two design points involved are sufficiently far away from each other. The covariance matrix is then approximated by a sparse matrix. Both reduced-rank matrices and sparse matrices entail fast inversion algorithms. From a modeling perspective, these two approximation schemes emphasize long-scale and short-scale dependences respectively, but meanwhile fail to capture the other end of the spectrum \citep{SangHuang12}.  We refer to \citet[Chapter 12]{BanerjeeCarlinGelfand14} and \citet[Chapter 8]{RasmussenWilliams06} for reviews with a focus on geostatistics and machine learning, respectively. Moreover, approximation schemes usually result in spurious quantification of the uncertainty about the prediction; see, e.g., \cite{ShahriariSwerskyWangAdamsdeFreitas16} and references therein.

Another popular approach to the big $n$ problem is to use Gaussian Markov random fields (GMRFs), which discard the concept of covariance function and model the precision matrix, i.e., the inverse of the covariance matrix,  directly; see \citet{RueHeld05} for a thorough exposition on the subject and \cite{SalemiSongNelsonStaum17} for its application in large-scale simulation optimization. To construct a GMRF one first stipulates a graph, with nodes denoting locations of interest in the design space. The edges in the graph characterize the ``neighborhood'' of each node, and define a Markovian structure. In particular, given all its neighbors, each node is conditionally independent of its non-neighbors. A crucial property of GMRF is that entry $(i,j)$ of the precision matrix is nonzero if, and only if, node $i$ and node $j$ are neighbors. Hence, the precision matrix is sparse if each node has a small neighborhood. The sparsity is then taken advantage of to reduce the inversion-related computational time. 

Despite its computational efficiency, GMRFs have clear disadvantages. First and foremost, they do not model association directly, and thus one cannot specify desired correlation behavior. Indeed, the relationship between entries in the precision matrix and the covariance matrix is very complex. This is because the joint distribution of the responses at two locations depends on the joint distribution of the responses at all the other locations. Second, GMRFs are built on graphs, and the discrete nature forbids predicting responses at locations that are not included in the graph, which is problematic for continuous design spaces. 

The methodology developed in the present paper is closely related to GMRFs. Our work can be viewed as one way to extend GMRFs from discrete domains to continuous domains. But it is by no means  a trivial extension, because we establish an explicit relationship between the form of a covariance function and the sparsity in the corresponding precision matrices. This allows us to combine the best of two worlds -- modeling association directly while preserving the computational tractability of GMRFs.


The rest of the paper is organized as follows. In \S\ref{sec:model}, we introduce the SK metamodel and motivate our approach to the big $n$ problem. In \S\ref{sec:MCF}, we introduce MCFs and characterize their essential structure, which effectively bridges the gap between Gaussian processes and GMRFs. In \S \ref{sec:SL}, we link MCFs with S-L differential equations. In \S\ref{sec:MLE}, we discuss maximum likelihood estimation of the unknown parameters, with an emphasis on the numerical stability as a result of the use of MCFs. We conduct extensive numerical experiments in \S\ref{sec:numerical} to demonstrate the scalability of SK in the presence of MCFs, and conclude in \S\ref{sec:conclusions}. The Appendices collect some technical proofs. 

\section{Stochastic Kriging and the Big $n$ Problem}\label{sec:model}

Let $\BFx\in \mathscr{X}\subseteq\Real^D$ denote the design variable of a computationally expensive simulation model, with $\mathscr{X}$ being the design space. Let $\SFZ(\BFx)$ denote the unknown response surface of that model. Suppose that the simulation model is run at design point $\BFx_i$ with $r_i$ independent replications, producing outputs $z_j(\BFx_i)$, $j=1,\ldots,r_i$, $i=1,\ldots,n$. Metamodeling is concerned with fitting  $\SFZ(\BFx)$ based on the simulation outputs. The SK metamodel casts $\SFZ(\BFx)$ into a \textit{realization} of a Gaussian process,
\begin{equation}
\label{eq:uni_nriging}
\SFZ(\BFx) = \BFbeta^\intercal \BFf(\BFx)+ \SFM(\BFx),
\end{equation}
where $\BFf(\BFx)$ is a vector of known functions (e.g., polynomial basis functions) and $\BFbeta$ is a vector of unknown parameters of compatible dimension, and $\SFM$ is a mean zero Gaussian process that is randomly sampled from a space of functions mapping $\Real^D \mapsto\Real $. A particular feature of the SK metamodel \eqref{eq:uni_nriging} is the spatial correlation, i.e., $\SFM(\BFx)$ and $\SFM(\BFy)$ tend to be similar (resp., different) if $\BFx$ and $\BFy$ are close to (resp., distant from) each other in space. Let $k(\BFx,\BFy)\coloneqq\Cov(\SFM(\BFx), \SFM(\BFy))$ denote the covariance function of $\SFM$. It is crucial to specify $k(\BFx,\BFy)$ properly in order that SK provide a good fit globally over the design space $\mathscr{X}$. 

The simulation outputs become
\begin{equation*}\label{eq:SK}
z_j(\BFx_i) =  \SFZ(\BFx) + \varepsilon_j(\BFx_i)=\BFbeta^\intercal \BFf(\BFx_i) + \SFM(\BFx_i) + \varepsilon_j(\BFx_i),
\end{equation*}
where $\varepsilon_1(\cdot),\varepsilon_2(\cdot),\ldots$ are normally distributed simulation errors. Define $\bar z(\BFx_i)\coloneqq  r_i^{-1}\sum_{j=1}^{r_i}z_j(\BFx_i)$, $\bar \varepsilon(\BFx_i)\coloneqq  r_i^{-1}\sum_{j=1}^{r_i}\varepsilon_j(\BFx_i)$, and $\bar \BFz \coloneqq (\bar z(\BFx_1),\cdots,\bar z(\BFx_n))^\intercal$. Let $\BFSigma_\SFM$ denote the $n\times n$ covariance matrix of $(\SFM(\BFx_1),\ldots,\SFM(\BFx_n))$, i.e., entry $(i,j)$ of $\BFSigma_\SFM$ is $\Cov(\SFM(\BFx_i),\SFM(\BFx_j))=k(\BFx_i,\BFx_j)$. Likewise, let $\BFSigma_\varepsilon$ denote the covariance matrix of $(\bar \varepsilon(\BFx_1),\ldots,\bar \varepsilon(\BFx_n))$. 

We assume that the simulation errors are mutually independent and are independent of $\SFM$. This assumption effectively  rules out the use of common random numbers (CRN) because it will break the sparsity that our methodology critically hinges on. Nevertheless, this does not impose much practical restriction, since it is shown in \cite{ChenAnkenmanNelson12} that the use of CRN generally is detrimental to the prediction accuracy of SK.

Let $\BFx_0$ denote an arbitrary point in $\mathscr{X}$. SK is concerned with predicting $\SFZ(\BFx_0)$ based on $\bar \BFz$. The SK predictor that minimizes the mean squared error (MSE) of prediction is
\begin{equation}\label{eq:BLUP}
\widehat\SFZ(\BFx_0) = \BFbeta^\intercal \BFf(\BFx_0) + \BFgamma^\intercal  (\BFx_0) [\BFSigma_\SFM + \BFSigma_\varepsilon]^{-1}(\bar\BFz-\BFF\BFbeta),
\end{equation}
with optimal MSE 
\begin{equation}\label{eq:MSE}
\MSE^*(\BFx_0) = k(\BFx_0,\BFx_0) - \BFgamma^\intercal(\BFx_0)[\BFSigma_\SFM + \BFSigma_\varepsilon]^{-1}\BFgamma(\BFx_0),
\end{equation}
where $\BFgamma\coloneqq (k(\BFx_0, \BFx_1),\ldots,k(\BFx_0, \BFx_n))^\intercal$ and $\BFF\coloneqq (\BFf(\BFx_1),\ldots,\BFf(\BFx_n))^\intercal$, provided that $\BFbeta$, $\BFgamma(\BFx_0)$, $\BFSigma_\SFM$, and $\BFSigma_\varepsilon$ are known. Clearly, they need to be estimated from the simulation outputs in practice. A typical method for estimating the unknown parameters is the maximum likelihood estimation (MLE), which maximizes the following log-likelihood function 
\begin{equation}\label{eq:loglikelihood}
l(\BFbeta,\BFtheta) = -\frac{n}{2}\ln(2\pi) - \frac{1}{2}\ln|\BFSigma_\SFM + \BFSigma_\varepsilon| -\frac{1}{2}(\bar\BFz-\BFF\BFbeta)^\intercal [\BFSigma_\SFM + \BFSigma_\varepsilon]^{-1} (\bar\BFz-\BFF\BFbeta),
\end{equation}
where $|\cdot|$ denotes the determinant of a matrix and $\BFtheta$ denotes the unknown parameter involved for specifying the covariance function $k$; see \S\ref{sec:MLE} for more discussion. 

Obviously, computing \eqref{eq:BLUP}, \eqref{eq:MSE}, and \eqref{eq:loglikelihood} all requires inverting $\BFSigma_\SFM+\BFSigma_\varepsilon$, which comes with two numerical challenges and is referred to as the big $n$ problem in geostatistics literature \citep{BanerjeeCarlinGelfand14}. First, although $\BFSigma_\varepsilon$ is diagonal due to the independence assumption, $\BFSigma_\SFM+\BFSigma_\varepsilon$ is a dense matrix in general and inverting it typically takes $\CalO(n^3)$ computational time, which becomes prohibitive for large $n$ (e.g., $n>10^3$). Second, this matrix often becomes ill-conditioned, and thus inverting it is prone to numerical instability. This may happen either if there are two design points spatially close to each other (so that the two corresponding columns of $\BFSigma_\SFM$ are almost linearly dependent), or during the process of searching the parameter space for an estimate of $\BFtheta$ for maximizing \eqref{eq:loglikelihood}. Moreover, both of the issues will be amplified by the dimensionality of the design space. 

Existing solutions to the big $n$ problem heavily rely on approximation schemes, striving to approximate $\BFSigma_\SFM$ by another matrix that can be inverted much faster. However, the reduction in computational time comes at the cost of inaccurate prediction of the responses and even invalid characterization their variances; see, e.g., \cite{Quinonero-CandelaRasmussen05}, \cite{SangHuang12}, and references therein. 

By contrast, we propose in this paper a novel approach to the big $n$ problem. Instead of allowing any arbitrary covariance function and then seeking approximations of the associated covariance matrices, we will devise judiciously a specific but broad class of covariance functions having the following two properties: (i) $\BFSigma_\SFM$ can be inverted analytically, and (ii) $ \BFSigma_\SFM^{-1}$ is sparse.

These two properties make the computation of $[\BFSigma_\SFM + \BFSigma_\varepsilon]^{-1}$ substantially easier. To see this, notice that by the Woodbury matrix identity \cite[\S 0.7.4]{HornJohnson12}, 
\begin{equation}\label{eq:woodbury}
[\BFSigma_\SFM + \BFSigma_\varepsilon]^{-1} = \BFSigma_\SFM^{-1} -  \BFSigma_\SFM^{-1}[ \BFSigma_\SFM^{-1} +  \BFSigma_\varepsilon^{-1}]^{-1} \BFSigma_\SFM^{-1}. 
\end{equation}
Since $\BFSigma_\SFM^{-1}$ has a known analytical expression and $\BFSigma_\SFM^{-1} +  \BFSigma_\varepsilon^{-1}$ is sparse, $[\BFSigma_\SFM^{-1} +  \BFSigma_\varepsilon^{-1}]^{-1}$ can be computed in $\CalO(n^2)$ time by leveraging a particular sparse structure that will become clear in \S\ref{sec:MCF}. The matrix multiplications in \eqref{eq:woodbury} require $\mathcal O(n^2)$ time also due to the sparsity of $\BFSigma_\SFM$, as opposed to $\CalO(n^3)$ for multiplications of dense matrices. Therefore, computing \eqref{eq:woodbury} requires $\CalO(n^2)$ time, reducing one order of magnitude without resorting to any matrix approximation at all. Further, if the simulation errors are negligible, i.e., $\BFSigma_\varepsilon\approx\BFzero$, then $[\BFSigma_\SFM + \BFSigma_\varepsilon]^{-1}\approx\BFSigma_\SFM^{-1}$, which can be inverted analytically, then numerical inversion would become unnecessary. This implies that the computation of the SK predictor \eqref{eq:BLUP}, which is reduced to multiplications of vectors and sparse matrices, can be completed in $\CalO(n)$ time. The same goes for the computation of the optimal MSE \eqref{eq:MSE}. 

Two central questions follow immediately. What structure needs to be imposed on the covariance function $k(\BFx,\BFy)$ so that the covariance matrix $\BFSigma_\SFM$ has the two desirable properties? How broad is this specific class of covariance functions? This paper provides comprehensive answers.

\section{Markovian Covariance Functions} \label{sec:MCF}

In order to motivate the structure that we impose on the covariance function, we first introduce Gaussian Markov random fields (GMRFs) briefly and refer to \cite{RueHeld05} for a comprehensive treatment on the subject. Consider a graph consisting of $n$ nodes, each of which is labeled with $\BFx_i$ and has a random value $\SFM(\BFx_i)$, $i=1,\ldots,n$. Let $\BFX$ denote all the nodes and $\mathscr{N}(\BFx_i)$ denote the neighbors of $\BFx_i$, for each $i=1,\ldots,n$.  Suppose that the joint distribution of $(\SFM(\BFx_1),\ldots,\SFM(\BFx_n))$ is multivariate normal. Then, $(\SFM(\BFx_1),\ldots,\SFM(\BFx_n))$ is called a GMRF if it has the  Markovian structure (i.e., conditional independence structure) as follows. Given $\{\SFM(\BFx):\BFx\in\mathscr{N}(\BFx_i)\}$, the values of the neighbors of node $\BFx_i$, $\SFM(\BFx_i)$ is conditionally independent of the values of its non-neighbors, $\{\SFM(\BFx):\BFx\in E \setminus \mathscr{N}(\BFx_i)\}$. A critical property of GMRFs is that entry $(i,j)$ of the precision matrix $\BFSigma_\SFM^{-1}$ is nonzero if, and only if, $\BFx_i$ and $\BFx_j$ are neighbors. Hence, $\BFSigma_\SFM^{-1}$ is sparse  if each node has a small neighborhood in the graph. 

The fundamental cause for the sparsity of $\BFSigma_\SFM^{-1}$ in GMRFs is obviously the Markovian structure. This inspires us to consider Gaussian processes that are Markovian. In particular, we consider three 1-dimensional examples -- Brownian motion, Brownian bridge, and the Ornstein-Uhlenbeck (O-U) process -- and calculate their associated precision matrices, respectively. 


\begin{example}[Brownian Motion]\label{example:BM}
The covariance function of the standard 1-dimensional Brownian motion is $k_{\mathrm{BM}}(s,t)=\min(x,y)$, $x,y\geq 0$. Suppose that the design points $\{x_1,\ldots,x_n\}$ are equally spaced, i.e., $x_i=ih$ for some $h>0$. Then, it can be shown that $\BFSigma_\SFM^{-1}$ is a tridiagonal matrix:
\[\BFSigma_\SFM^{-1} = \frac{1}{h}
\begin{pmatrix}
2 & -1 &   &  &  \\
-1 & 2 & -1 &     & \\
\cdots &  & \cdots &  &  \cdots\\
 & &   -1 & 2 & -1 \\
 & &  & -1 & 1
\end{pmatrix}. 
\]
\end{example}

\begin{example}[Brownian Bridge]\label{example:BB} 
The covariance function of the Brownian bridge defined on $[0,1]$ is $k_{\mathrm{BB}}(x,y) = \min(x,y)-xy$,  $x,y\in[0,1]$. Suppose that the design points are $x_i=\frac{i}{n+1}$, $i=1,\ldots,n$. Then, it can be shown that $\BFSigma_\SFM^{-1}$ is a tridiagonal matrix:
\[\BFSigma_\SFM^{-1} = (n+1)
\begin{pmatrix}
2 & -1 &   &  &  \\
-1 & 2 & -1 &     & \\
\cdots &  & \cdots &  &  \cdots\\
 & &   -1 & 2 & -1 \\
 & &  & -1 & 2
\end{pmatrix}. 
\]
\end{example}

\begin{example}[O-U Process] \label{example:OU}
The O-U process is defined via the  stochastic differential equation 
\[\dif X(t) = (\mu-\theta X(t))\,\dif t + \sigma\,\dif B(t), \quad t\geq 0\]
where $\mu$, $\theta>0$, and $\sigma>0$ are parameters, and $B(t)$ is the standard 1-dimensional Brownian motion. Then, the covariance function under the stationary distribution is $k_{\mathrm{OU}}(x,y) =\frac{\sigma^2}{2\theta} e^{-\theta|x-y|}$,  $x,y\geq 0$. Using the same design points as  Example \ref{example:BM}, it can be shown that $\BFSigma_\SFM^{-1}$ is a tridiagonal matrix:\footnote{The discovery of the precision matrices associated with the O-U process being tridiagonal was initially made through several numerical trials. Together with Examples \ref{example:BM} and \ref{example:BB}, the tridiagonal pattern was already enough to motivate us to consider the functional form \eqref{eq:kernel}. The analytical expression of the precision matrix in Example \ref{example:OU} was calculated as a corollary of Theorem \ref{theo:tridiag} after we proved it.}
\[\BFSigma_\SFM^{-1} = \frac{\theta }{\sigma^2 \sinh(\theta h)}
\begin{pmatrix}
\exp(\theta h) & -1 & \phantom{\exp(\theta h)}  &  &   \\
-1 & 2\cosh(\theta h) & -1 &    &   \\
\cdots &  &  \cdots &  &  \cdots \\
 & &   -1 & 2\cosh(\theta h) & -1 \\
 & & & -1 & \exp(\theta h)
\end{pmatrix}.
\]
\end{example}

Now that all the three examples have tridiagonal precision matrices, we naturally try to find the common feature in their covariance functions.

\subsection{Symmetric Tridiagonal Structure}

The key observation here is that the covariance functions in Examples \ref{example:BM}--\ref{example:OU} share the same form: 
\begin{equation}\label{eq:kernel}
k(x,y)=p(x)q(y)\ind_{\{x\leq y\}} + p(y)q(x)\ind_{\{x>y\}},
\end{equation}
for some functions $p$ and $q$, where $\ind_{\{\cdot\}}$ is the indicator function. Specifically, 
\begin{align*}
k_{\mathrm{BM}}(x,y) = & \min(x,y) =  x \ind_{\{x\leq y\}} + y \ind_{\{x>y\}},\\
k_{\mathrm{BB}}(x,y) = & \min(x,y)-xy = x(1-y) \ind_{\{x\leq y\}} + y(1-x) \ind_{\{x>y\}}, \\ 
k_{\mathrm{OU}}(x,y) = & \frac{\sigma^2}{2\theta} e^{-\theta|x-y|} = \frac{\sigma^2}{2\theta}\left[e^{\theta x}e^{-\theta y}\ind_{\{x\leq y\}} + e^{\theta y}e^{-\theta x}\ind_{\{x> y\}}\right].
\end{align*}

Therefore, we conjecture that for Gaussian processes with a 1-dimensional domain, a covariance function of form \eqref{eq:kernel} would yield tridiagonal precision matrices. This turns out to be true in general under mild conditions and the design points do not need to be equally spaced. We present the result below as Theorem \ref{theo:tridiag}. The proof is done by induction on $n$ and is based on explicit calculations. We will use the Laplace expansion for the determinant of a square matrix. This is a classic result in linear algebra; see \citet[\S0.3.1]{HornJohnson12}.

\begin{lemma}[Laplace Expansion]\label{lemma:laplace}
Let $\BFK=(k_{i,j})$ be a $n\times n$ matrix and $M_{i,j}$ be its $(i,j)$ minor, i.e., the determinant of the submatrix formed by deleting the $i^{\mathrm{th}}$ row and $j^{\mathrm{th}}$ column of $\BFK$. Then, 
\[|\BFK| = \sum_{\ell=1}^n (-1)^{i+\ell} k_{i,\ell} M_{i,\ell} = \sum_{\ell=1}^n (-1)^{\ell+j} k_{\ell,j} M_{\ell,j}.\]
\end{lemma}

To facilitate the presentation, we define several notations. Let $\CalX = \{x_1,\ldots,x_n\}$ denote a set of distinct points in $\Real$, with $x_1<\cdots<x_n$. Fixing a function $k(x,y)$ of the form \eqref{eq:kernel}, let $\BFK=\BFK(\CalX,\CalX)$ be the $n\times n$ matrix whose entry $(i,j)$ is $k(x_i,x_j)$. For two subsets $\CalX',\CalX''\subseteq \CalX$, we use $\BFK(\CalX',\CalX'')$ to denote the submatrix of $\BFK$ formed by keeping the rows and columns that correspond to $\CalX'$ and $\CalX''$, respectively. Finally, let $p_i=p(x_i)$ and $q_i=q(x_i)$, $i=1,\ldots,n$.

\begin{theorem}\label{theo:tridiag} 
Let $n\geq 3$. If $\BFK$ is nonsingular, then $\BFK^{-1}$ is a symmetric tridiagonal matrix. 
\end{theorem}
\begin{proof}  Since $k(x,y)=k(y,x)$, the symmetry of $\BFK$ is straightforward, and thus $\BFK^{-1}$ is symmetric. 

To prove that $\BFK^{-1}$ is tridiagonal, i.e., $(\BFK^{-1})_{i,j}=0$ if $|j-i|\geq 2$, we use the relationship between the inverse and the minors of a square matrix \citep[\S0.8.2]{HornJohnson12}, 
\begin{equation}\label{eq:adjugate}
(\BFK^{-1})_{i,j} = \frac{1}{|\BFK|} (-1)^{i+j} M_{j,i},
\end{equation}
where $M_{j,i}$ is the $(j,i)$ minor of $\BFK$. Hence, it suffices to show that $M_{i,j}=0$ if $|j-i|\geq 2$, or equivalently,
\begin{equation}\label{eq:minor}
|\BFK(\CalX\setminus\{x_i\}, \CalX\setminus\{x_j\})|=0, \quad \mbox{if }j-i\geq 2,
\end{equation}
because of the symmetry of $\BFK$.  We prove \eqref{eq:minor} by induction on $n$. For $n=3$, 
\begin{equation*}
\BFK(\CalX,\CalX)=\begin{pmatrix} 
     p_1q_1 & p_1q_2 & p_1q_3\\
     p_1q_2 & p_2q_2 & p_2q_3\\
     p_1q_3 & p_2q_3 & p_3q_3\\
    \end{pmatrix}.
\end{equation*}
Then, 
\begin{equation*}
M_{1,3}=\begin{vmatrix} 
     p_1q_2 & p_2q_2 \\
     p_1q_3 & p_2q_3 \\
    \end{vmatrix}=0.
\end{equation*}

Now we suppose that \eqref{eq:minor} holds for any $n\leq N-1$. Then, for $n=N$ and $j\geq i+2$,
\begin{align}
M_{i,j} =& |\BFK(\CalX\setminus\{x_i\}, \CalX\setminus\{x_j\})| \nonumber\\
=& \sum_{\ell<j}  (-1)^{(j-1)+\ell} k(x_j,x_\ell) |\BFK(\CalX\setminus\{x_i,x_j\}, \CalX\setminus\{x_j,x_\ell\})| \label{eq:minor_expan_1} \\  
+ &\sum_{\ell>j}  (-1)^{(j-1)+(\ell-1)} k(x_j,x_\ell) |\BFK(\CalX\setminus\{x_i,x_j\}, \CalX\setminus\{x_j,x_\ell\})| \nonumber  
\end{align}
where the second equality follows from the Laplace expansion along the row of the submatrix $\BFK(\CalX\setminus\{x_i\}, \CalX\setminus\{x_j\})$ that corresponds to $x_j$. Here, $(j-1)$ and $(\ell-1)$ in the exponents reflect the necessary changes in the indices of the rows and columns of submatrix $\BFK(\CalX\setminus\{x_i\}, \CalX\setminus\{x_j\})$.

Let $\CalX'=\CalX\setminus\{x_j\}$. Then, the submatrix that appears in the Laplace expansion in  \eqref{eq:minor_expan_1} can be rewritten as $\BFK(\CalX\setminus\{x_i,x_j\}, \CalX\setminus\{x_j,x_\ell\})=\BFK(\CalX'\setminus\{x_i\}, \CalX'\setminus\{x_\ell\})$. Hence, its determinant is the $(i,\ell)$ minor of $\BFK(\CalX', \CalX')$ if $\ell<j$, or the $(i, \ell-1)$ minor if $\ell>j$. It follows that $\BFK(\CalX'\setminus\{x_i\}, \CalX'\setminus\{x_\ell\})=0$ if $|\ell-i|\geq 2$ and $\ell<j$, or if $|\ell-1-i|\geq 2$ and $\ell>j$, by the induction assumption. Therefore, \eqref{eq:minor_expan_1} can be simplified to 
\begin{equation}\label{eq:minor3}
M_{i,j} = \sum_{\ell=i-1}^{i+1} (-1)^{(j-1)+\ell} k(x_j,x_\ell) |\BFK(\CalX'\setminus\{x_i\}, \CalX'\setminus\{x_\ell\})|, 
\end{equation}
since $j\geq i+2$. Clearly, it suffices to show
\begin{equation}\label{eq:minor2}
\sum_{\ell=i-1}^{i+1} (-1)^{\ell} k(x_j,x_\ell) |\BFK(\CalX'\setminus\{x_i\}, \CalX'\setminus\{x_\ell\})| = 0,
\end{equation}
in order to prove \eqref{eq:minor}. To that end, we further apply the Laplace expansion. 

We now assume that $i\geq 4$. The cases $i=1,2,3$ can be proved in a similar fashion. For $\ell=i-1$, $\BFK(\CalX'\setminus\{x_i\}, \CalX'\setminus\{x_\ell\})$ is 
\setcounter{MaxMatrixCols}{20}
\begin{equation}\label{eq:bigmatrix}
\kbordermatrix{
 & x_1  & x_2 & \cdots &  x_{i-2} & x_{i} & x_{i+1} & \cdots  & x_{j-1} & x_{j+1}&\cdots & x_N \\
x_1 & p_1q_1  & p_1q_2 & \cdots &  p_1q_{i-2} & p_1q_{i} & p_1q_{i+1} & \cdots  & p_1q_{j-1} & p_1q_{j+1}&\cdots & p_1q_N \\
x_2 & p_1q_2 & p_2q_2 & \cdots & p_2q_{i-2} & p_2q_{i} & p_2q_{i+1} & \cdots & p_2q_{j-1} & p_2q_{j+1} &\cdots & p_2q_N \\
\vdots& &  & \vdots &  &   &   & \vdots & & & \vdots &\\
x_{i-1} &p_1q_{i-1} & p_2q_{i-1} & \cdots & p_{i-2}q_{i-1} & p_{i-1}q_{i} & p_{i-1}q_{i+1} & \cdots & p_{i-1}q_{j-1} & p_{i-1}q_{j+1} &\cdots & p_{i-1}q_N \\
x_{i+1} &p_1q_{i+1} & p_2q_{i+1} & \cdots & p_{i-2}q_{i+1} & p_{i}q_{i+1} & p_{i+1}q_{i+1} & \cdots & p_{i+1}q_{j-1} & p_{i+1}q_{j+1} &\cdots & p_{i+1}q_N \\
\vdots & &  & \vdots &  &   &   & \vdots & & & \vdots &\\
x_{j-1} &p_1q_{j-1} & p_2q_{j-1} & \cdots & p_{i-2}q_{j-1} & p_{i}q_{j-1} & p_{i+1}q_{j-1} & \cdots & p_{j-1}q_{j-1} & p_{j-1}q_{j+1} &\cdots & p_{j-1}q_N \\
x_{j+1} &p_1q_{j+1} & p_2q_{j+1} & \cdots & p_{i-2}q_{j+1} & p_{i}q_{j+1} & p_{i+1}q_{j+1} & \cdots & p_{j-1}q_{j+1} & p_{j+1}q_{j+1} &\cdots & p_{j+1}q_N \\
\vdots & &  & \vdots &  &   &   & \vdots & & & \vdots &\\
x_N & p_1q_N & p_2q_N & \cdots & p_{i-2}q_N & p_{i}q_N & p_{i+1}q_N & \cdots & p_{j-1}q_N & p_{j+1}q_N &\cdots & p_Nq_N
}.
\end{equation}

With $i\geq 4$, the transpose of the submatrix of $\BFK(\CalX'\setminus\{x_i\}, \CalX'\setminus\{x_{i-1}\})$ formed by deleting the first row and keeping the first two columns in \eqref{eq:bigmatrix} is  
\[
\BFK(\CalX'\setminus\{x_1,x_i\}, \{x_1,x_2\})^\intercal = 
\kbordermatrix{ 
& x_2 & \cdots & x_{i-1} & x_{i+1} & \cdots & x_{j-1} & x_{j+1} &\cdots & x_N \\
x_1 & p_1q_2 & \cdots & p_1q_{i-1} & p_1q_{i+1} & \cdots & p_1q_{j-1} & p_1q_{j+1} & \cdots & p_1q_N\\
x_2 & p_2q_2 & \cdots & p_2q_{i-1} & p_2q_{i+1} & \cdots & p_2q_{j-1} & p_2q_{j+1} & \cdots & p_2q_N
},
\]
whose rows are linear dependent, obviously. Hence, if we apply the Laplace expansion to \eqref{eq:bigmatrix} along the first row, then only the first two terms in the expansion are nonzero. This is because the minors in the other terms all involve two linearly dependent columns, thereby being zero. Hence,  
\begin{equation}\label{eq:expansion2}
\begin{aligned}
& |\BFK(\CalX'\setminus\{x_i\}, \CalX'\setminus\{x_{i-1}\})| \\ =&   p_1q_1 |\BFK(\CalX'\setminus\{x_i,x_1\},\CalX'\setminus\{x_{i-1},x_1\})| - p_1q_2|\BFK(\CalX'\setminus\{x_i,x_1\},\CalX'\setminus\{x_{i-1},x_2\})|
\end{aligned}
\end{equation}

We next consider two cases, $p_2\neq 0$ and $p_2=0$, separately. 

\textbf{Case 1 ($p_2\neq 0$).} Notice that $\BFK(\CalX'\setminus\{x_i,x_1\},\CalX'\setminus\{x_{i-1},x_1\})$ and $\BFK(\CalX'\setminus\{x_i,x_1\},\CalX'\setminus\{x_{i-1},x_2\})$ differ by only their first columns, and that the first column of the latter is a multiple of that of the former. In particular,
\begin{equation}\label{eq:multiple}
|\BFK(\CalX'\setminus\{x_i,x_1\},\CalX'\setminus\{x_{i-1},x_2\})| = \frac{p_1}{p_2} |\BFK(\CalX'\setminus\{x_i,x_1\},\CalX'\setminus\{x_{i-1},x_1\})|,
\end{equation}
and thus \eqref{eq:expansion2} becomes, for $\ell=i-1$,
\begin{equation}\label{eq:expansion3}
\begin{aligned}
|\BFK(\CalX'\setminus\{x_i\}, \CalX'\setminus\{x_\ell\})|  =    \left(p_1q_1 - \frac{p_1^2q_2}{p_2} \right) |\BFK(\CalX'\setminus\{x_i,x_1\},\CalX'\setminus\{x_\ell,x_1\})|.
\end{aligned}
\end{equation}

One can check easily that \eqref{eq:expansion3} holds for  $\ell=i,i+1$ as well. Then, the left-hand-side of \eqref{eq:minor2} becomes 
\begin{align}
& \sum_{\ell=i-1}^{i+1} (-1)^{\ell} k(x_j,x_\ell) |\BFK(\CalX'\setminus\{x_i\}, \CalX'\setminus\{x_\ell\})| \nonumber \\
=&  
 \left(p_1q_1 - \frac{p_1^2q_2}{p_2} \right) \sum_{\ell=i-1}^{i+1}(-1)^\ell k(x_j,x_\ell) 
 |\BFK(\CalX'\setminus\{x_i,x_1\},\CalX'\setminus\{x_\ell,x_1\})|. \label{eq:sum} 
\end{align}

Let $\CalX''=\CalX'\setminus\{x_1\}$. Then, for the summation in \eqref{eq:sum}, 
\[\sum_{\ell=i-1}^{i+1}(-1)^\ell k(x_j,x_\ell) 
 |\BFK(\CalX'\setminus\{x_i,x_1\},\CalX'\setminus\{x_\ell,x_1\})| = \sum_{\ell=i-1}^{i+1}(-1)^\ell k(x_j,x_\ell) 
 |\BFK(\CalX''\setminus\{x_i\},\CalX''\setminus\{x_\ell\})|,\]
which equals the $(i,j)$ minor of $\BFK(\CalX'', \CalX'')$ multiplied by $(-1)^{j-1}$, following the argument leading to \eqref{eq:minor3}. But the $(i,j)$ minor of $\BFK(\CalX'', \CalX'')$ is 0 by the induction assumption, since $j\geq i+2$. Therefore, \eqref{eq:sum} equals 0, which proves \eqref{eq:minor2}.

\textbf{Case 2 ($p_2= 0$).} It is easy to see that $\BFK(\CalX'\setminus\{x_i,x_1\},\CalX'\setminus\{x_{i-1},x_1\})$ is singular, since its first column is all zeros. Hence, \eqref{eq:expansion2} becomes
\[
|\BFK(\CalX'\setminus\{x_i\}, \CalX'\setminus\{x_{i-1}\})|  = - p_1q_2|\BFK(\CalX'\setminus\{x_i,x_1\},\CalX'\setminus\{x_{i-1},x_2\})|.
\]
Since the first row of $\BFK(\CalX'\setminus\{x_i,x_1\},\CalX'\setminus\{x_{i-1},x_2\})$ is now $(p_1q_2,0,\ldots,0)$, we apply the Laplace expansion to this row to obtain
\begin{equation}\label{eq:expansion4}
|\BFK(\CalX'\setminus\{x_i\}, \CalX'\setminus\{x_\ell\})| = - p_1^2q_2^2|\BFK(\CalX'\setminus\{x_i,x_1,x_2\},\CalX'\setminus\{x_\ell,x_2,x_1\})|,
\end{equation}
for $\ell=i-1$. Likewise, we can show that  \eqref{eq:expansion4} holds for  $\ell=i,i+1$ as well. Then, the left-hand-side of \eqref{eq:minor2} becomes, letting $\CalX'''=\CalX\setminus\{x_1,x_2\}$, 
\[\sum_{\ell=i-1}^{i+1} (-1)^{\ell} k(x_j,x_\ell) |\BFK(\CalX'\setminus\{x_i\}, \CalX'\setminus\{x_\ell\})|
= -p_1^2q_2^2\sum_{\ell=i-1}^{i+1} (-1)^{\ell} k(x_j,x_\ell) |\BFK(\CalX'''\setminus\{x_i\}, \CalX'''\setminus\{x_\ell\})|.
\]
Then, we can prove \eqref{eq:minor2} using the same argument as the last paragraph of Case 1.   
\end{proof}

Provided that $\BFK$ is nonsingular,  not only can we show that $\BFK^{-1}$ is symmetric and tridiagonal, but also we can calculate the nonzero entries of $\BFK^{-1}$ analytically. The fact that $\BFK^{-1}$ is analytically invertible makes $\BFK$ highly computationally tractable. Before  presenting the analytical expressions of the nonzero entries of $\BFK^{-1}$, we first calculate the determinant of $\BFK$. 

\begin{proposition}\label{prop:determinant}
For $n\geq 2$, 
\begin{equation}\label{eq:G_det}
|\BFK(\CalX,\CalX)|=p_1q_n\prod_{i=2}^n(p_iq_{i-1}-p_{i-1}q_i).
\end{equation}
\end{proposition}
\begin{proof} 
We prove \eqref{eq:G_det} by induction on $n$. The base case $n=2$ is straightforward: 
\[|\BFK(\CalX,\CalX)| =     
\begin{vmatrix}
p_1q_1 & p_1q_2 \\
p_1q_2 & p_2q_2
\end{vmatrix}
= p_1q_1p_2q_2- p_1^2q_2^2 = p_1q_2(p_2q_1-p_1q_2).
\]

Now we suppose that \eqref{eq:G_det} holds for any $n\leq N-1$. Then, for $n=N$, applying the Laplace expansion to the first row of $\BFK(\CalX,\CalX)$,
\begin{align}
|\BFK(\CalX,\CalX)|&=\sum_{\ell=1}^N(-1)^{1+\ell} p_1q_\ell |\BFK(\CalX\setminus\{x_1\},\CalX\setminus\{x_\ell\})| \nonumber \\
&= p_1q_1|\BFK(\CalX\setminus\{x_1\},\CalX\setminus\{x_1\})|-p_1q_2|\BFK(\CalX\setminus\{x_1\},\CalX\setminus\{x_2\})|, \label{eq:expansion5}
\end{align}
where the second equality follows from \eqref{eq:minor}. From the induction assumption, 
\begin{equation}\label{eq:G_det_induction}
|\BFK(\CalX\setminus\{x_1\},\CalX\setminus\{x_1\})|=p_2q_N\prod_{i=3}^N(p_iq_{i-1}-p_{i-1}q_i).
\end{equation}

Notice that 
\begin{equation*}
\BFK(\CalX\setminus\{x_1\},\CalX\setminus\{x_1\})=
\kbordermatrix{
 & x_2 & x_3 & \cdots & x_N \\
x_2 & p_2q_2 & p_2q_3 & \dots & p_2q_N\\
x_3 & p_2q_3 & p_3q_3 & \dots & p_3q_N\\
\vdots& & & \vdots & \\
x_N & p_2q_N & p_3q_N & \dots & p_Nq_N\\
},
\end{equation*}
and 
\begin{equation*}
\BFK(\CalX\setminus\{x_1\},\CalX\setminus\{x_2\})=
\kbordermatrix{
 & x_1 & x_3 & \cdots & x_N \\
x_2 & p_1q_2 & p_2q_3 & \dots & p_2q_N\\
x_3 & p_1q_3 & p_3q_3 & \dots & p_3q_N\\
\vdots& & & \vdots & \\
x_N & p_1q_N & p_3q_N & \dots & p_Nq_N\\
}.
\end{equation*}
Clearly, the above two matrices differ by only their first columns,  and the first column of one matrix is a multiple of the other. Hence, if $p_2\neq 0$, then $|\BFK(\CalX\setminus\{x_1\},\CalX\setminus\{x_2\})| = \frac{p_1}{p_2}|\BFK(\CalX\setminus\{x_1\},\CalX\setminus\{x_1\})|$. Thus, by \eqref{eq:expansion5} and \eqref{eq:G_det_induction},
\begin{align*}
|\BFK(\CalX,\CalX)|=&\left(p_1q_1-\frac{p_1^2q_2}{p_2}\right) |\BFK(\CalX\setminus\{x_1\},\CalX\setminus\{x_1\})| \\
=& \left(p_1q_1-\frac{p_1^2q_2}{p_2}\right) p_2q_N\prod_{i=3}^N(p_iq_{i-1}-p_{i-1}q_i)  
= p_1q_N \prod_{i=2}^N(p_iq_{i-1}-p_{i-1}q_i).
\end{align*}

On the other hand, if $p_2=0$, then by \eqref{eq:expansion5} and \eqref{eq:G_det_induction},
\begin{align*}
|\BFK(\CalX,\CalX)|= & -p_1q_2|\BFK(\CalX\setminus\{x_1\},\CalX\setminus\{x_2\})|\\ 
=& -p_1q_2 \cdot p_1q_2|\BFK(\CalX\setminus\{x_1,x_2\},\CalX\setminus\{x_2,x_1\})|\\ 
=& -p_1^2q_2^2  p_3q_N\prod_{i=4}^N(p_iq_{i-1}-p_{i-1}q_i),
\end{align*}
where the second equality follows from the Laplace expansion along the first row of $\BFK(\CalX\setminus\{x_1\},\CalX\setminus\{x_2\})$, whereas the last equality from the induction assumption. At last, notice that with $p_2=0$,
\begin{align*}
p_1q_N\prod_{i=2}^N(p_iq_{i-1}-p_{i-1}q_i) =& p_1q_N (p_2q_1-p_1q_2)(p_3q_2-p_2q_3)\prod_{i=4}^N(p_iq_{i-1}-p_{i-1}q_i)\\
=& -p_1^2q_2^2  p_3q_N\prod_{i=4}^N(p_iq_{i-1}-p_{i-1}q_i).
\end{align*}
Therefore, (\ref{eq:G_det}) holds for $n=N$.  \hfill$\Box$ 
\end{proof}

By using the Laplace expansion and mathematical induction in a similar fashion, we can also prove the following result but defer the proof to Appendix \ref{app:A}.
\begin{proposition}\label{prop:minor}
For $n\geq 2$ and $2\leq i\leq n$,
\[|\BFK(\CalX\setminus\{x_{i-1}\}, \CalX\setminus\{x_i\})| =
p_1q_n\prod_{j=2,j\neq i}^n (p_j q_{j-1}-p_{j-1}q_j).
 \]
\end{proposition}

With Propositions \ref{prop:determinant} and \ref{prop:minor}, the nonzero entries of $\BFK^{-1}$ can be readily calculated. 

\begin{theorem}\label{theo:inverse}
For $n\geq 3$, if $\BFK$ is nonsingular, then the nonzero entries of $\BFK^{-1}$ are given as follows,
\[(\BFK^{-1})_{i,i} = 
\left\{
\begin{array}{ll}
\displaystyle\frac{p_2}{p_1(p_2q_1-p_1q_2)},&  \quad \mbox{if }i=1,\\[2.5ex]
\displaystyle\frac{p_{i+1}q_{i-1}-p_{i-1}q_{i+1}}{(p_iq_{i-1}-p_{i-1}q_i)(p_{i+1}q_i-p_iq_{i+1})},&  \quad \mbox{if }2\leq i\leq n-1,\\[2.5ex]
\displaystyle\frac{q_{n-1}}{q_n(p_nq_{n-1}-p_{n-1}q_n)},&  \quad \mbox{if }i=n,
\end{array}
\right.
\]
and 
\[(\BFK^{-1})_{i-1,i} = (\BFK^{-1})_{i,i-1} = \frac{-1}{p_iq_{i-1}-p_{i-1}q_i}, \quad i=2,\ldots,n. \]
\end{theorem}

\begin{proof} 
It follows from the identity \eqref{eq:adjugate} that 
\[
(\BFK^{-1})_{i,i} =  \frac{1}{|\BFK|}|\BFK(\CalX\setminus\{x_i\},\CalX\setminus\{x_i\})|\quad\mbox{and}\quad (\BFK^{-1})_{i-1,i} =  \frac{-1}{|\BFK|}|\BFK(\CalX\setminus\{x_{i-1}\},\CalX\setminus\{x_i\})|.
\]
The results can then be shown by a straightforward calculation using Propositions \ref{prop:determinant} and \ref{prop:minor}. \hfill$\Box$ \end{proof}

\begin{remark}
There are two significant implications of Theorems \ref{theo:tridiag} and \ref{theo:inverse}. First,   $\BFK^{-1}$ can be computed in $\CalO(n)$ time, since it is tridiagonal, having only $3n-2$ nonzero entries. Second, the numerical stability regarding the computation of $\BFK^{-1}$ is improved substantially, since its nonzero entries have simple analytical expressions and numerical algorithms for matrix inversion are no longer needed. 
\end{remark}

\subsection{Positive Definiteness}

Theorem \ref{theo:tridiag} characterizes the essential structure of the covariance function of Gaussian processes with a 1-dimensional domain that yields sparse precision matrices. However, in order that a function of the form \eqref{eq:kernel} is a covariance function, the matrix $\BFK(\CalX,\CalX)$ must be positive semidefinite for any $\CalX=\{x_1,\ldots,x_n\}$. We further require $\BFK(\CalX,\CalX)$ to be positive definite so that it is invertible. The following conditions on $p$ and $q$ that constitute the function \eqref{eq:kernel} turn out to be both sufficient and necessary for the positive definiteness of $\BFK(\CalX,\CalX)$, provided that $p$ and $q$ are continuous. 

\begin{assumption}\label{assump:MCF}
Let $(L,U)$ be an open interval in $\Real$, where $L$ and $U$ are allowed to be $-\infty$ and $\infty$, respectively. For all $x,y\in(L,U)$, 
\begin{enumerate}[label=(\roman*)]
\item
$p(x)q(y)-p(y)q(x)<0$ if $x<y$, and
\item 
$p(x)q(y)>0$.
\end{enumerate}
\end{assumption}

\begin{remark}
It is straightforward to check that the covariance functions in Examples \ref{example:BM}--\ref{example:OU} all satisfy Assumption \ref{assump:MCF}. 
\end{remark}

\begin{theorem}\label{theo:PD}
Suppose that $p:(L,U)\mapsto\Real$ and $q:(L,U)\mapsto\Real$ are both continuous. Then, $\BFK(\CalX,\CalX)$ is positive definite for any $\CalX\subset (L,U)$ with $|\CalX|=n\geq 2$ if and only if $p$ and $q$ satisfy Assumption \ref{assump:MCF}. 

\end{theorem}

\begin{proof}
We first prove the ``if'' part. Fix an arbitrary $\CalX=\{x_1,\ldots,x_n\}$ with $x_1<\cdots<x_n$.  The symmetry of $\BFK(\CalX,\CalX)$ is obvious. Then, the first leading principal minor of of $\BFK(\CalX,\CalX)$ is $p_1q_1=p(x_1)q(x_1)>0$. Moreover, for any $\ell=2,\ldots,n$, it follows from Proposition \ref{prop:determinant} that  the $\ell^{\mathrm{th}}$ leading principal minor of $\BFK(\CalX,\CalX)$ is
\begin{align*}
|\BFK(\{x_1,\ldots,x_\ell\},\{x_1,\ldots,x_\ell\})| = & p_1q_\ell\prod_{i=2}^\ell (p_iq_{i-1}-p_{i-1}q_i) \\
=& p(x_1)q(x_\ell) \prod_{i=2}^\ell [p(x_i)q(x_{i-1})-p(x_{i-1})q(x_i)] > 0.
\end{align*}
Hence, $\BFK(\CalX,\CalX)$ is positive definite by Sylvester's criterion. 

Now, we suppose that $\BFK(\CalX,\CalX)$ is positive definite for any $\CalX$, and prove the ``only if'' part by contradiction. Specifically, we show that if condition (i) or (ii) is false, then we can construct a matrix $\BFK(\CalX,\CalX)$ that violates Sylvester's criterion. 

Assume that condition (i) is false, i.e., there exists $r<t$ for which $p(r)q(t)-p(t)q(r)\geq 0$. If $p(r)q(t)-p(t)q(r)= 0$, or if $p(r)q(t)-p(t)q(r)> 0$ and $p(r)q(t)\geq 0$, then 
\[|\BFK(\{r,t\},\{r,t\})|=p(r)q(t)[p(t)q(r)-p(r)q(t)] \leq 0.\]

If $p(r)q(t)-p(t)q(r)> 0$ and $p(r)q(t)>0$, then we show that $h(s)\coloneqq p(r)q(s)-p(s)q(r)>0$ for any $s\in(r,t)$. To see this, notice that $h(r)=0$ and $h(t)>0$. It then follows from the continuity of $h(s)$ that $h(s)>0$, since $h(s)$ would has a zero $s_0\in(r,t)$ otherwise, which would imply that $|\BFK(\{r,s_0)\},\{r,s_0)\}|=0$. Likewise, we can show that $p(s)q(t)-p(t)q(s)>0$ for any $s\in(r,t)$. Hence, 
\[|\BFK(\{r,s,t\},\{r,s,t\})|=p(r)q(t)[p(s)q(r)-p(r)q(s)][p(t)q(s)-p(s)q(t)]< 0.\]
Thus, we conclude that condition (i) must be true. 

Assume that condition (ii) is false, i.e., there exist $r$ and $s$ such that $p(r)q(s)\leq 0$. If $r=s$, then for any $t>s$, the first leading principal minor of $\BFK(\{r,t\},\{r,t\})$ is $p(r)q(r)\leq 0$. If $r\neq s$,  assuming $r<s$ without loss of generality, then $p(s)q(r)-p(r)q(s)>0$ since we have shown condition (i) must be true, and thus
\[|\BFK(\{r,s\},\{r,s\})| = p(r)q(s)[p(s)q(r)-p(r)q(s)] \leq 0,\]
which completes the proof. 
\hfill$\Box$  \end{proof}

Through Theorems \ref{theo:tridiag}--\ref{theo:PD}, we have effectively characterized a class of computationally tractable covariance functions for Gaussian processes with a 1-dimensional domain.  We call covariance functions of the form \eqref{eq:kernel} that satisfy Assumption \ref{assump:MCF} \emph{(1-dimensional) Markovian covariance functions} (MCFs). 

\begin{remark}
MCFs establish an explicit connection between Gaussian processes and GMRFs. Let $\SFM(x)$ be a Gaussian process equipped with an MCF. Then, for any $\CalX=\{x_1,\ldots,x_n\}$, $\{\SFM(x):x\in\CalX\}$ forms a GMRF. Assuming that $x_1<\cdots<x_n$, the neighborhood structure of this GMRF is defined as follows: $x_i$ and $x_j$ are neighbors if and only if $|i-j|=1$, which is implied by the tridiagonal structure of the precision matrix $\BFSigma_\SFM^{-1}$. 
\end{remark}

\begin{corollary}\label{cor:change_var}
Let $\CalT:(L,U)\mapsto \Real$ be a strictly increasing function and $\CalT^{-1}$ denotes is inverse. If $k(x,y)$ is an MCF for $x,y\in(L,U)$, then $k(\CalT(x),\CalT(y))$ is an MCF for $x,y\in(\CalT^{-1}(L), \CalT^{-1}(U))$.
\end{corollary}
\begin{proof} 
Suppose that $k(x,y)=p(x)q(y)\ind_{\{x\leq y\}}+p(y)q(x)\ind_{\{x>y\}}$ with $p(x)$ and $q(x)$ satisfying Assumption \ref{assump:MCF}. Then, 
\begin{align*}
k(h(x), h(y)) = & p(h(x))q(h(y))\ind_{\{h(x)\leq h(y)\}}+ p(h(y))q(h(x))\ind_{\{h(x)>h(y)\}}\\
= &\tilde p(x)\tilde q(y)\ind_{\{x\leq y\}}+\tilde p(y)\tilde q(x)\ind_{\{x>y\}}
\end{align*}
where $\tilde p(x)=p(h(x))$ and $\tilde q(x)=q(h(x))$. Here, the second equality follows from the strict increasing monotonicity of $h$. Moreover,  it is easy to see that $\tilde p(x)$ and $\tilde q(x)$ satisfy Assumption \ref{assump:MCF}. \hfill$\Box$ \end{proof}

We will provide in \S\ref{sec:SL} a convenient approach  to constructing MCFs based on ordinary differential equations (ODEs), provided that the ODE involved is analytically tractable. Corollary \ref{cor:change_var} provides an additional tool to construct new MCFs by modifying known ones. 

\subsection{Multidimensional Extension}\label{sec:multidim}


So far, we have been focusing on Gaussian processes with a 1-dimensional domain. Unfortunately, there is no multidimensional analog to the S-L theory that we can take advantage of. We circumvent this difficulty by defining a $D$-dimensional MCF in the following ``composite'' manner: $k(\BFx,\BFy) = \prod_{i=1}^D k_i(x^{(i)},y^{(i)})$, where $\BFx=(x^{(1)},\ldots,x^{(D)})$, $\BFy=(y^{(1)},\ldots,y^{(D)})$, and  $k_i(\cdot,\cdot)$ is a 1-dimensional MCF defined along the $i^{\mathrm{th}}$ dimension, $i=1,\ldots,D$. We remark that these 1-dimensional MCFs do not need to be the same and can be chosen to capture different correlation behaviors in each dimension.

The composite structure preserves the sparsity of the precision matrix, but it comes at the cost of restriction in selecting the design points $\CalX=\{\BFx_1,\ldots,\BFx_n\} $. In particular, we assume that $\CalX$ forms a regular lattice, that is, it can be expressed as a Cartesian product. But the coordinates along each dimensional do not need to be equally spaced.

\begin{assumption}\label{assump:lattice}
$\CalX= \bigtimes_{i=1}^D \{x^{(i)}_1,x^{(i)}_2,\ldots,x^{(i)}_{n_i}\}$ and $n=\prod_{i=1}^Dn_i$, where $n_i$ is the number of points along the $i^{\mathrm{th}}$ dimension and $x^{(i)}_1<x^{(i)}_2<\ldots<x^{(i)}_{n_i}$,  $i=1,\ldots,D$.
\end{assumption}

It follows that the covariance matrix associated with $k(\cdot,\cdot)$, the $D$-dimensional MCF, can be written as $\BFK = \bigotimes_{i=1}^D \BFK_i$, where $\BFK_i$ is the covariance matrix corresponding to $k_i(\cdot,\cdot)$ and $\{x^{(i)}_1,\ldots,x^{(i)}_{n_i}\}$,  and $\bigotimes$ denotes the tensor product of matrices. We refer to \citet[Chapter 13]{Laub05} for introduction of basic properties of tensor product. Then, the precision matrix can also be written as a tensor product: $\BFK^{-1} =\bigotimes_{i=1}^D \BFK_i^{-1} $. Hence, $\BFK^{-1}$ is also a sparse matrix since each $\BFK_i^{-1}$ is a tridiagonal matrix. The reduction in computational complexity suggested by \eqref{eq:woodbury} remains valid.

\section{Green's Function}\label{sec:SL}

The conditions in Assumption \ref{assump:MCF} can be trivially met by choosing a positive, strictly increasing function $p(x)$ and setting $q(x)\equiv 1$. The covariance function of a Brownian motion in Example \ref{example:BM} is indeed the case. However, this would mean that $k(x,y)=p(\min(x,y))$ is independent of $x$ for any $x>y$, which is not a reasonable feature in general. Despite the formal simplicity of the conditions in Assumption \ref{assump:MCF}, it is not immediately clear how to construct a wide spectrum of nontrivial functions $p(x)$ and $q(x)$ in a convenient way. We develop in this section a flexible, principled approach to constructing 1-dimensional MCFs. The key is to recognize that the function form \eqref{eq:kernel} resembles the Green's function of a Sturm-Liouville (S-L) differential equation. Since all second-order linear ODEs can be recast in the form of an S-L equation, the number of Green's functions that can be calculated analytically is potentially large; see \citet[Chapter 2.1]{ODEHandbook}.  

The relation between Green's functions and covariances was also identified in \cite{DolphWoodbury52}. There are three critical differences between their work and ours. First, they work on higher-order Markov processes \cite[Appendix B]{RasmussenWilliams06} whereas we focus on the Markovian processes in the conventional sense, which is of order one. Second, this kind of generality instead restricts their analysis to the setting where  the boundary condition of the S-L equation involved is imposed at infinity; further, their result which is similar to ours (Theorem \ref{theo:Green_MCF})  holds only for the case that the S-L equation has constant coefficients, which corresponds to the stationary O-U process. By contrast, in our analysis the boundary condition can be defined either on a finite interval or at infinity, and the coefficients of the S-L equation can be variable. Third, as a result of the last difference, the covariance functions constructed in their work are stationary, whereas our approach permits nonstationary covariance functions. In particular, we will construct an MCF that is nonstationary and even more computationally tractable than $k_{\mathrm{OU}}$, which is a stationary MCF; see the discussion in \S\ref{sec:MLE}. However, we do not discuss the nonstationarity from a modeling perspective in the present paper but refer interested readers to  \cite{Sampson10}. 

\subsection{Sturm-Liouville Equation}

Consider the following S-L equation defined on a finite interval $[L, U]$, 
\begin{equation}\label{eq:SL}
\mathscr{L}f(x)\coloneqq \frac{1}{w(x)}\left[\frac{\dif }{\dif x}\left(-u(x)\frac{\dif f(x)}{\dif x}\right) + v(x)f(x)\right] = 0, 
\end{equation}
with the boundary condition (BC)
\begin{equation}\label{eq:boundarycondition}
\left\{
\begin{aligned}
&\alpha_Lf(L)+\beta_Lf'(L)=0, \\
&\alpha_Uf(U)+\beta_Uf'(U)=0,
\end{aligned}
\right.
\end{equation}
where for some functions $\{u(x), v(x), w(x)\}$ and some constants $\{\alpha_L,\beta_L,\alpha_U,\beta_U\}$. We will consider three common BCs as follows.
\begin{itemize}
\item 
Dirichlet BC: $\alpha_L=\alpha_U=1$ and $\beta_L=\beta_U=0$, i.e., $f(L)=f(U)=0$;
\item 
Cauchy BC: $\alpha_L=\beta_U=1$ and $\alpha_U=\beta_L=0$, i.e.,  $f(L)=f'(U)=0$;
\item 
Neumann BC: $\beta_L=\beta_U=1$ and $\alpha_L=\alpha_U=0$, i.e., $f'(L)=f'(U)=0$.
\end{itemize}

The Green's function $g(x,y)$ of the above S-L equation is the solution to $\mathscr{L}g(x,y)=\delta(x-y)$ with the same BC, where $\delta(\cdot)$ is the Dirac delta function. It is a classical result in S-L theory that the Green's function has the following form
\begin{equation}\label{eq:Green_MCF}
g(x,y) = Cf_1(x)f_2(y)\ind_{\{x\leq y\}} + Cf_1(y)f_2(x)\ind_{\{x>y\}},
\end{equation}
where $f_1$ and $f_2$ satisfy
\begin{equation}\label{eq:Green_components}
\left\{
\begin{array}{l}
\mathscr{L}f_1(x)=0,\mbox{ }x\in[L,U] \\
\alpha_Lf(L)+\beta_Lf'(L)=0
\end{array}
\right.
\quad
\mbox{ and } 
\quad
\left\{
\begin{array}{l}
\mathscr{L}f_2(x)=0,\mbox{ }x\in[L,U] \\
\alpha_Uf(U)+\beta_Uf'(U)=0
\end{array}
\right. 
.
\end{equation}
Here, the constant $C$ is determined in such a way that 
\[\lim_{\epsilon\downarrow0} \bigg[\frac{\dif g(x,y)}{\dif x}\Big|_{x=y+\epsilon} - \frac{\dif g(x,y)}{\dif x}\Big|_{x=y-\epsilon}\bigg] = \frac{-1}{u(y)};\]
see \citet[Chapter 5.4]{Teschl12}. Consequently, the Green's function $g(x,y)$ has exactly the form \eqref{eq:kernel}. 

Clearly, not every S-L equation has a Green's function that satisfies Assumption \ref{assump:MCF}. Proper conditions need to be imposed on the functions $\{u(x), v(x), w(x)\}$ in the S-L equation \eqref{eq:SL} as well as on the BC \eqref{eq:boundarycondition}, in order that the Green's function be positive definite. 

\subsection{A General Result}

We show now that the Green's functions associated with a wide class of S-L equations are indeed MCFs. We assume that the S-L equation \eqref{eq:SL} is \emph{regular}, i.e., $u(x)$ is continuously differentiable, $v(x)$ and $w(x)$ are continuous, and $u(x)>0$ and $w(x)>0$ for $x\in[L,U]$; see \citet[Chapter 5.3]{Teschl12}. This is because the Green's function of a regular S-L equation enjoys an eigen-decomposition, which implies that the Green's function is positive semidefinite if the eigenvalues of the differential operator $\mathscr{L}$ are all positive.

\begin{theorem}\label{theo:Green_MCF}
Suppose that the S-L equation \eqref{eq:SL} is regular with $v(x)>0$ for $x\in[L,U]$ and the Dirichlet BC. Then, its Green's function is an MCF. 
\end{theorem}

\begin{proof}

Fix a set of distinct points $\CalX=\{x_1,\ldots,x_n\}\subset (L,U)$. Let $\BFG(\CalX,\CalX)$ denote the matrix whose entry $(i,j)$ is $g(x_i,x_j)$. Given the fact that the Green's function has the form \eqref{eq:Green_MCF}, by Theorems \ref{theo:tridiag} and \ref{theo:PD} it suffices to show that $\BFG(\CalX,\CalX)$  is positive definite.

Consider the eigenvalue problem associated with the S-L equation \eqref{eq:SL} (i.e., the so-called S-L problem): $\mathscr{L}\phi(x) = \lambda \phi(x)$, with $\phi(x)$ satisfying the BC \eqref{eq:boundarycondition}. It is well known in ODE theory that if the S-L equation is regular and satisfies the BC \eqref{eq:boundarycondition}, then the S-L problem has a countable number of eigenvalues $\{\lambda_\ell:\ell= 1,2,\ldots\}$, and the normalized eigenfunctions $\{\phi_\ell(x):\ell=1,2,\ldots\}$ can be chosen real-valued and form an orthonormal basis in the space of functions 
\[\mathsf{L}^2([L,U], w(x), \dif x)\coloneqq \left\{h:[L,U]\mapsto\Real\Big|\int_L^U h^2(x)w(x)\dif x<\infty \right\},\]
endowed with the inner product $\langle h_1, h_2  \rangle\coloneqq \int_L^U h_1(x)h_2(x)w(x)\dif x$. In particular, $\langle \phi_i,\phi_j\rangle$ equals 1 if $i=j$ and 0 otherwise. Moreover, the eigenvalues are all positive if $v(x)$ is positive on $[L, U]$ and the BC \eqref{eq:boundarycondition} is of the Dirichlet type. We refer to \citet[\S0.2.5]{ODEHandbook} for a discussion on the S-L problem and its properties. 

Then, the Green's function can be expressed as the following eigen-decomposition  
\[g(x,y) = \sum_{\ell=1}^\infty \lambda_\ell^{-1}\phi_\ell(x)\phi_\ell(y),\]
since $\lambda_\ell >0$ for each $\ell=1,2,\ldots$; see \citet[Chapter 10.1]{ArfkenWeberHarris12} for a proof. Notice that 
\[\int_L^U\int_L^Uh(x)h(y)g(x,y)w(x)w(y)\dif x\dif y = \sum_{\ell=1}^\infty \lambda_\ell^{-1}\langle h,\phi_\ell\rangle^2 \geq 0,\]
for any  $h\in \mathsf{L}^2([L,U], w(x), \dif x)$. Hence,  $g(x,y)$ is positive semidefinite, which implies that $\BFG(\CalX,\CalX)$ is positive semidefinite; see, e.g., \citet[Chapter 4.1]{RasmussenWilliams06}.

What remains is to prove $|\BFG(\CalX,\CalX)|\neq 0$. It follows from Sturm's comparison theorem \citep[Theorem 5.20]{Teschl12} that if $v(x)>0$, then any function that satisfies $\mathscr{L}f(x)=0$  has at most one zero in $[L, U]$. In particular, consider the functions $f_1$ and $f_2$ that constitute the Green's function in the expression \eqref{eq:Green_MCF}. Due to the Dirichlet BC, we know from \eqref{eq:Green_components} that $f_1(L)=f_2(U)=0$. Therefore, $f_1$ and $f_2$ have no other zeros in $(L,U)$, and thus 
\begin{equation}\label{eq:PD_condition_1}
f_1(x)f_2(y)\neq 0,\quad \mbox{ for all $x,y\in(L,U)$}.
\end{equation}

Next, we show by contradiction that 
\begin{equation}\label{eq:PD_condition_2}
f_1(x)f_2(y)-f_1(y)f_2(x)\neq 0,\quad \mbox{ for all $x,y\in(L,U)$ if $x>y$}. 
\end{equation}
Assume that \eqref{eq:PD_condition_2}  is false, i.e., there exist $s>t$ in $(L,U)$ such that $f_1(s)f_2(t)=f_1(t)f_2(s)$, or equivalently, $f_1(s)/f_2(s)=f_1(t)/f_2(t)$, since we have shown that $f_2(x)\neq 0$ for all $x\in(L,U)$. 

Notice that for any $c\neq 0$, if we replace $f_1(x)$ by $cf_1(x)$ and adjust the constant $C$ to $C/c$ in the expression \eqref{eq:Green_components}, then we retain the functional form of an MCF. Hence, we can assume, without loss of generality,  that $f_1$ is properly scaled so that $f_1(s)/f_2(s)=f_1(t)/f_2(t)=1$. This implies that $f_1(s)-f_2(s) = f_1(t)-f_2(t)=0$, i.e., $f_1(x)-f_2(x)$ has two zeros in $(L, U)$. However, since $f_1(x)-f_2(x)$ is a solution to $\mathscr{L}f(x)=0$, this contradicts the implication of Sturm's comparison theorem, namely, any solution to $\mathscr{L}f(x)=0$ has at most one zero in $[L, U]$ if $v(x)>0$ for $x\in[L,U]$.

At last, it follows from \eqref{eq:PD_condition_1}, \eqref{eq:PD_condition_2}, and Proposition \ref{prop:determinant} that $|\BFG(\CalX,\CalX)|\neq 0$.  \hfill$\Box$
\end{proof}

\begin{remark}
It can be seen from the proof of Theorem \ref{theo:Green_MCF} that for a regular S-L equation, it suffices to assume $v(x)>0$ in order that its Green's function be a covariance function on the finite interval $[L,U]$. But the covariance matrix may be singular for BCs that are not of the Dirichlet type. Nevertheless, this does not mean that the Green's function cannot be a positive definite covariance function when $v(x)$ is not a positive function, or when other types of BCs are imposed. In general, if the Green's function of an S-L equation can be solved analytically in the form of \eqref{eq:Green_MCF}, then we can check whether it is an MCF by simply verifying verify Assumption \ref{assump:MCF}.

\end{remark}

\subsection{Some Examples}

We now use the Green's-function approach to construct several MCFs  which turn out to have excellent performance when applied in SK for predicting response surfaces in the numerical experiments in \S\ref{sec:numerical}.  

We assume that the domain of the S-L equation is $[L,U]=[0,1]$; otherwise, we use the change-of-variable technique to make it so. Consider the following ODE with constant coefficients
\begin{equation}\label{eq:Ising}
- f''(x)+\nu f(x)=0,
\end{equation}
by setting $u(x)\equiv 1$, $v(x)\equiv \nu$, and $w(x)\equiv 1$ in \eqref{eq:SL}. The Green's function has a different form, depending on the sign of $\nu$ and the BC. Theorem \ref{theo:Green_MCF} stipulates that the Green's function is an MCF if $\nu>0$ and the Dirichlet BC is imposed. For the other cases, we can easily verify that Assumption \ref{assump:MCF} is indeed satisfied if $\nu$ is above a (negative) threshold. Since it is a routine exercise to solve \eqref{eq:Ising} for the Green's function with a BC of the Dirichlet, Cauchy, or Neumann type, we omit the details and only present the results.

\begin{theorem}\label{theo:MCF_examples}
The Green's function of equation (\ref{eq:Ising}) is $g(x,y)=\eta^2 [p(x)q(y)\ind_{\{x\leq y\}} + p(y)q(x)\ind_{\{x>y\}}]$, where $\eta^2$, $p(x)$, and $q(x)$ are given in Table \ref{tab:Green}. Moreover, $g(x,y)$ is an MCF if any of the following three conditions is satisfied: (i) the Dirichlet BC is imposed and $\nu>-\pi^2$; (ii) the Cauchy BC is imposed and $\nu>-\frac{\pi^2}{4}$; (iii) the Neumann condition is imposed and $\nu>0$. 
\end{theorem}

\begin{table}[t]
\small
\begin{center}
\caption{The Green's Function of Equation \eqref{eq:Ising}.}  \label{tab:Green}
    \begin{tabular}{ccccc}
    \toprule 
    Boundary & $\nu$ & $\eta^2$ & $p(x)$ & $q(x)$ \\
    \midrule 
    Dirichlet &  $\displaystyle \nu\in(-\pi^2,0)$ & $\displaystyle \frac{1}{\gamma\sin(\gamma)}$ & $\sin(\gamma x)$ & $\sin(\gamma(1-x))$ \\
    \arrayrulecolor{white}
    \midrule
    Dirichlet & $\nu=0$ & $1$ & $x$ & $1-x$ \\
    \midrule 
    Dirichlet &  $\displaystyle \nu>0$ & $\displaystyle \frac{1}{\gamma\sinh(\gamma)}$ & $\sinh(\gamma x)$ & $\sinh(\gamma(1-x))$ \\
    \arrayrulecolor{black}
    \midrule
    Cauchy & $\displaystyle \nu\in(-\frac{\pi^2}{4},0)$ & $\displaystyle \frac{1}{\gamma\cos(\gamma)}$ & $\sin(\gamma x)$ & $\cos(\gamma(1-x))$  \\ 
    \arrayrulecolor{white}
    \midrule
    Cauchy & $\nu=0$ & $1$ & $x$ & $1$ \\
    \midrule 
    Cauchy &  $\displaystyle \nu>0$ & $\displaystyle \frac{1}{\gamma\cosh(\gamma)}$ & $\sinh(\gamma x)$ & $\cosh(\gamma(1-x))$ \\
    \arrayrulecolor{black}
    \midrule
    Neumann & $\displaystyle \nu>0$ & $\displaystyle \frac{1}{\gamma\sinh(\gamma)}$ & $\cosh(\gamma x)$ & $\cosh(\gamma(1-x))$ \\
    \bottomrule
    \end{tabular}
\end{center}
\small{\textit{Note.} {$\gamma=\sqrt{|\nu|}$.}}
\end{table}

It turns out that if the set of points $\CalX=\{x_1,\ldots,x_n\}$ are equally spaced, the  precision matrix associated with the MCFs in Theorem \ref{theo:MCF_examples}  has an even simpler structure than being symmetric tridiagonal. The proof relies on direct calculations suggested by Theorem \ref{theo:inverse} and is deferred to Appendix \ref{app:B}.

\begin{corollary}\label{cor:Dirichlet}
Let $g(x,y)=\eta^2[p(x)q(y)\ind_{\{x\leq y\}} + p(y)q(x)\ind_{\{x>y\}}]$, where $\eta^2>0$ is a free parameter, and $p(x)$ and $q(x)$ are the functions in Table \ref{tab:Green}.  Suppose that $\CalX=\{x_1,\ldots,x_n\}\subset (0,1)$, where $x_i=x_1+(i-1)h$ with $h=\frac{x_n-x_1}{n-1}$, $i=1,\ldots,n$. Then, $\BFG^{-1}(\CalX,\CalX)$ is a symmetric, tridiagonal matrix:
\begin{equation}\label{eq:simple_precision}
\BFG^{-1}(\CalX,\CalX) =
\eta^{-2}a
\begin{pmatrix}
b & -1 &   &  &  \\
-1 & c & -1 &     & \\
\cdots &  & \cdots &  &  \cdots\\
 & &   -1 & c & -1 \\
 & &  & -1 & d
\end{pmatrix},
\end{equation}
where the parameters $(a,b,c,d)$ are given in Table \ref{tab:inverse}.
\end{corollary}

\begin{table}[t]
\small
\begin{center}
\caption{Parameters in the Inverse Matrix \eqref{eq:simple_precision}.} \label{tab:inverse}
    \begin{tabular}{cccccc}
    \toprule 
    Boundary & $\nu$ & $a$ & $b$ & $c$ & $d$\\
    \midrule 
    Dirichlet & $\displaystyle \nu\in(-\pi^2,0)$ & $\displaystyle \frac{1}{\sin(\gamma) \sin(\gamma h)}$ & $\displaystyle  \frac{\sin(\gamma (x_1+h))}{\sin(\gamma x_1)}$ & $2\cos(\gamma h)$  & $\displaystyle \frac{\sin(\gamma(1-x_n+h))}{\sin(\gamma(1-x_n))}$    \\
    \arrayrulecolor{white}
    \midrule
    Dirichlet & $\nu=0$ & $\displaystyle \frac{1}{h}$ & $\displaystyle 1+\frac{h}{x_1}$ & $2$  & $\displaystyle 1+\frac{h}{1-x_n}$    \\
    \midrule
    Dirichlet & $\displaystyle \nu>0$ & $\displaystyle \frac{1}{\sinh(\gamma) \sinh(\gamma h)}$ & $\displaystyle  \frac{\sinh(\gamma (x_1+h))}{\sinh(\gamma x_1)}$ & $2\cosh(\gamma h)$  & $\displaystyle \frac{\sinh(\gamma(1-x_n+h))}{\sinh(\gamma(1-x_n))}$    \\
    \arrayrulecolor{black}
    \midrule 
    Cauchy & $\displaystyle \nu\in(-\frac{\pi^2}{4},0)$ & $\displaystyle \frac{1}{\sin(\gamma) \sin(\gamma h)}$ & $\displaystyle  \frac{\sin(\gamma (x_1+h))}{\sin(\gamma x_1)}$ & $2\cos(\gamma h)$  & $\displaystyle \frac{\cos(\gamma(1-x_n+h))}{\cos(\gamma(1-x_n))}$    \\
    \arrayrulecolor{white}
    \midrule
    Cauchy & $\nu=0$ & $\displaystyle \frac{1}{h}$ & $\displaystyle 1+\frac{h}{x_1}$ & $2$  & $\displaystyle 1$    \\
    \midrule
    Cauchy & $\displaystyle \nu>0$ & $\displaystyle \frac{1}{\sinh(\gamma) \sinh(\gamma h)}$ & $\displaystyle  \frac{\sinh(\gamma (x_1+h))}{\sinh(\gamma x_1)}$ & $2\cosh(\gamma h)$  & $\displaystyle \frac{\cosh(\gamma(1-x_n+h))}{\cosh(\gamma(1-x_n))}$    \\
    \arrayrulecolor{black}
    \midrule 
    Neumann & $\displaystyle \nu>0$ & $\displaystyle \frac{1}{\sinh(\gamma) \sinh(\gamma h)}$  & $\displaystyle  \frac{\cosh(\gamma (x_1+h))}{\cosh(\gamma x_1)}$ & $2\cosh(\gamma h)$  & $\displaystyle \frac{\cosh(\gamma(1-x_n+h))}{\cosh(\gamma(1-x_n))}$    \\  
    \bottomrule 
    \end{tabular}
\end{center}
\small{\textit{Note.} {$\gamma=\sqrt{|\nu|}$.}}
\end{table}

Corollary \ref{cor:Dirichlet} has two important implications from the computational perspective. First, by choosing a set of equally spaced design points, the precision matrix associated with the MCFs in Theorem \ref{theo:MCF_examples} can be computed in $O(1)$ time since its nonzero entries can be expressed in terms of only  4 quantities, regardless of the size of the matrix. This is a further reduction in complexity compared to computing the precision matrix of a general MCF, which amounts to $O(n)$. 

Second, the expression \eqref{eq:simple_precision} allows reparameterization of the MCFs in Theorem \ref{theo:MCF_examples}. Instead of estimating the parameters of an MCF, we can express the likelihood function in terms of the parameters in the precision matrix. Under mild conditions, the resulting MLE can be solved without any matrix inversion, thereby improving substantially the computational efficiency and numerical stability. We discuss this matter in details in \S\ref{sec:MLE}.

\begin{table}[t]
\small
\begin{center}
\caption{Computational Complexity.} \label{tab:complex}
    \begin{tabular}{cccc}
    \toprule 
    Covariance Function & $\BFSigma_\SFM^{-1}$  & $[\BFSigma_\SFM+\BFSigma_\varepsilon]^{-1}$ & SK Predictor + MSE  \\
    \midrule
    General & $\CalO(n^3)$ & $\CalO(n^3)$  & $\CalO(n^3)$\\
    MCF & $\CalO(n)$ & $\CalO(n^2)$ & $\CalO(n^2)$, or $\CalO(n)$ if $\BFSigma_\varepsilon=\BFzero$ \\ 
    CF in Table \ref{tab:Green} under Condition & $\CalO(1)$ & $\CalO(n^2)$ & $\CalO(n^2)$, or $\CalO(n)$ if $\BFSigma_\varepsilon=\BFzero$  \\
    \bottomrule
    \end{tabular}
\end{center}
\small{\textit{Note.} {Condition: design points are equally spaced. }}

\end{table}

In order to highlight the computational enhancement of MCFs relative to general covariance functions, we summarize the complexity for computing various quantities using different covariance functions in Table \ref{tab:complex}. First, for computing $\BFSigma_\SFM^{-1}$, MCFs reduce the complexity from $\CalO(n^3)$ to $\CalO(n)$ because of the sparsity of the inverse matrix and the analytical expression of its nonzero entries; the Green's function in Table \ref{tab:Green} further reduce the complexity to $\CalO(1)$ by taking advantage of the experiment design. Second, it can be seen that the existence of the simulation errors increases the computational complexity dramatically and offsets largely the benefit of MCFs. Third, once $[\BFSigma_\SFM+\BFSigma_\varepsilon]^{-1}$ is computed, the bulk of the computation of the SK predictor \eqref{eq:BLUP} and its MSE \eqref{eq:MSE} is to multiply the inverse matrix by a vector, which takes $\CalO(n^2)$ in general but is reduced to $\CalO(n)$ by the sparsity induced by MCFs.

\begin{remark}
The fact that entry $(i-1,i)$ of $\BFG^{-1}$ is independent of $i$ deserves an interpretation. Using the notations in Theorem \ref{theo:inverse}, this means that $p_iq_{i-1}-p_{i-1}q_i$ is a constant, which turns out to be related to the so-called Wronskian determinant $W(x)$ associated with the S-L equation. In particular, for two linearly independent solutions $p(x)$ and $q(x)$ to equation  \eqref{eq:SL}, the Wronskian is defined as 
\[W(x)= 
\begin{vmatrix}
p(x) & q(x) \\
p'(x) & q'(x)
\end{vmatrix}
=p(x)q'(x)-p'(x)q(x).
\]
On the other hand, if we fix $x_{i-1}$, then 
\[\frac{p_iq_{i-1}-p_{i-1}q_i}{h} =  \frac{p_i(q_{i-1}-q_i)-q_i(p_{i-1}-p_i)}{h} \to -p(x_{i-1})q'(x_{i-1})+q(x_{i-1})p'(x_{i-1})=-W(x_{i-1}),\]
as $h\downarrow 0$. Hence, $p_iq_{i-1}-p_{i-1}q_i$ can be viewed as a ``discretized'' Wronskian. It is known in the theory of S-L equations that $u(x)W(x)$ is a constant for $x\in[L,U]$. Since $\mu(x)\equiv \mu$ in equation \eqref{eq:Ising}, $W(x)$ is a constant. Nevertheless, we must emphasize that  in general, a constant Wronskian does not imply that $p_iq_{i-1}-p_{i-1}q_i$ is independent of $i$. 
\end{remark}

\subsection{Illustration}

A particularly important application of SK, besides response surface prediction, is to facilitate the exploration-exploitation trade-off during the random search for solving simulation optimization problems \citep{sun2014}. To that end, the uncertainty about the prediction, which is a result of the interplay between the extrinsic uncertainty imposed by SK to the unknown response surface and the intrinsic uncertainty from the simulation errors, should be characterized meaningfully. 

Given the fact that the squared exponential covariance function $k_{\mathrm{SE}}(x,y) = \eta^2 e^{-\theta(x-y)^2}$ is a standard choice in SK literature, we now compare MCFs with $k_{\mathrm{SE}}$ in terms of the performance in uncertainty quantification in stochastic simulation. Specifically, we consider two distinct MCFs: (i) the exponential covariance function $k_{\mathrm{Exp}}(x,y) = \eta^2 e^{-\theta|x-y|}$, which is the essentially same as the covariance function of the OU process in Example \ref{example:OU}; (ii) the Green's function associated with the Dirichlet BC in Theorem \ref{theo:MCF_examples}
\begin{equation}\label{eq:kernel_Diri}
k_{\mathrm{Dir}}(x,y) \coloneqq  \left\{
\begin{array}{ll}
\eta^2\left[\sin(\gamma x)\sin(\gamma(1-y))\ind_{\{x\leq y\}}+\sin(\gamma y)\sin(\gamma(1-x))\ind_{\{x> y\}}\right],&\mbox{if }  \nu <0, \\[1ex]
\eta^2\left[x(1-y)\ind_{\{x\leq y\}}+y(1-x)\ind_{\{x> y\}}\right],&\mbox{if }\nu =0, \\[1ex]
\eta^2\left[\sinh(\gamma x)\sinh(\gamma(1-y))\ind_{\{x\leq y\}}+\sinh(\gamma y)\sinh(\gamma(1-x))\ind_{\{x> y\}}\right], &\mbox{if } \nu>0, 
\end{array}
\right.
\end{equation}
for $x,y\in(0,1)$, where $\gamma=\sqrt{|\nu|}$.

We assume that a 1-dimensional continuous surface is observed with errors having variance $\sigma^2$. Given the observations, we first fit the SK metamodel equipped with each of the three covariance functions using MLE which is detailed in \S\ref{sec:MLE}, and then predict the surface using the SK predictor \eqref{eq:BLUP} with the parameter estimates. We also compute the standard deviation (S.D.) of the prediction, i.e., the square root of the prediction MSE \eqref{eq:MSE}, in order to measure the uncertainty about the predicted surface. We consider both $\sigma=0$ and $\sigma=0.1$. The results are shown in Figure \ref{fig:Uncertainty}.

\begin{figure}[t]
\begin{center}
\caption{Uncertainty Quantification of the SK Prediction.} \label{fig:Uncertainty}
$\begin{array}{cc}
\includegraphics[width=0.45\textwidth]{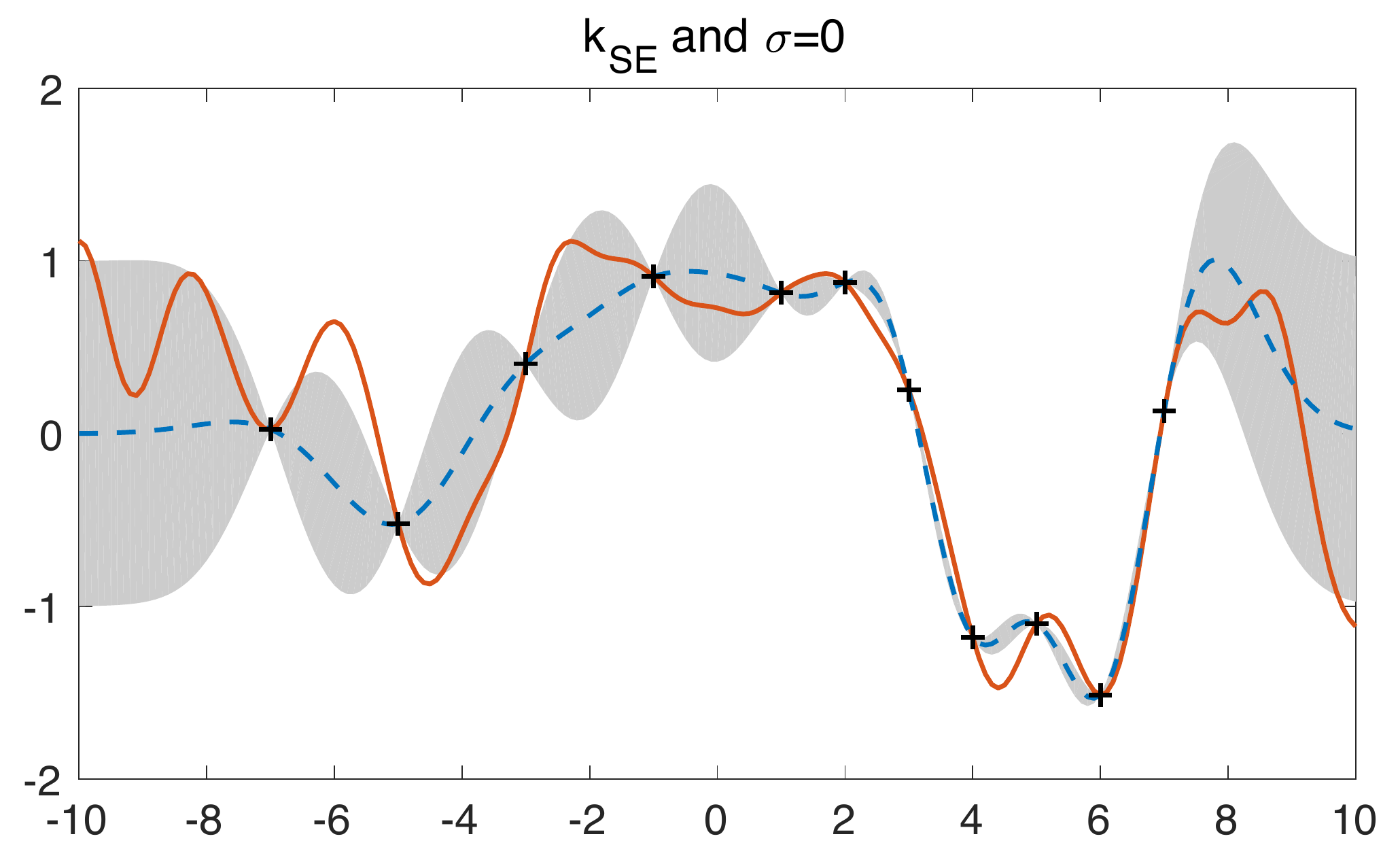}  &
\includegraphics[width=0.45\textwidth]{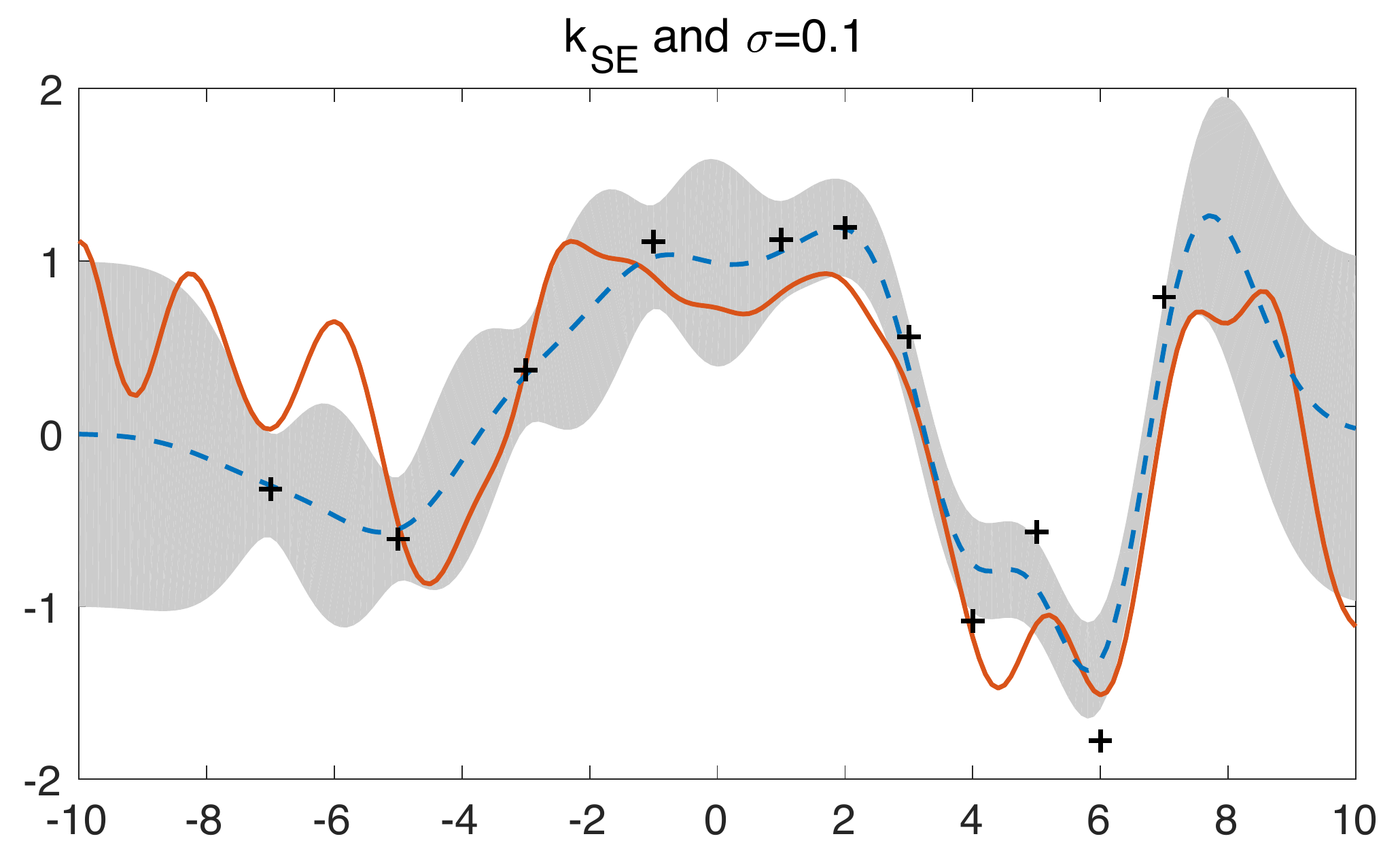} \\ 
\includegraphics[width=0.45\textwidth]{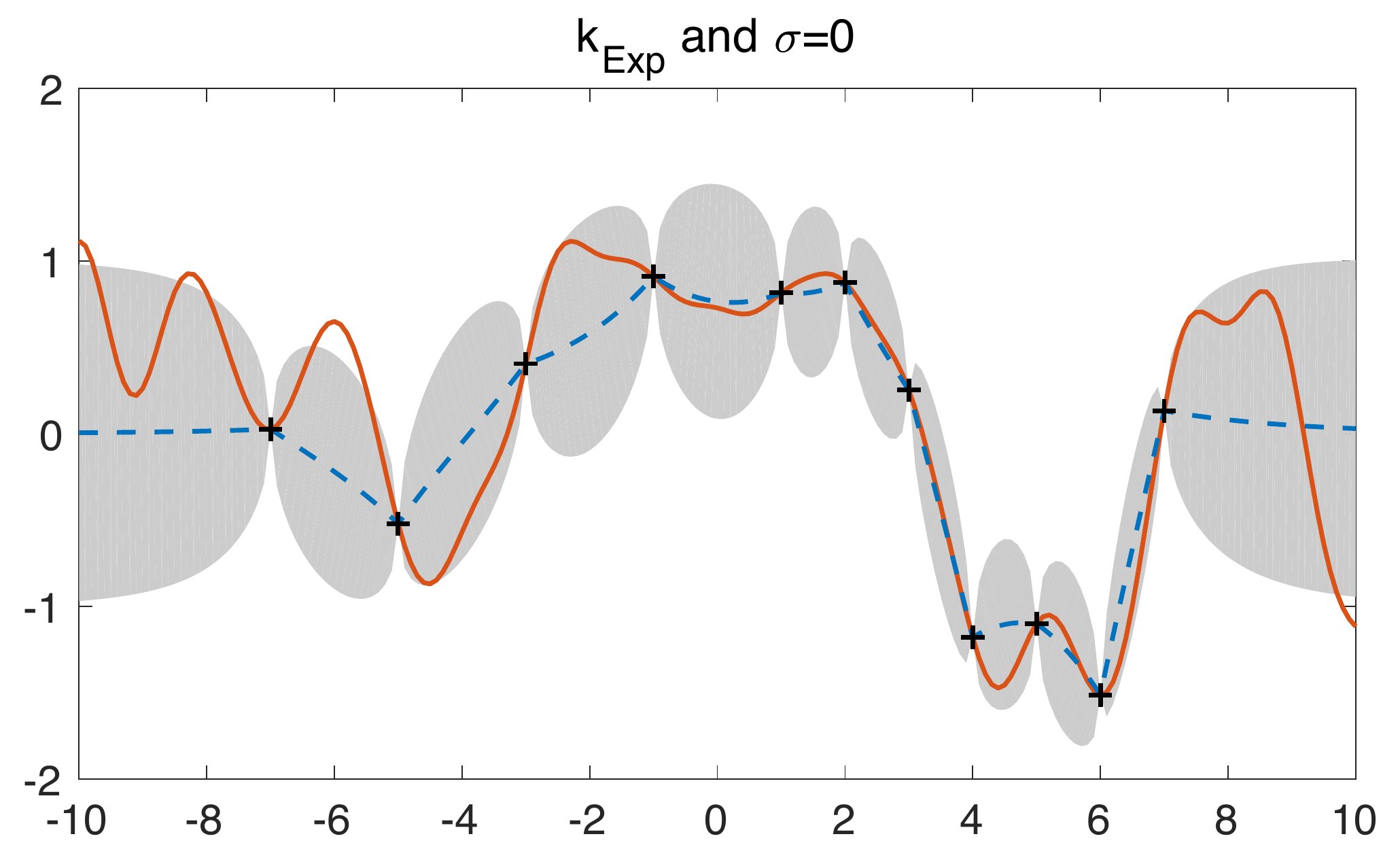} & 
\includegraphics[width=0.45\textwidth]{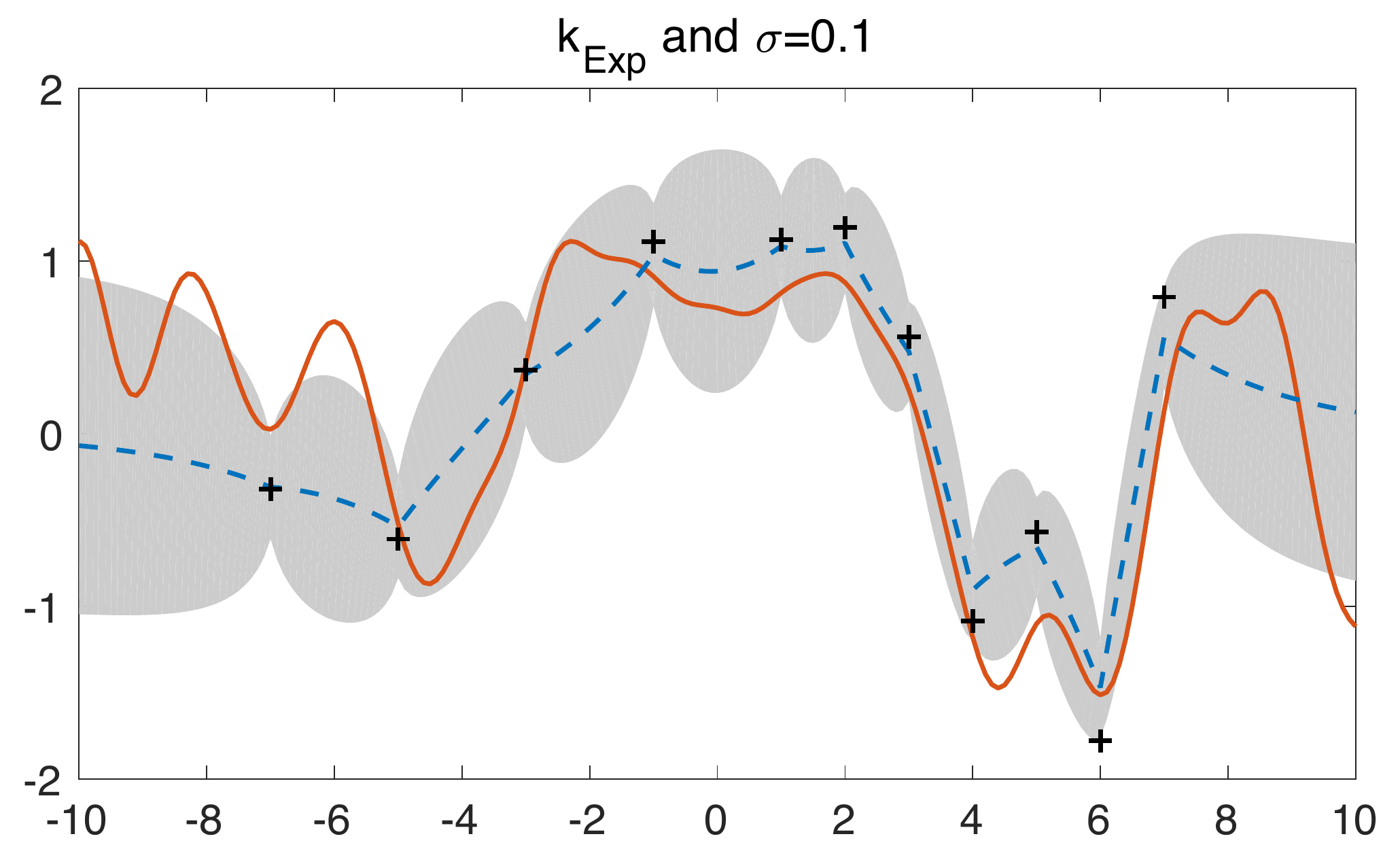} \\ 
\includegraphics[width=0.45\textwidth]{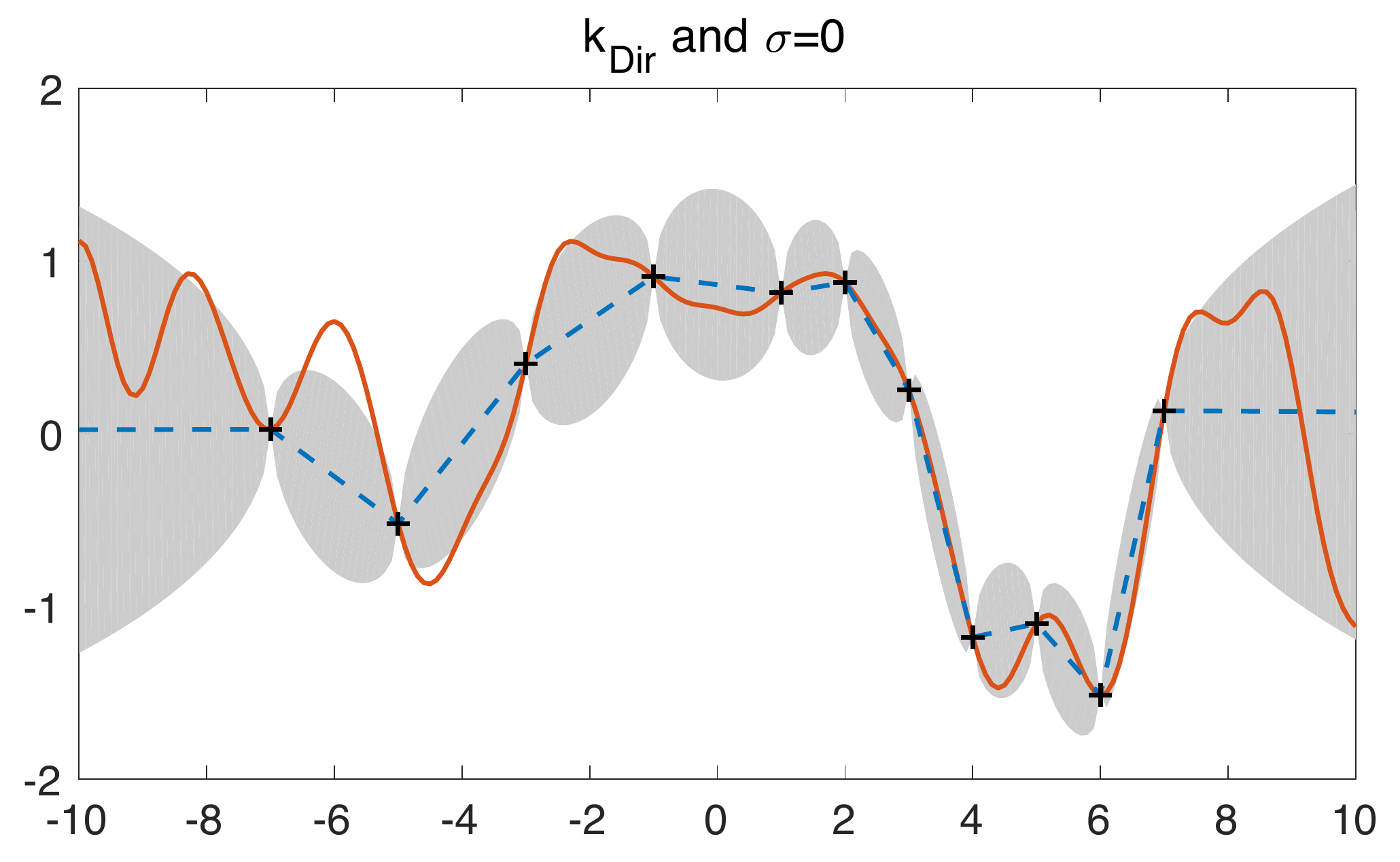}  &
\includegraphics[width=0.45\textwidth]{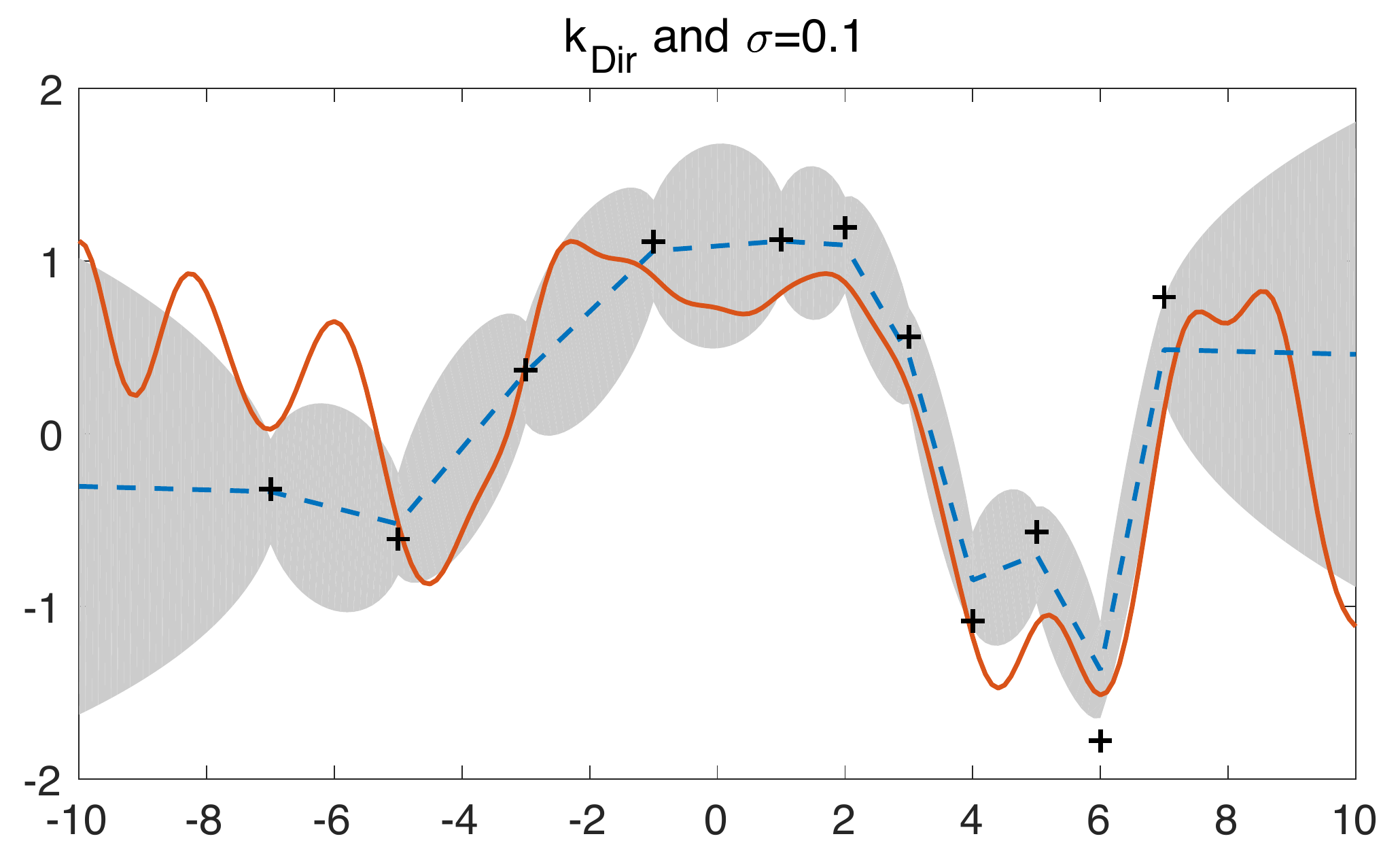}
\end{array}$
\end{center}
\small{\textit{Note.}  True surface (solid line), data ($+$), prediction (dashed line), $\pm$ standard deviation (shaded area).}
\end{figure}

Overall, all the three covariance functions can deliver meaningful uncertainty quantification of the unknown surface. For each covariance function, the 1-S.D. confidence band  can mostly cover the true surface, and it is inflated by the observation noise. Moreover, the confidence band is wider for regions with fewer observations (e.g., the interval $[-6,0]$) than regions with more (e.g., $[0,6]$), and it is particularly wide for extrapolation (e.g., $|x|\geq 8$). A main difference between $k_{\mathrm{SE}}$ and the two MCFs that is revealed in Figure \ref{fig:Uncertainty} is that both the predicted surface and the confidence band are smoother for the former. But the lack of smoothness in the predicted surface does not appear to cause significant issues as far as the prediction accuracy is concerned, which will be shown in the extensive numerical experiments in \S\ref{sec:numerical}.




\section{Parameter Estimation}\label{sec:MLE}

Let $\BFtheta$ denote the parameters used to specify the covariance function and $\BFK(\BFtheta)$ denote the covariance matrix. We now discuss the estimation of $\BFtheta$ and $\BFbeta$, the parameters that determine the trend of the response surface. We develop a highly efficient and numerically stable MLE scheme for a specific class of MCFs. We assume in this section that $\BFSigma_\varepsilon$, the variances of the simulation outputs, is known. This is a standard treatment regarding SK in simulation literature. In practice, $\BFSigma_\varepsilon$ is replaced by the sample variances $\widehat\BFSigma_\varepsilon$. 

\subsection{Numerically Stable MLE}\label{sec:stable_MLE}

Recall the log-likelihood function \eqref{eq:loglikelihood},
\begin{equation}\label{eq:log-likelihood2}
l(\BFbeta,\BFtheta) = -\frac{n}{2}\ln(2\pi) - \frac{1}{2}\ln|\BFK(\BFtheta) + \BFSigma_\varepsilon| -\frac{1}{2}(\bar\BFz-\BFF\BFbeta)^\intercal [\BFK(\BFtheta) + \BFSigma_\varepsilon]^{-1} (\bar\BFz-\BFF\BFbeta).
\end{equation}
The first order optimality conditions are derived using standard results of matrix calculus in \cite{AnkenmanNelsonStaum10},
\begin{equation}\label{eq:MLE_conditions1}
\left\{
\begin{aligned}
\BFzero  = & \frac{\partial l(\BFbeta,\BFtheta)}{\partial \BFbeta} = \BFF^\intercal \BFV^{-1}(\BFtheta)(\bar \BFz-\BFF\BFbeta), \\
\BFzero = & \frac{\partial l(\BFbeta,\BFtheta)}{\partial \BFtheta} = -\frac{1}{2}\mathrm{trace}\left[\BFV^{-1}(\BFtheta)\frac{\partial \BFV(\BFtheta)}{\partial \BFtheta}\right]+\frac{1}{2}(\bar\BFz-\BFF\BFbeta)^\intercal \left[\BFV^{-1}(\BFtheta)\frac{\partial \BFV(\BFtheta)}{\partial \BFtheta}\BFV^{-1}(\BFtheta)\right] (\bar\BFz-\BFF\BFbeta),
\end{aligned}
\right.
\end{equation}
where $\BFV(\BFtheta) = \BFK(\BFtheta) + \BFSigma_\varepsilon$. The Newton-Raphson algorithm or the Fisher scoring algorithm can be used to solve the above set of equations. 

It is a well-known issue \cite[Chapter 5.4]{fang2006} that $\BFK(\BFtheta)$ often becomes nearly singular when searching over the parameter space of $(\BFbeta,\BFtheta)$, causing serious numerical instability when numerically inverting $\BFV(\BFtheta)$. Admittedly, the presence of $\BFSigma_\varepsilon$ somewhat mitigates the issue, since it is a diagonal matrix whose  diagonal entries are all positive. But unless all the diagonal entries of $\BFSigma_\varepsilon$ are sufficiently large, which is not very likely when the number of design points is large, the numerical instability persists. 

Nevertheless, if $\BFK(\BFtheta)$ is constructed from an MCF, then $\BFK^{-1}(\BFtheta)$ is a sparse matrix having closed-form entries, thanks to Theorem \ref{theo:inverse} and Assumption \ref{assump:lattice}. Instead of using numerical methods such as Gaussian elimination to invert $\BFK(\BFtheta)$, we apply the Woodbury matrix identity \eqref{eq:woodbury},
\[\BFV^{-1}(\BFtheta) = \BFK^{-1}(\BFtheta) +  \BFK^{-1}(\BFtheta)[ \BFK^{-1}(\BFtheta) +  \BFSigma_\varepsilon^{-1}]^{-1} \BFK^{-1}(\BFtheta). 
\]
Hence, numerical inversion is only needed for computing   $[ \BFK^{-1}(\BFtheta) +  \BFSigma_\varepsilon^{-1}]^{-1}$. Notice that the diagonal entries of $\BFSigma_\varepsilon^{-1}$ are  $r_i/\Var[\varepsilon(\BFx_i)]$, $i=1,\ldots,n$, which can be made sufficiently far away from 0 by increasing $r_i$, the number of simulation replications at $\BFx_i$. Therefore, $\BFK^{-1}(\BFtheta) +  \BFSigma_\varepsilon^{-1}$ is not ill-conditioned in general. The numerical stability of MLE can be significantly improved. 

\subsection{Further Enhancement}\label{sec:enhanced_MLE}

If the covariance function $k_{\mathrm{Dir}}(x,y)$ is adopted, we can further improve the computational efficiency and numerical stability of MLE. For notational simplicity, we focus on the 1-dimensional case but the result can be extended to the $D$-dimensional case  without essential difficulty. 

Suppose that the design points are $\CalX=\{x_i=ih:i=1,\ldots,n\}$ with $h=1/(n+1)$. By Corollary \ref{cor:Dirichlet}, the precision matrix associated with $k_{\mathrm{Dir}}(x,y)$ and $\CalX$ is 
\begin{equation}\label{eq:Toeplitz}
\BFK^{-1}=\phi
\begin{pmatrix}
c & -1 &   &  &  \\
-1 & c & -1 &     & \\
\cdots &  & \cdots &  &  \cdots\\
 & &   -1 & c & -1 \\
 & &  & -1 & c
\end{pmatrix},
\end{equation}
where $\phi=\eta^{-2}a$ and
\begin{equation}\label{eq:bijection}
\left\{
\begin{array}{lll}
a = \sin^{-1}(\gamma) \sin^{-1}(\gamma h), & c = 2\cos(\gamma h), & \mbox{ if } \nu<0, \\[1ex]
a = h^{-1}, & c = 2,  & \mbox{ if }  \nu=0, \\[1ex]
a = \sinh^{-1}(\gamma) \sinh^{-1}(\gamma h), & c = 2\cosh(\gamma h), & \mbox{ if } \nu>0.
\end{array}
\right.
\end{equation}
Namely, all the diagonal entries of $\BFK^{-1}$ are made equal by the specific values of $x_1$ and $x_n$ in $\CalX$, and thus $\BFK^{-1}$ becomes a Toeplitz matrix. (However, this property does not hold for the Green's functions that correspond to the Cauchy or Neumann BC in Theorem \ref{theo:MCF_examples}.) 

A symmetric diagonal Toeplitz matrix enjoys a closed-form eigen-decomposition. Let $\{\lambda_i:i=1,\ldots,n\}$  be the eigenvalues of any matrix of the form \eqref{eq:Toeplitz} and $v_i^\intercal = (v_{i,1},\ldots,v_{i,n})$ be the eigenvector associated with $\lambda_i$, $i=1,\ldots,n$. Then, 
\begin{equation}\label{eq:eigenparis}
\lambda_i = \phi\left[c + 2\cos\left(i\pi(n+1)^{-1}\right)\right] \quad\mbox{and}\quad v_{i,j} = \sin\left(ij\pi(n+1)^{-1}\right);
\end{equation}
see   \cite{NoschesePasquiniReichel13}.

Notice that the mapping $(\eta^2, \nu)\mapsto(\phi,c)$ is bijective. Hence, we can reparameterize the MCF \eqref{eq:kernel_Diri} with $(\phi,c)$. Notice also that the eigenvector $v_i$ is independent of $(\phi,c)$, $i=1,\ldots,n$. Let  $\BFP$ be the matrix whose $i^{\mathrm{th}}$ row is $v_i^\intercal$. Then, $\BFP^{-1}=\BFP^\intercal$, since $\BFK^{-1}$ is positive definite. Let $\BFLambda(\phi,c)$ be the diagonal matrix whose $i^{\mathrm{th}}$ diagonal entry is $\lambda_i=\lambda_i(\phi,c)$. Then, $\BFK^{-1}(\phi,c)=\BFP^\intercal \BFLambda(\phi,c)\BFP$.

We now assume that $\BFSigma_\varepsilon$ has equal diagonal entries, i.e., $\Var[\varepsilon(x_i)]/r_i= \delta$, $i=1,\ldots,n$. This appears a reasonable assumption if (i) the simulation outputs have equal variances, i.e., $\Var[\varepsilon(x)]$ is a constant for all $x$, and the simulation budget is equally allocated, i.e., $r_1=\cdots=r_n$; or (ii) $r_i$ is chosen to be roughly proportional to $\Var[\varepsilon(x_i)]$. Under this assumption, $\BFK(\phi,c)+\BFSigma_\varepsilon = \BFP^\intercal [\BFLambda^{-1}(\phi,c) + \delta \BFI]\BFP$, where $\BFI$ denotes the identity matrix. Hence, $|\BFK(\phi,c)+\BFSigma_\varepsilon| = \prod_{i=1}^n (\lambda_i^{-1}(\phi,c)+\delta) $ and 
\[[\BFK(\phi,c)+\BFSigma_\varepsilon]^{-1} = \BFP^\intercal \mathrm{Diag}\left(\frac{1}{\lambda_1^{-1}(\phi,c)+\delta},\ldots,\frac{1}{\lambda_n^{-1}(\phi,c)+\delta}\right)\BFP \coloneqq \BFP^\intercal \BFD(\phi,c)\BFP,\]
where $\BFD(\phi,c)$ is diagonal whose the $i^{\mathrm{th}}$ diagonal entry is $d_i(\phi,c)=1/(\lambda_i^{-1}(\phi,c)+\delta)$. It follows that the log-likelihood function \eqref{eq:log-likelihood2} can be rewritten as 
\[l(\BFbeta,\phi,c) = -\frac{n}{2}\ln(2\pi) + \frac{1}{2}\sum_{i=1}^n\ln(d_i(\phi,c))-\frac{1}{2}(\bar\BFz-\BFF\BFbeta)^\intercal\BFP^\intercal\BFD(\phi,c)\BFP(\bar\BFz-\BFF\BFbeta).
\]
The first order optimality conditions for maximizing $l(\BFbeta,\phi,c)$  are
\begin{equation}\label{eq:MLE_conditions2}
\left\{
\begin{aligned}
\BFzero = & \frac{\partial l(\BFbeta,\phi,c)}{\partial \BFbeta} = \BFF^\intercal\BFP^\intercal \BFD(\phi,c)\BFP(\bar \BFz-\BFF\BFbeta),\\
0 = &  \frac{\partial l(\BFbeta,\phi,c)}{\partial \theta}  = \frac{1}{2}\sum_{i=1}^nd_i^{-1}(\phi,c)\frac{\partial d_i(\phi,c)}{\partial \theta} -\frac{1}{2}(\bar\BFz-\BFF\BFbeta)^\intercal\BFP^\intercal\frac{\partial \BFD(\phi,c)}{\partial \theta}  \BFP(\bar\BFz-\BFF\BFbeta),\quad \theta=\phi, c.
\end{aligned}
\right.
\end{equation}
Notice that $\frac{\partial \BFD(\phi,c)}{\partial \theta}$ , $\theta=\phi,c$,  is diagonal and can be calculated easily given  \eqref{eq:eigenparis}. In particular, the conditions in \eqref{eq:MLE_conditions2} do not involve any matrix inversion, thereby representing a further enhancement of computational efficiency relative to the optimality conditions of MLE for general MCFs. At last, with the maximum likelihood estimates of $(\phi, c)$, we can use \eqref{eq:bijection} to compute the estimates of $(\eta^2, a)$.

We summarize the differences in the use of MLE between MCFs and general covariance functions  in Table \ref{tab:MLE}. It needs to be emphasized, however, that the two parametric families of MCFs in Table \ref{tab:Green} other than $k_{\mathrm{Dir}}$ do not yield the kind of numerical enhancement discussed in this section. This is because the inverse matrix induced by them does not have the Toeplitz structure by Corollary \ref{cor:Dirichlet}.

\begin{table}[t]
\small
\begin{center}
\caption{Comparison on MLE.}\label{tab:MLE}
    \begin{tabular}{cccc}
    \toprule 
    Covariance Function & Inversion Needed? & Optimality Conditions & Stability Enhanced?\\
    \midrule
    General &  Yes & Eq. \eqref{eq:MLE_conditions1} & No \\
    MCF & Yes &  Eq. \eqref{eq:MLE_conditions1} & Yes \\ 
    $k_{\mathrm{Dir}}$ under Conditions & No & Eq. \eqref{eq:MLE_conditions2} & Yes\\
    \bottomrule
    \end{tabular}
\end{center}
\small{\textit{Note.} {Conditions: (i) design points are equally spaced; (ii) $\BFSigma_\varepsilon=\sigma^2\BFI$.}}
\end{table}

\begin{remark}
For a $D$-dimensional MCF $k(\BFx,\BFy)=\prod_{i=1}^Dk_i(x^{(i)},y^{(i)})$, where $k_i(\cdot,\cdot)$ is of the form \eqref{eq:kernel_Diri}, the optimality conditions of MLE can be derived in a similar manner. The key is to use the fact that the eigenvalues (resp., eigenvectors) of the tensor product $\bigotimes_{i=1}^D\BFK_i$ can be expressed as the tensor product of the eigenvalues (resp., eigenvectors) of each $\BFK_i$; see \citet[Theorem 13.12]{Laub05}. 
\end{remark}

\begin{remark}
By applying  Corollary \ref{cor:change_var}, we can relax the requirement on the form of the MCF from \eqref{eq:kernel_Diri} to $k_{\mathrm{Dir}}(\CalT(x),\CalT(y))$ for some  strictly increasing function $\CalT$. However, we need to change the design points accordingly to $\{\CalT^{-1}(ih):i=1,\ldots,n\}$, where $\CalT^{-1}$ is the inverse function of $\CalT$.
\end{remark}

\section{Numerical Experiments}\label{sec:numerical}

The big $n$ problem of SK has two aspects -- computational inefficiency and numerical instability. We have shown rigorously that with use of MCFs, the computational time related to matrix inversion, which is the core of the computation of both MLE and the SK predictor \eqref{eq:BLUP}, can be reduced from $\CalO(n^3)$ to $\CalO(n^2)$; see Table \ref{tab:complex}. The numerical stability issue, on the other hand, is detrimental to the prediction accuracy of SK in a more subtle way. For instance, it may cause numerical optimization of the MLE to fail, returning erroneous estimates of the parameters and further producing unreasonable predictions. In this section, we demonstrate via extensive numerical experiments, with emphasis on the stability aspect, that MCFs represent an elegant solution to the big $n$ problem of SK.

We compare the following three covariance functions. 
\begin{itemize}
\item 
Squared exponential: $k_{\mathrm{SE}}(\BFx,\BFy)=\eta^2\exp\left(-\sum_{i=1}^D\theta_i(x_i-y_i)^2\right)$;
\item 
Exponential: $k_{\mathrm{Exp}}(\BFx,\BFy)=\eta^2\exp\left(-\sum_{i=1}^D\theta_i|x_i-y_i|\right)$;     
\item 
Multidimensional extension of $k_{\mathrm{Dir}}(x,y)$ with distinct parameters in each dimension.
\end{itemize}
As discussed in \S\ref{sec:stable_MLE}, $k_{\mathrm{Exp}}$ can benefit from the tractability of MCFs, making its MLE significantly more stable than $k_{\mathrm{SE}}$. Further, $k_{\mathrm{Dir}}$ enjoys the ``inverse-free'' MLE scheme in \S\ref{sec:enhanced_MLE}, and thus has the highest computational tractability among the three competing alternatives. The computing environment of the following numerical experiments is a desktop PC with an Intel(R) Core(TM) i7-4790 3.60GHz processor and 16 GB of RAM, running Windows 7 Enterprise. The codes are written in Matlab R2015a. In the sequel, we assume that the SK metamodel \eqref{eq:uni_nriging} has a constant trend, i.e., $\SFZ(\BFx) = \beta+ \SFM(\BFx)$.

\subsection{Two-Dimensional Response Surfaces} \label{sec:2d_surface}

Consider three distinct 2-dimensional response surfaces which are defined and illustrated in Table \ref{tab:art}  and Figure \ref{fig:art}, respectively. 

\begin{table}[t]
\small
\begin{center}
\caption{Two-Dimensional Response Surfaces.} \label{tab:art}
    \begin{tabular}{lll}
    \toprule 
    Function Name & Expression & Domain  \\
    \midrule
    Three-Hump Camel & $\SFZ(x,y)=2x^2-1.05x^4+\frac{x^6}{6}+xy+y^2$ & $x,y\in [-2,2]$   \\
    Matyas & $\SFZ(x,y)=0.26(x^2+y^2)-0.48xy$ & $x,y\in[-10, 10]$  \\ 
    Bohachevsky & $\SFZ(x,y)=x^2+2y^2-0.3\cos(3\pi x)-0.4\cos(4\pi y)+0.7$ & $x,y\in[-100,100]$  \\     
    \bottomrule
    \end{tabular}
\end{center}
\end{table}

\begin{figure}[t]
\caption{Response Surfaces of the Functions in Table \ref{tab:art}.} \label{fig:art}
\begin{center}
    \includegraphics[width=0.32\textwidth]{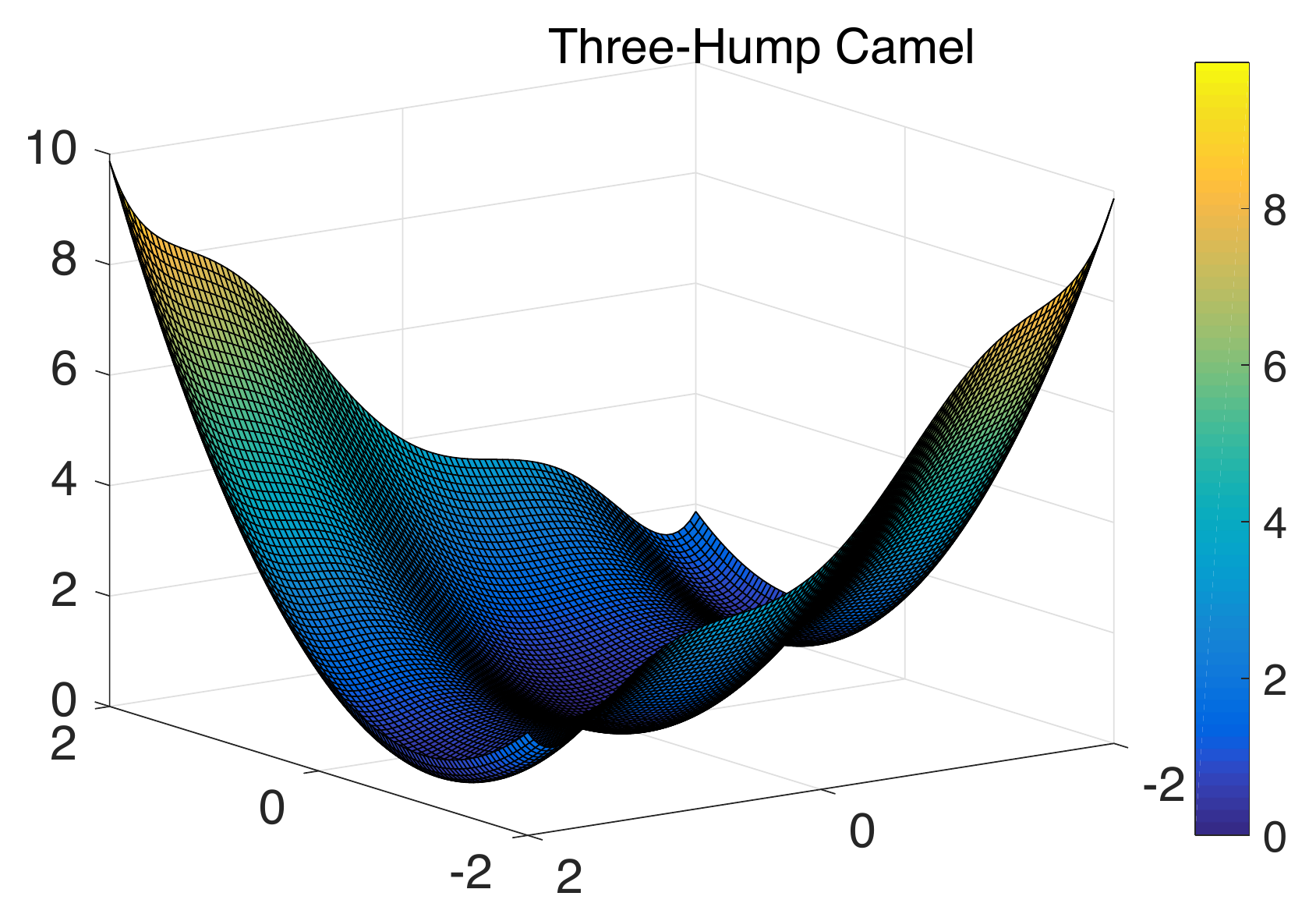} 
    \includegraphics[width=0.32\textwidth]{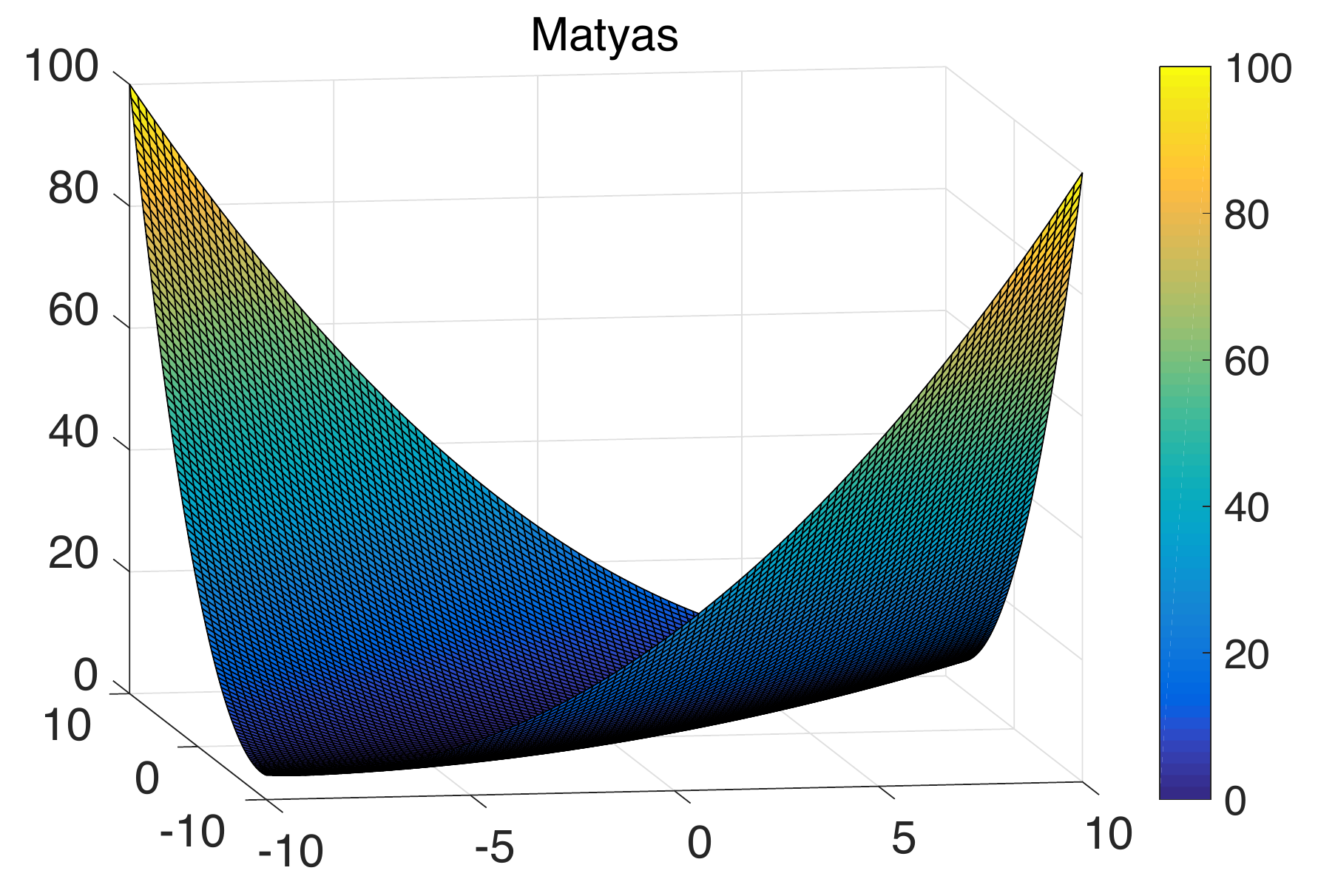} 
    \includegraphics[width=0.32\textwidth]{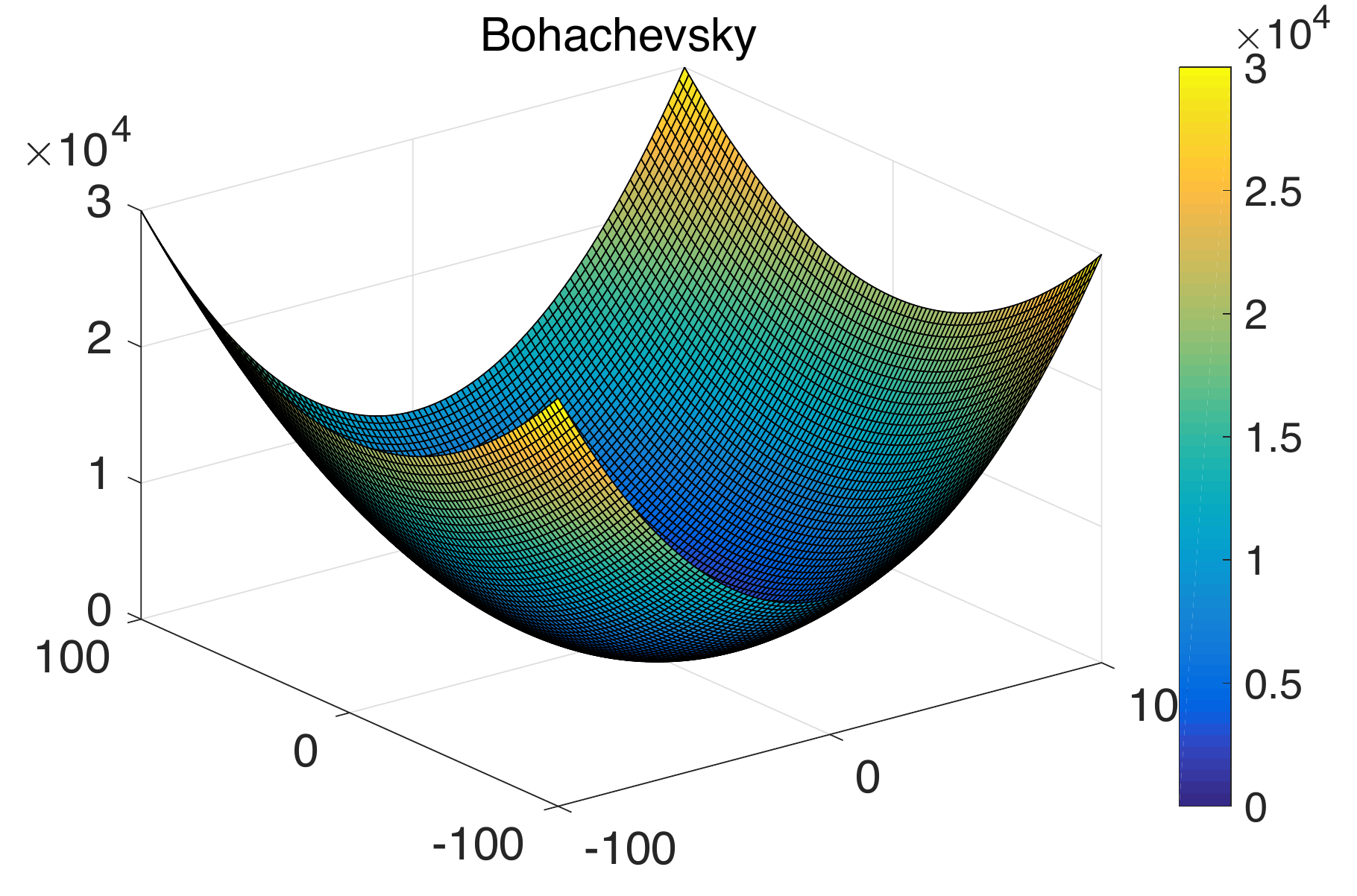}    
\end{center}

\end{figure}

For each surface, we choose $n=m^2$ design points and let them form an equally spaced lattice within the design space, for some integer $m\geq 3$. For instance, for the three-hump camel function whose domain is $[-2,2]^2$, we set the design points to be $\{(x_i,y_j)|x_i=\frac{4i}{m+1}-2, \frac{4j}{m+1}-2, i,j=1,\ldots,m\}$.  We set the number of prediction points  to be $K=100^2$ and place them equally spaced in the same way. For simplicity, we assume that the sampling variance  is $\sigma^2$ for each design point, implying that the covariance matrix of the sampling errors is $\BFSigma_\varepsilon=\sigma^2 \BFI$, where $\BFI$ denotes the $n\times n$ identity matrix. Given a covariance function (i.e., $k_{\mathrm{Dir}}$, $k_{\mathrm{Exp}}$,  or $k_{\mathrm{SE}}$), we first estimate the unknown parameters with MLE as discussed in \S\ref{sec:MLE}, and then compute the SK predictor $\hat\SFZ(\BFx_i)$ for each prediction point $\BFx_i$, $i=1,\ldots,K$ by plugging the parameter estimates into \eqref{eq:SK}. In order to assess the prediction accuracy, we compute the standardized root mean squared error (SRMSE) as follows 
\begin{equation*}\label{eq:SRMSE}
\mathrm{SRMSE} = \frac{\sqrt{\sum_{i=1}^K  \left[\SFZ(\BFx_i) - \hat\SFZ(\BFx_i)\right]^2 }}{\sqrt{\sum_{i=1}^K  \left[\SFZ(\BFx_i)-K^{-1}\sum_{h=1}^K\SFZ(\BFx_h)\right]^2}}.
\end{equation*}
since the three surfaces are of substantially different scales and the standardization facilitates the comparison. We repeat the experiment for both noiseless ($\sigma=0$) and noisy ($\sigma>0$) data, for each of the three surfaces, each  of the three covariance functions and $m=3,4,\ldots, 12$. The results are presented in Figure \ref{fig:Matyas}.

\begin{figure}[t]
\begin{center}
\caption{Accuracy for Predicting the Surfaces in Figure \ref{fig:art}.} \label{fig:Matyas}
$
\begin{array}{cc}
\includegraphics[width=0.45\textwidth]{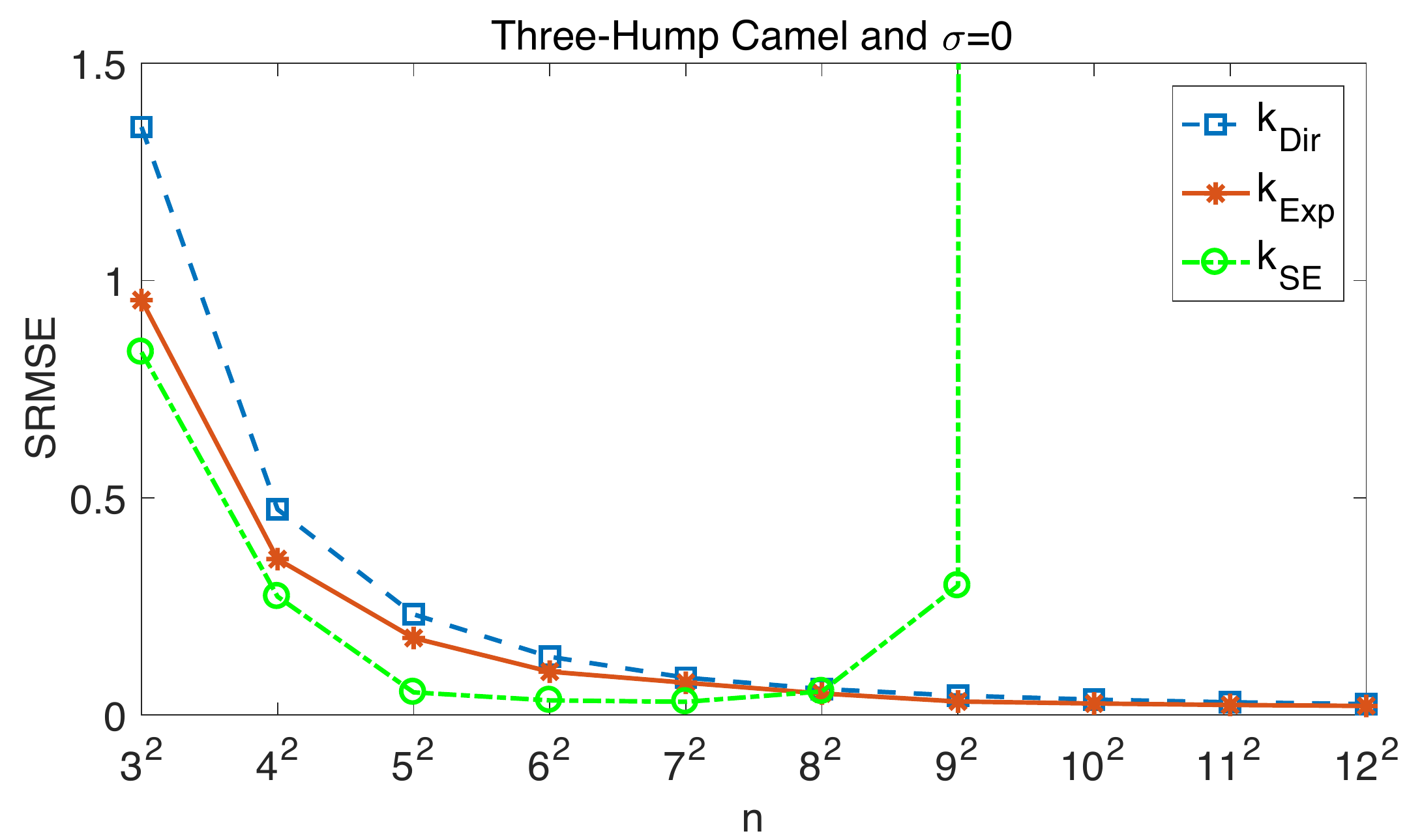} &
\includegraphics[width=0.45\textwidth]{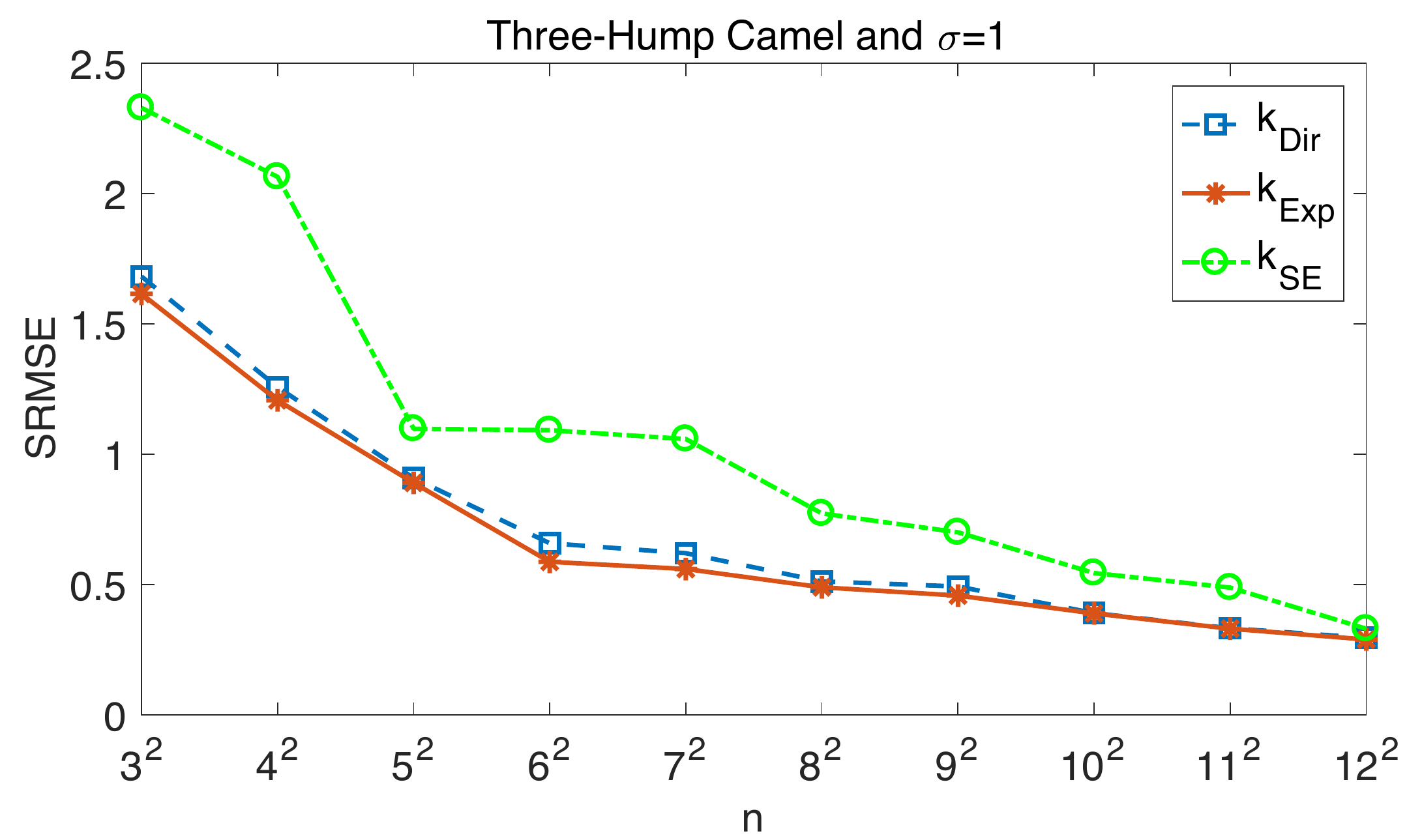} \\
\includegraphics[width=0.45\textwidth]{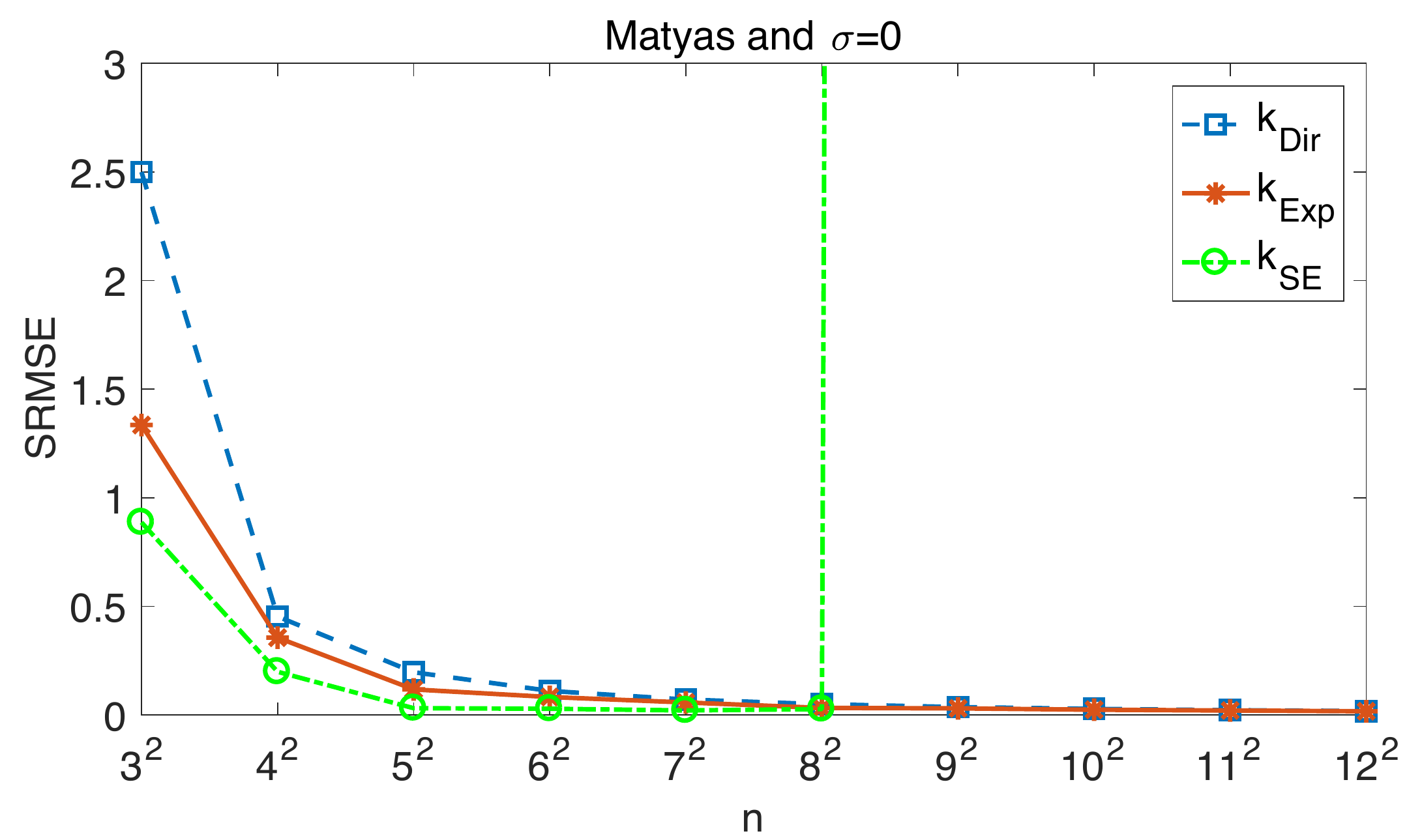} & 
\includegraphics[width=0.45\textwidth]{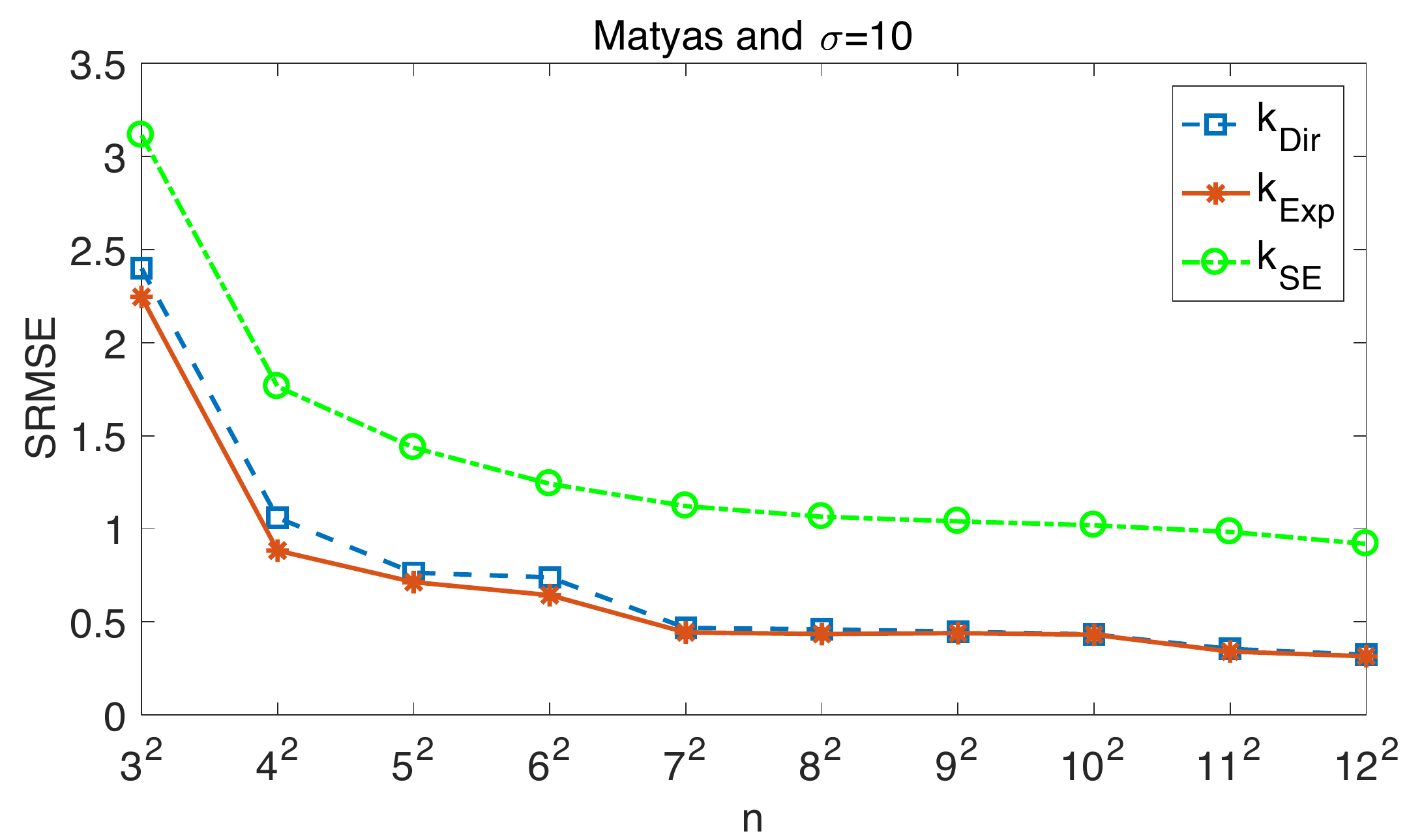} \\ 
\includegraphics[width=0.45\textwidth]{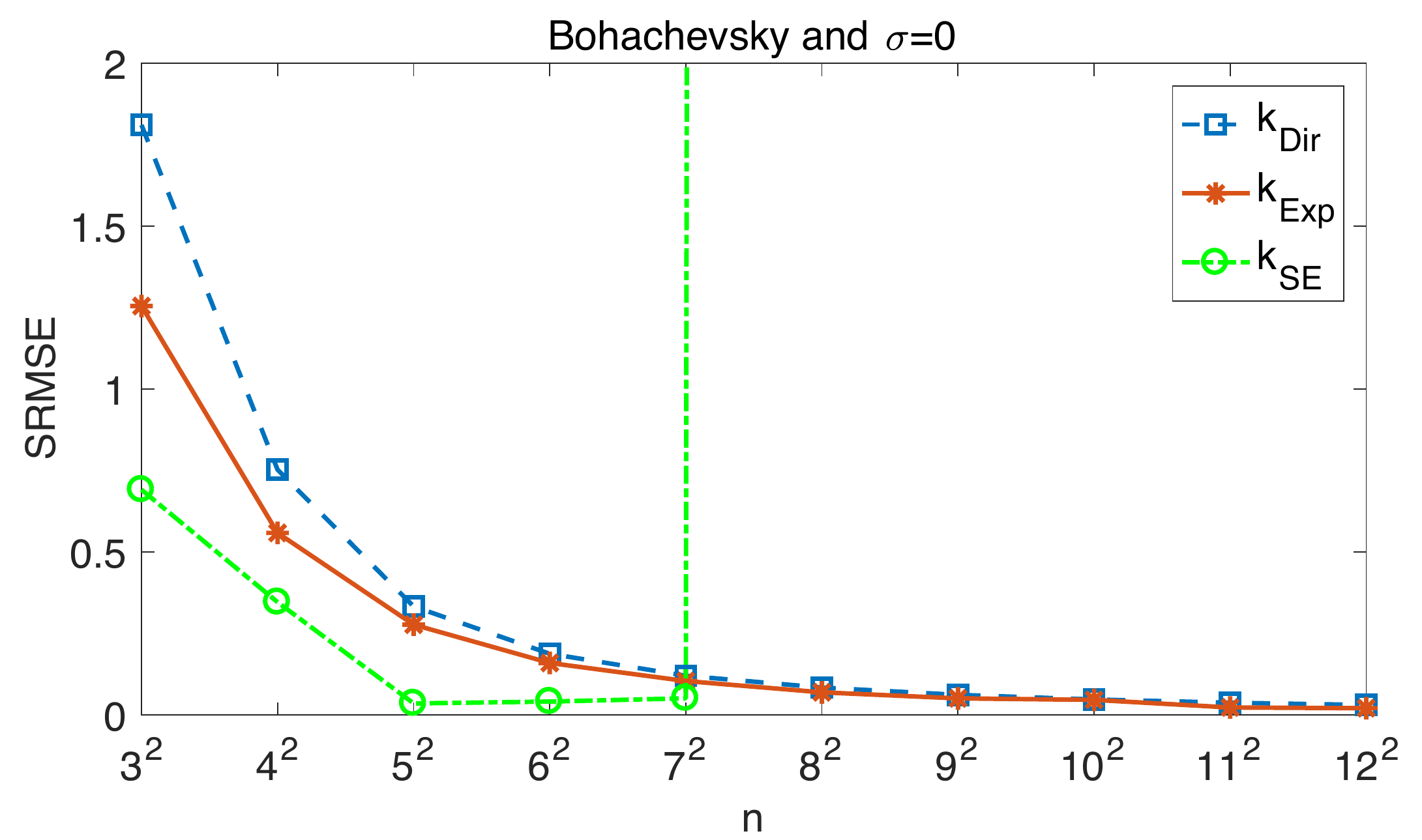} &
\includegraphics[width=0.45\textwidth]{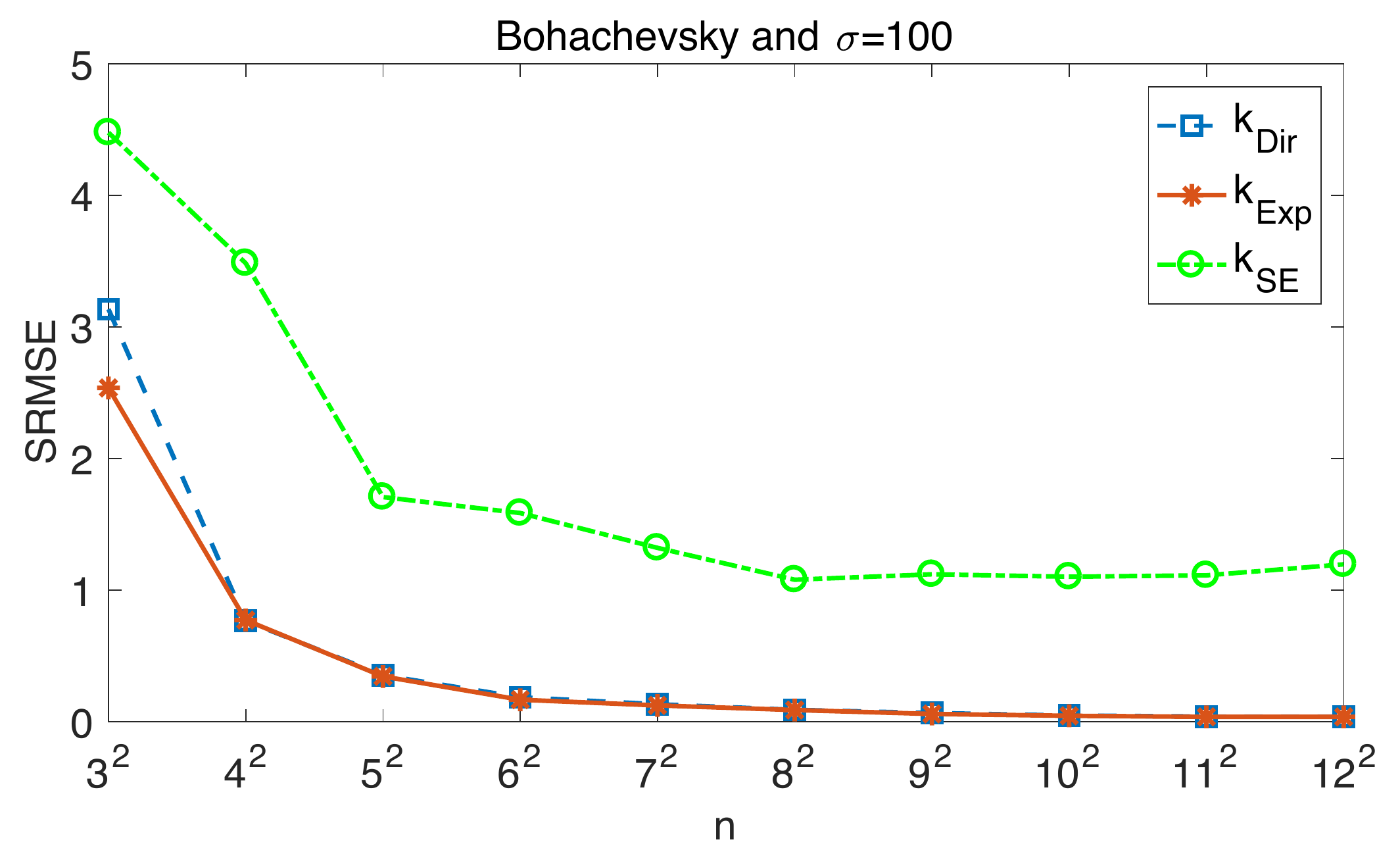} 
\end{array}
$
\end{center}
\end{figure}

Clearly, in the absence of simulation errors, i.e., $\sigma=0$, $k_\mathrm{SE}$ has yields highest prediction accuracy especially when $n$ is small, while $k_{\mathrm{Exp}}$ and $k_{\mathrm{Dir}}$ have almost identical performance. However, when $n$ is large, $k_{\mathrm{SE}}$ will encounter the serious numerical instability issue, as reflected by the sudden ``blow-up'' in SRMSE. This is because for large $n$, e.g., $n>50$, $\BFSigma_\SFM$ becomes highly ill-conditioned during the numerical procedure of solving MLE, in which case the numerical error associated with computing $\BFSigma_\SFM^{-1}$ is overwhelming, and both the parameter estimates and the prediction are unreliable.

On the other hand, in the presence of simulation errors, the numerical instability issue is mitigated greatly and we do not observe the ``blow-up'' behavior in SRMSE in our experiments even for large $n$. This is because the matrix that needs to be inversed now in order to compute the MLE and the SK predictor is $\BFSigma_\SFM+\BFSigma_\varepsilon$, which is far away from being singular despite the ill-condition of $\BFSigma_\SFM$. Nevertheless, the simulation errors degrade the prediction accuracy of SK in general, and $k_{\mathrm{SE}}$ appear to suffer the most. Specifically, the SRMSE associated with $k_{\mathrm{SE}}$ is significantly higher than the other two. The performances of $k_{\mathrm{Exp}}$ and $k_{\mathrm{Dir}}$ are comparable with the former noticeably better.

In order to reveal clearly the possible numerical instability associated matrix inversion, we compute the \emph{condition number} (associated with the $\mathsf{L}^2$ vector norm) of $\BFSigma_\SFM+\BFSigma_\varepsilon$, which measures how roundoff errors during computation impact the entries of the computed inverse matrix; see \citet[Chapter 5.8]{HornJohnson12} for exposition on the subject. The positive definiteness of $\BFSigma_\SFM+\BFSigma_\varepsilon$ implies that its condition number is the ratio between its largest eigenvalue to its smallest eigenvalue. The larger the condition number is, the more ill-conditioned the matrix is. In Figure \ref{fig:CondNum}, we plot the condition number of $\BFSigma_\SFM+\BFSigma_\varepsilon$ with plug-in parameter estimates from MLE for fitting the samples from the three-hump camel function. The plots for the Matyas function and the Bohachevsky function are highly similar, thereby omitted.

\begin{figure}[t]
\begin{center}
\caption{Condition Number of $\BFSigma_\SFM+\BFSigma_\varepsilon$.}\label{fig:CondNum}
\includegraphics[width=0.45\textwidth]{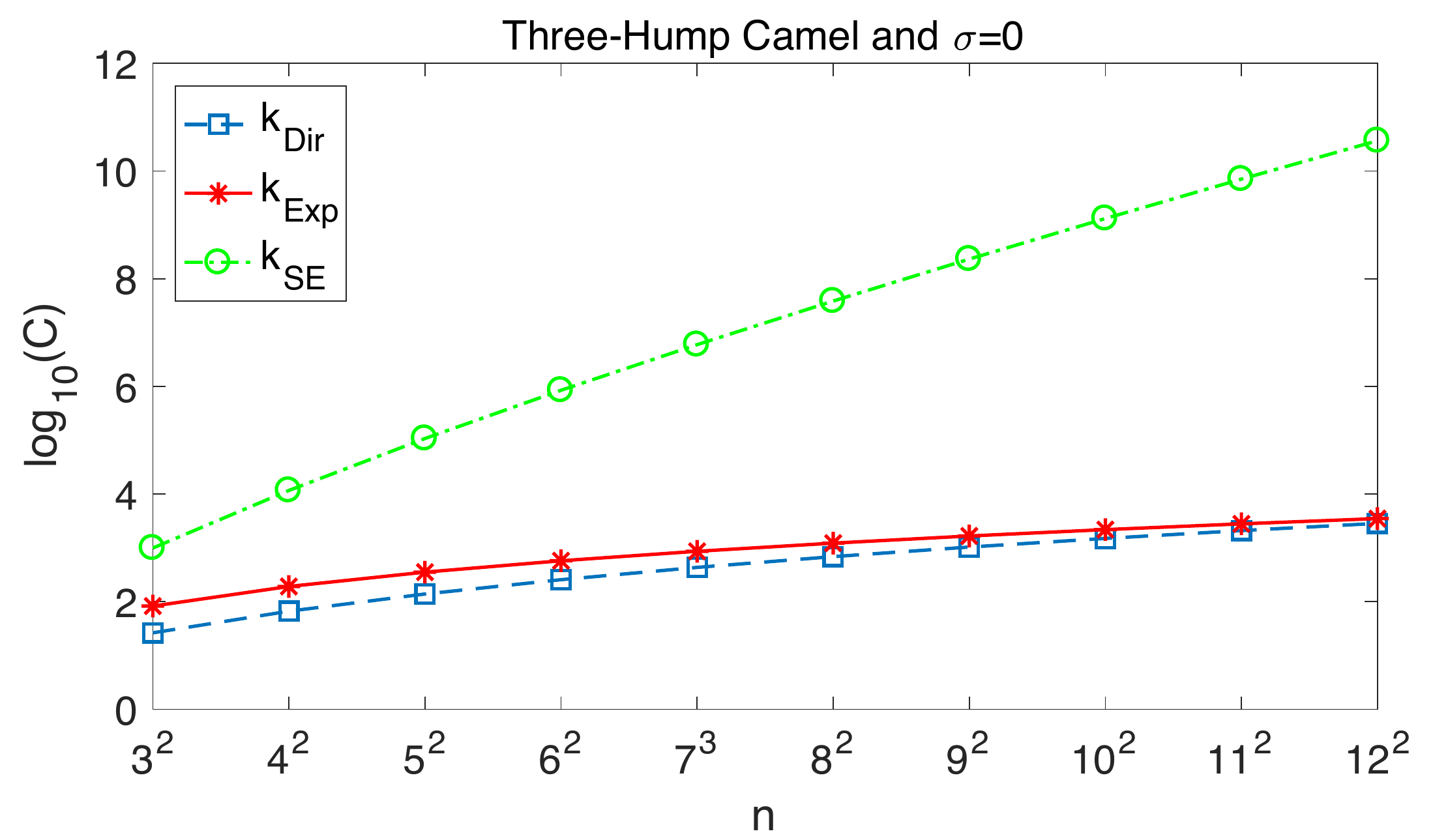} 
\includegraphics[width=0.45\textwidth]{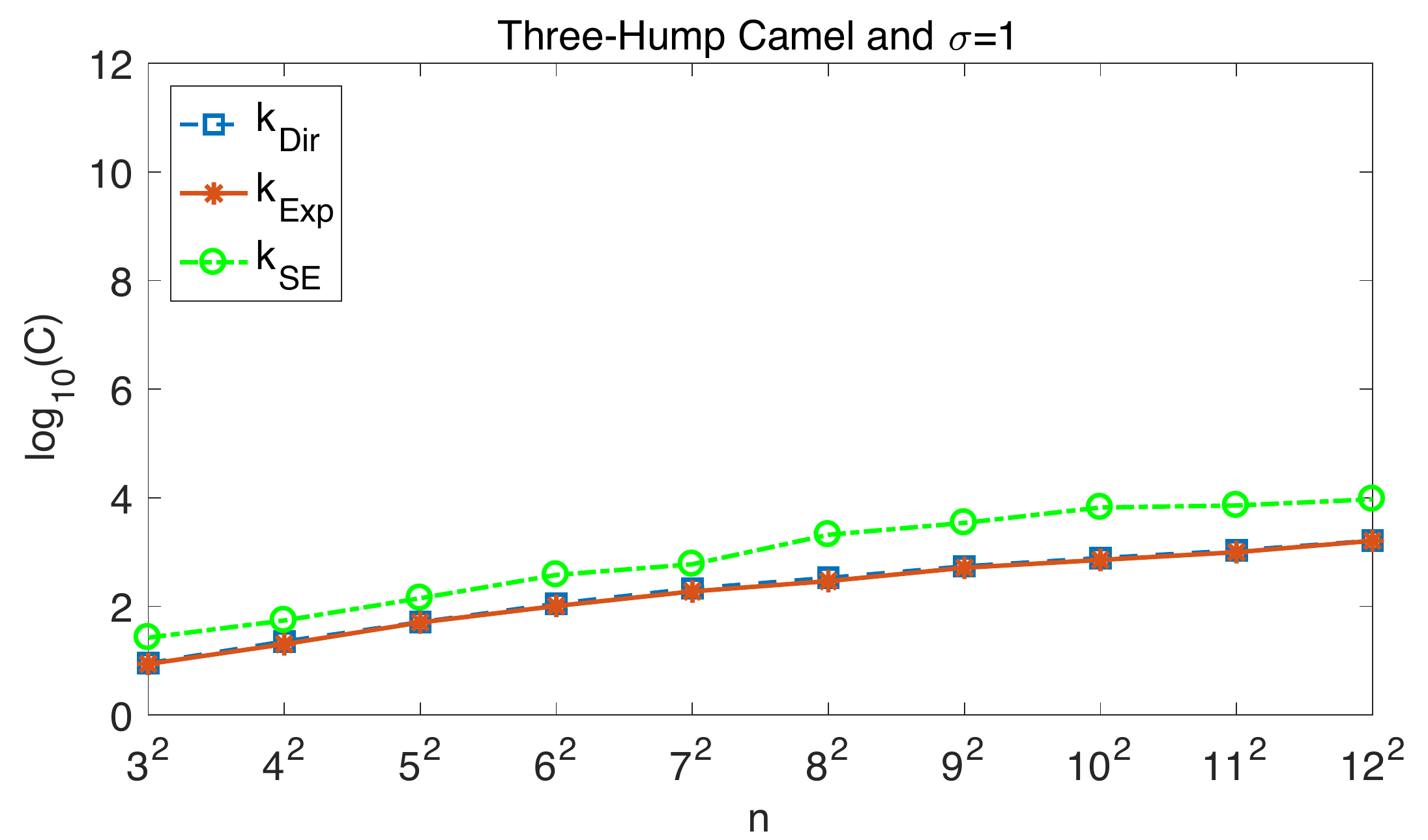}
\end{center}
\small{\textit{Note.} Three-hump camel function; $\BFSigma_\varepsilon=\BFzero$ if $\sigma=0$.}
\end{figure}

Figure \ref{fig:CondNum} shows that the condition number of $\BFSigma_\SFM$ for $k_{\mathrm{SE}}$ basically increases exponentially fast in $n$. For example, it is larger than $10^{10}$ for $n= 12^2$, which means that $\BFSigma_{\SFM}$ is severely ill-conditioned and explains the erroneous prediction results revealed in Figure \ref{fig:Matyas}. By contrast,  the condition number of $\BFSigma_\SFM$ grows dramatically slower for the other two covariance functions. However, in the presence of simulation errors, the condition number of $\BFSigma_\SFM+\BFSigma_\varepsilon$ is reduced substantially, especially for $k_{\mathrm{SE}}$.  Indeed, it has been well documented in geostatistics literature that the condition number of the covariance matrix associated with $k_{\mathrm{SE}}$ is particularly large. A typical treatment is to add artificially the so-called ``nugget effect'' which plays essentially the same role as $\BFSigma_\varepsilon$ mathematically; see, e.g., \cite{AbabouBagtzoglouWood94} and references therein.

\subsection{Scalability Demonstration}

We now demonstrate the scalability of SK when equipped with MCFs. In the experiments that follow, we do not incorporate $k_{\mathrm{SE}}$ in the comparison, because with it SK scales poorly as $n$ increases and almost certainly ends up with numerical failure  as shown in \S\ref{sec:2d_surface}. We consider two response surfaces. One is the Griewank function 
\[\SFZ(\BFx) = \sum_{i=1}^4 \left(\frac{x^{(i)}}{20}\right)^2 - 10\prod_{i=1}^D\cos\left(\frac{x^{(i)}}{\sqrt{i}}\right) + 10, \quad \BFx\in[-5,5]^4,\]
with $D=4$; see Figure \ref{fig:Griewank} (left panel) for its 2-dimensional projections.

\begin{figure}[t]
\begin{center}
\caption{Two-Dimensional Projections of the Griewank Function and the Expected Cycle Time.} \label{fig:Griewank}
\includegraphics[width=0.32\textwidth]{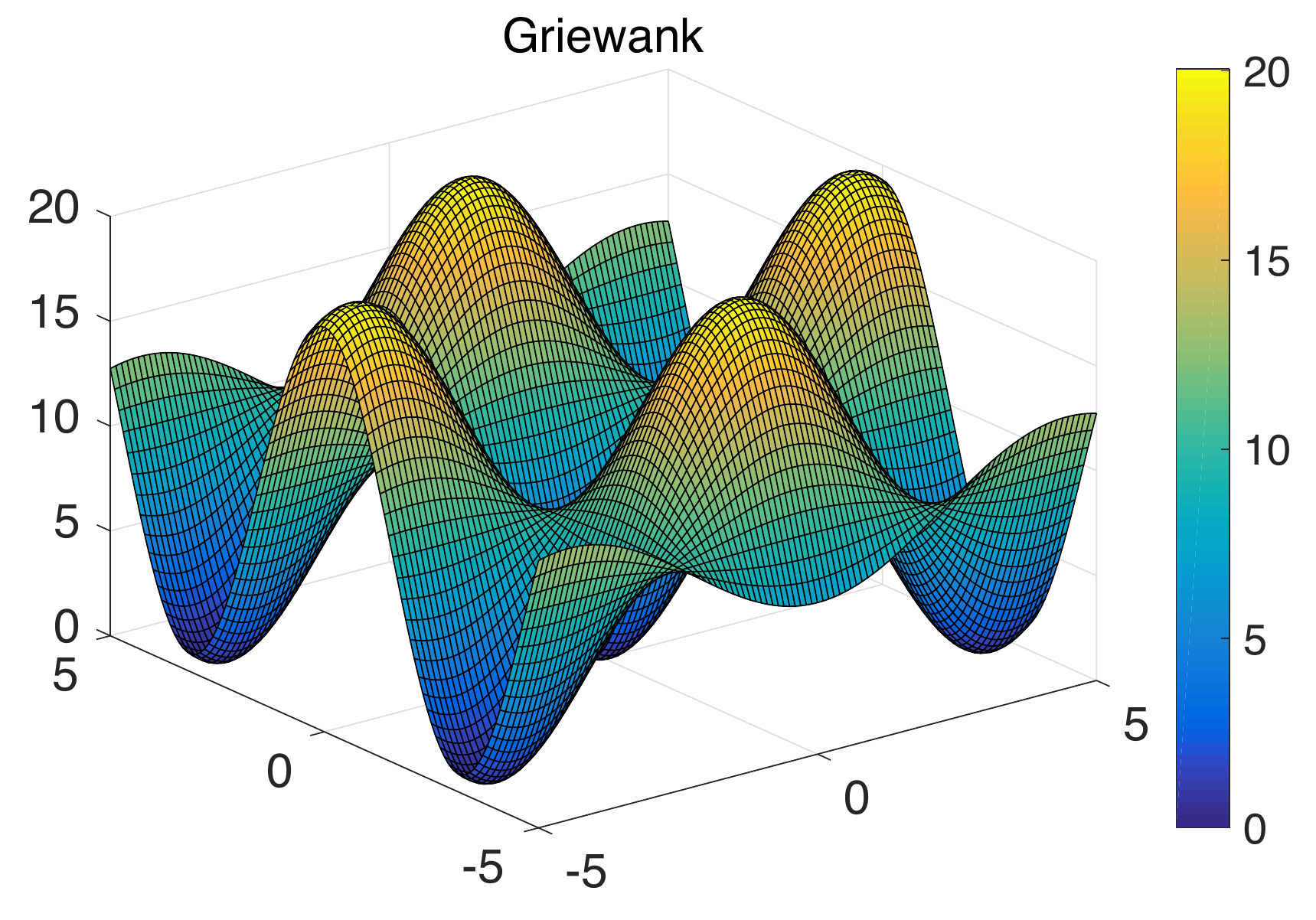}
\includegraphics[width=0.32\textwidth]{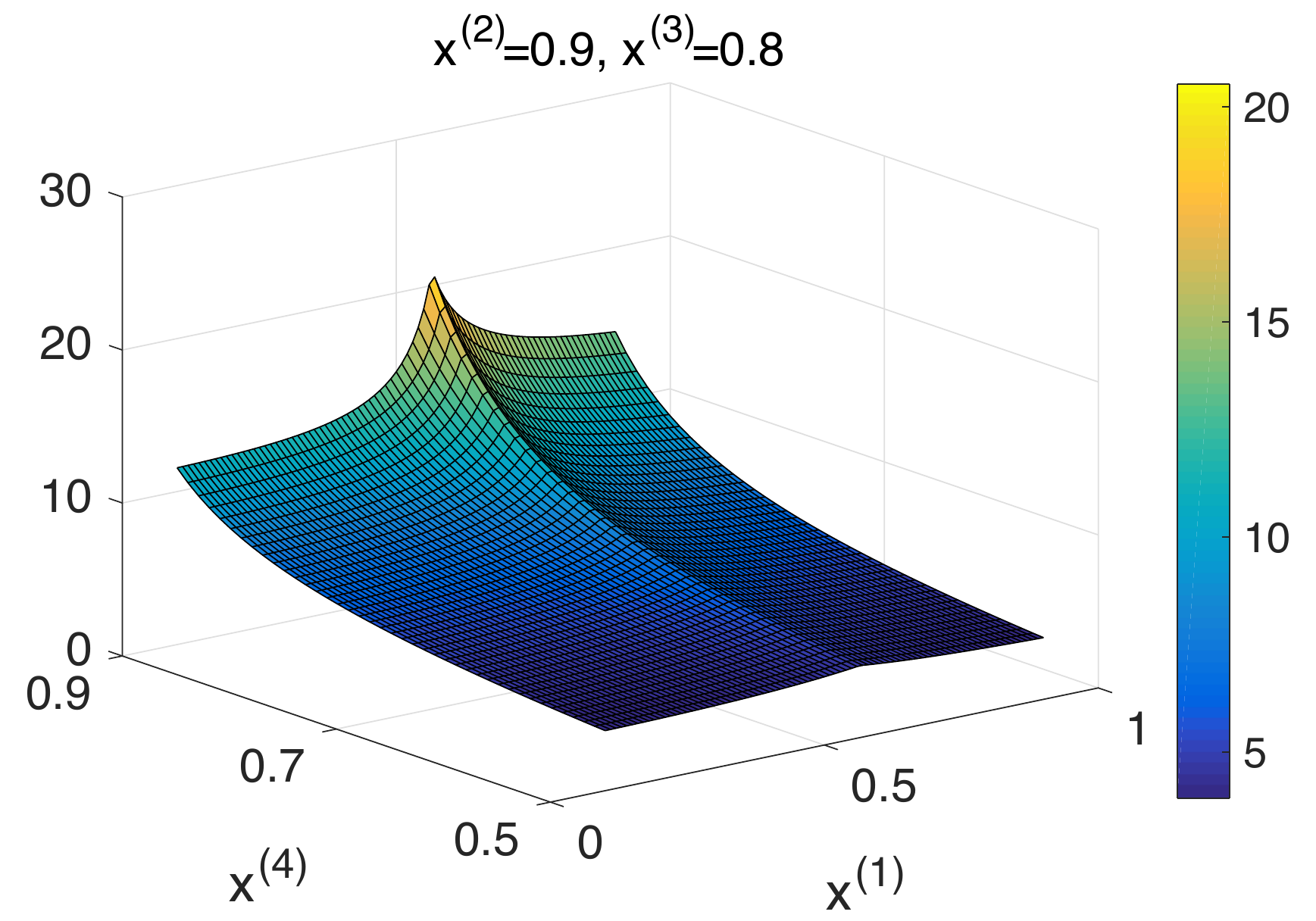}
\includegraphics[width=0.32\textwidth]{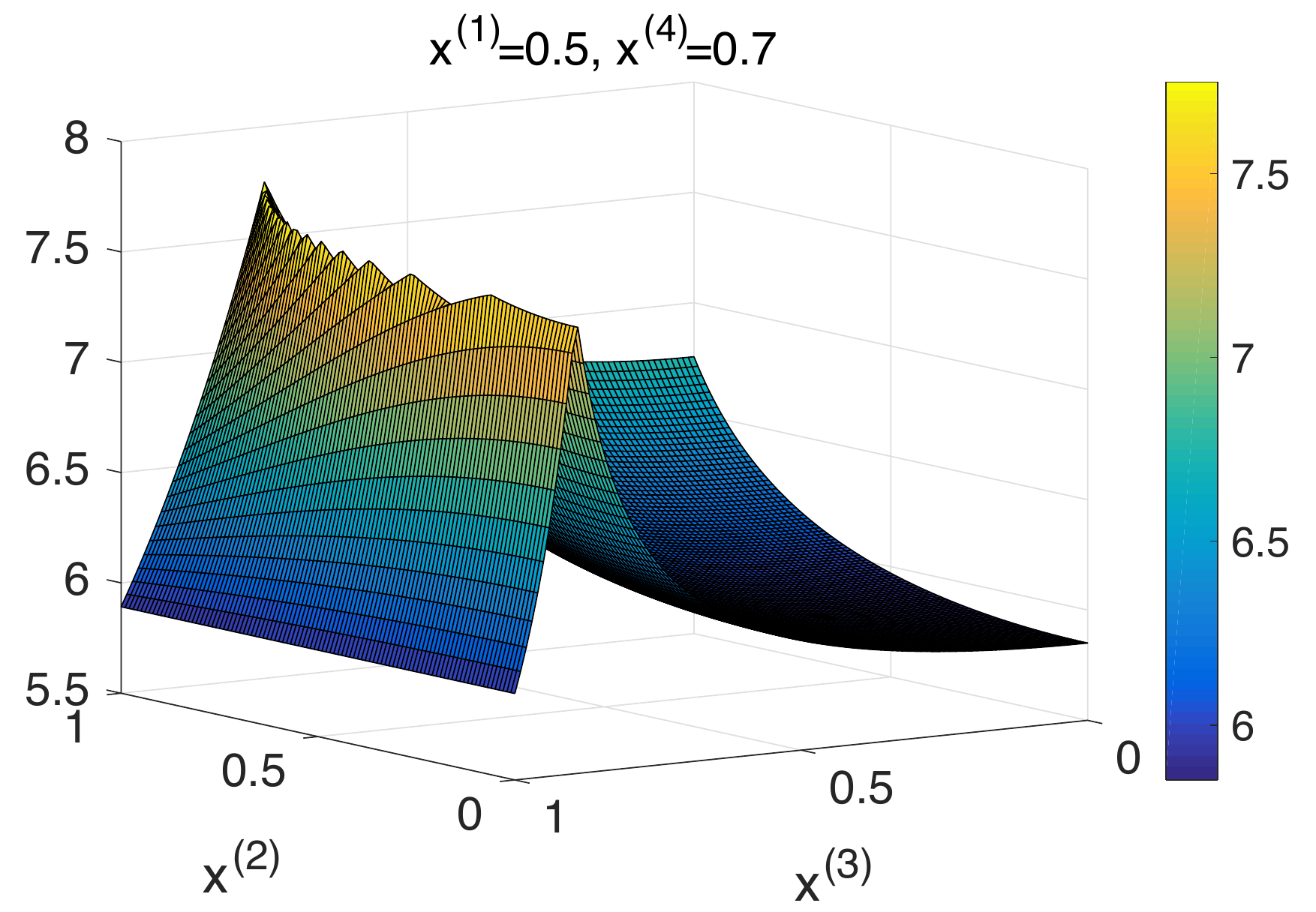} 
\end{center}
\end{figure}

The experiment is set up in the same way as \S\ref{sec:2d_surface}. We choose $n=m^4$ design points and $K=100^4$ prediction points, and make both sets of points equally spaced within the design space. The sampling variance at each design point is set to be $\sigma^2=1$. In addition to the prediction accuracy based on SRMSE, we compare $k_{\mathrm{Dir}}$ and $k_{\mathrm{Exp}}$ in terms of the computational efficiency as well. Notice that the implementation of SK comprises primarily two steps -- parameter estimation and computation of the predictor. Inversion of $\BFSigma_\SFM+\BFSigma_\varepsilon$, which is the scalability bottleneck, is performed repeatedly in the former step. By contrast, given the estimated  parameters, the matrix inversion is a one-time operation and thus can be stored to compute the predictor \eqref{eq:BLUP} at different design points, since the matrix is independent of the design point. 

As discussed in \S\ref{sec:MLE}, $k_{\mathrm{Dir}}$ enjoys a more efficient MLE scheme than general MCFs such as $k_{\mathrm{Exp}}$. We therefore compare their computational efficiency by measuring the CPU time used to solve the MLE. Specifically, we use the Matlab function \texttt{fsolve} to solve numerically the first-order optimality conditions \eqref{eq:MLE_conditions1} and \eqref{eq:MLE_conditions2} for $k_{\mathrm{Exp}}$ and $k_{\mathrm{Dir}}$, respectively. We set the initial point randomly, repeat the experiment 100 times, and compute the average CPU time. The results are presented in Figure \ref{fig:scalability} (upper panel).

A second surface arises from a queueing context and is adopted from \cite{YangLiuNelsonAnkenmanTongarlak11}. Consider a $N$-station Jackson network in which both the interarrival times and the service times are exponentially distributed. The arrivals consist of $D$ different types of products and the fraction of product $i$ is $\alpha_i$, $i=1,\ldots,D$. Suppose that station $j$ has a service rate $\mu_j$, regardless of the product type, $j=1,\ldots,N$. The station having the largest utilization among all is called the bottleneck station. Let $\rho$ denote the utilization of the bottleneck station. The design variable is $(\alpha_1,\ldots,\alpha_D, \rho)$, for $\alpha_i\in[0,1]$ with $\alpha_1+\cdots+\alpha_D=1$ and $\rho\in[0.5,0.9]$. The response surface of interest is the expected cycle time (CT) of, say, product 1. It is shown in \cite{YangLiuNelsonAnkenmanTongarlak11} that 
\begin{equation}\label{eq:CT}
\E[\mathrm{CT}_1]= \sum_{j=1}^N \frac{\delta_{1j}}{\mu_j\left[1- \rho\left(\frac{\sum_{i=1}^D\alpha_i\delta_{ij}/\mu_j}{\max_h \sum_{i=1}^D\alpha_i\delta_{ih}/\mu_h}\right)\right]},
\end{equation}
where $\delta_{i,j}$ is the expected number of visits to station $j$ by product $i$. The parameters $\mu_j$ and $\delta_{i,j}$ are generated randomly and given as follows: 
\[\mu = 
\begin{pmatrix}
1.25 \\
1.85 \\
1.97 \\
1.45
\end{pmatrix},
\quad 
\delta = 
\begin{pmatrix}
1.553 & 1.012 & 0.926 & 0.242 \\
0.127 & 1.066 & 1.115 & 0.536 \\ 
1.182 & 1.597 & 1.486 & 1.850 \\
1.800 & 1.310 & 1.029 & 1.179
\end{pmatrix}.
\]

Notice that the design space is not a hyperrectangle. To accommodate the requirement that the design points form a regular lattice, we conduct the following change of variables. Define 
$x^{(1)}=\sqrt{\alpha_1}$, $x^{(i)}=\sqrt{\alpha_{i}/(1-\sum_{h=1}^{i-1}\alpha_h)}$, $i=2,\ldots,D-2$, $x^{(D)}=\rho$. Then, $x^{(i)}\in[0,1]$ for $i=1,\ldots,D-1$ because $\alpha_1+\cdots+\alpha_D=1$. Let $\BFx=(x^{(1)},\ldots,x^{(D)})\in[0,1]^{D-1}\times[0.5,0.9]$ and $\SFZ(x)$ be the expression of \eqref{eq:CT} after the change of variables; see  Figure \ref{fig:Griewank} (middle and right panels) for its 2-dimensional projections. A critical difference between this surface and the others  is that it is not differentiable everywhere. This is because the bottleneck station varies as the product-mix vector $(\alpha_1,\ldots,\alpha_D)$ changes. 

We assume $D=N=4$. The experiment setup is the same as that for the Griewank function, except that we choose $n=5m^3$ design points as follows:  $x^{(i)}\in\{\frac{1}{m+1},\ldots,\frac{m}{m+1}\}$, $i=1,2,3$, and $x^{(D)}\in\{0.5,0.6,\ldots,0.9\}$, for $m=5,6,\ldots,15$.  The results are presented in Figure \ref{fig:scalability} (lower panel).

\begin{figure}[t]
\begin{center}
\caption{Efficiency for Solving MLE and Prediction Accuracy of SK with MCFs.} \label{fig:scalability}
$
\begin{array}{cc}
\includegraphics[width=0.45\textwidth]{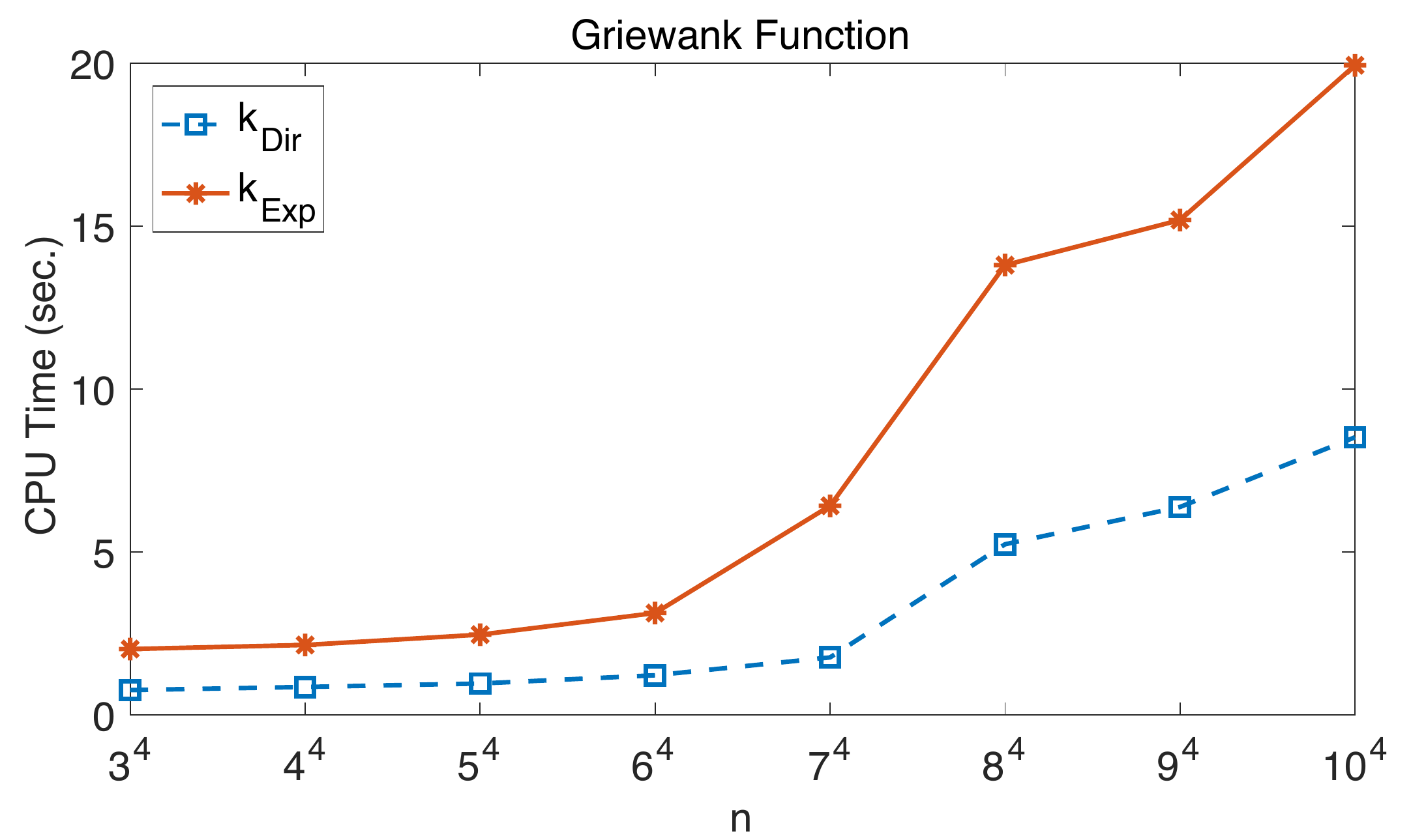} &
\includegraphics[width=0.45\textwidth]{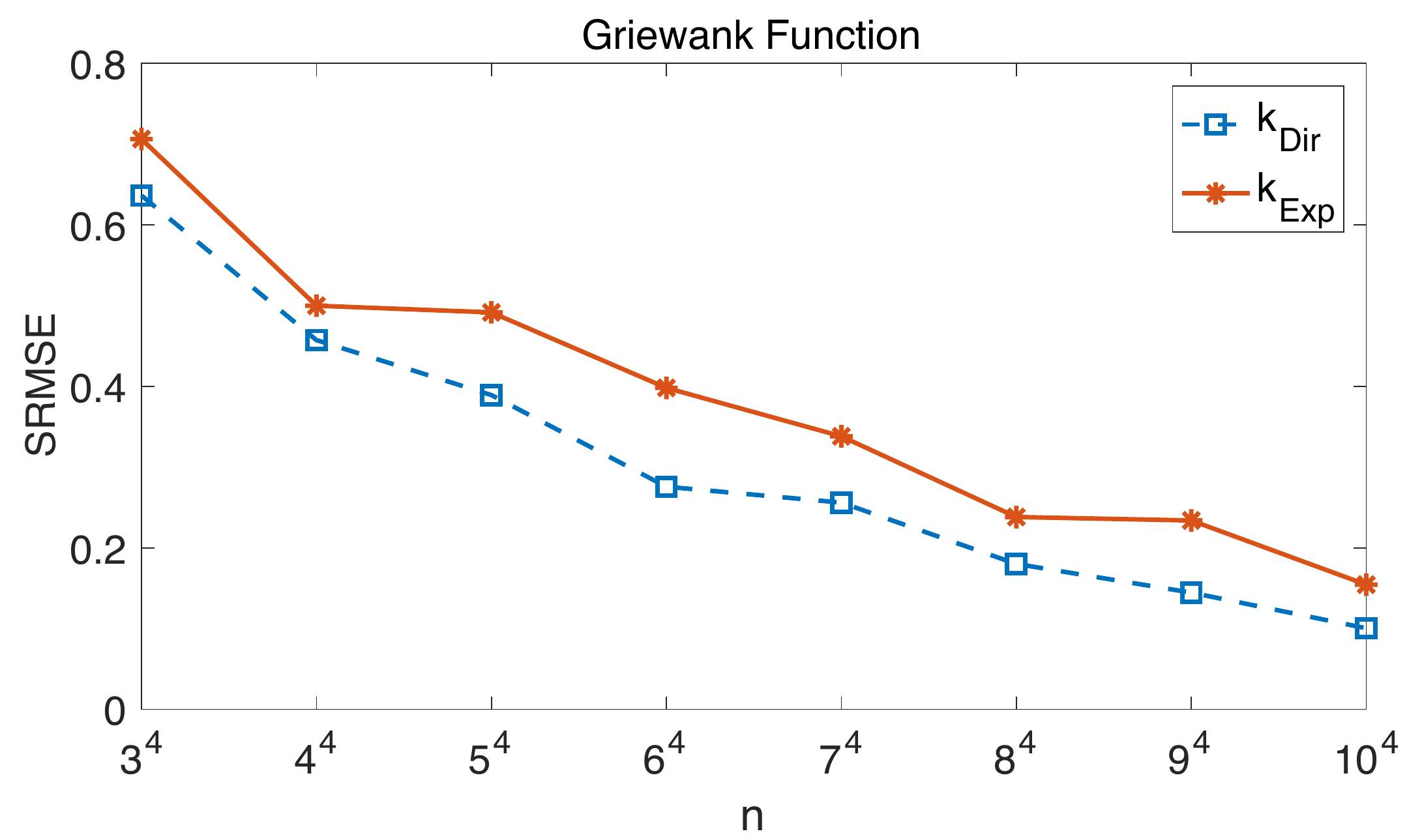}  \\ 
\includegraphics[width=0.45\textwidth]{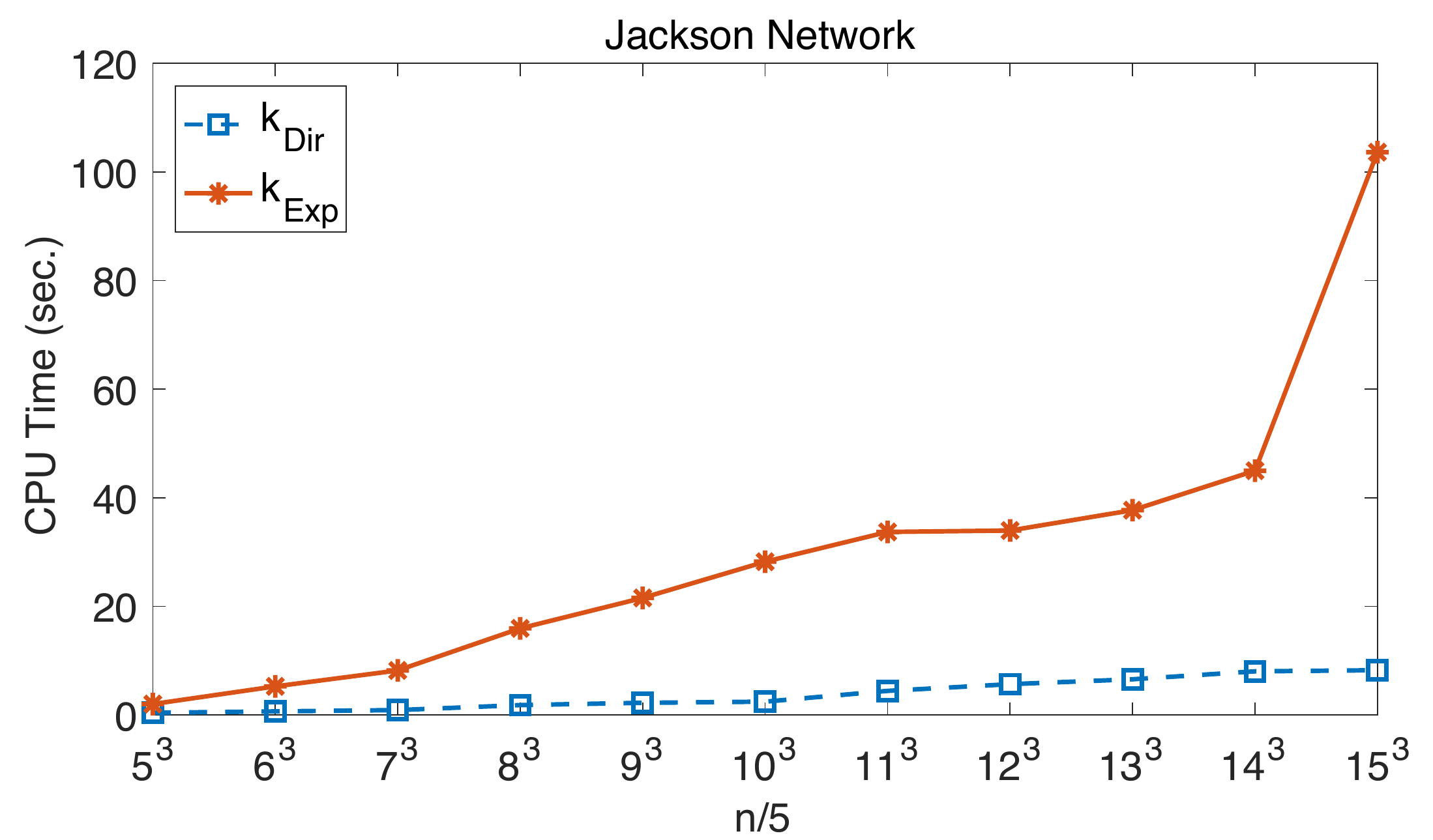} & 
\includegraphics[width=0.45\textwidth]{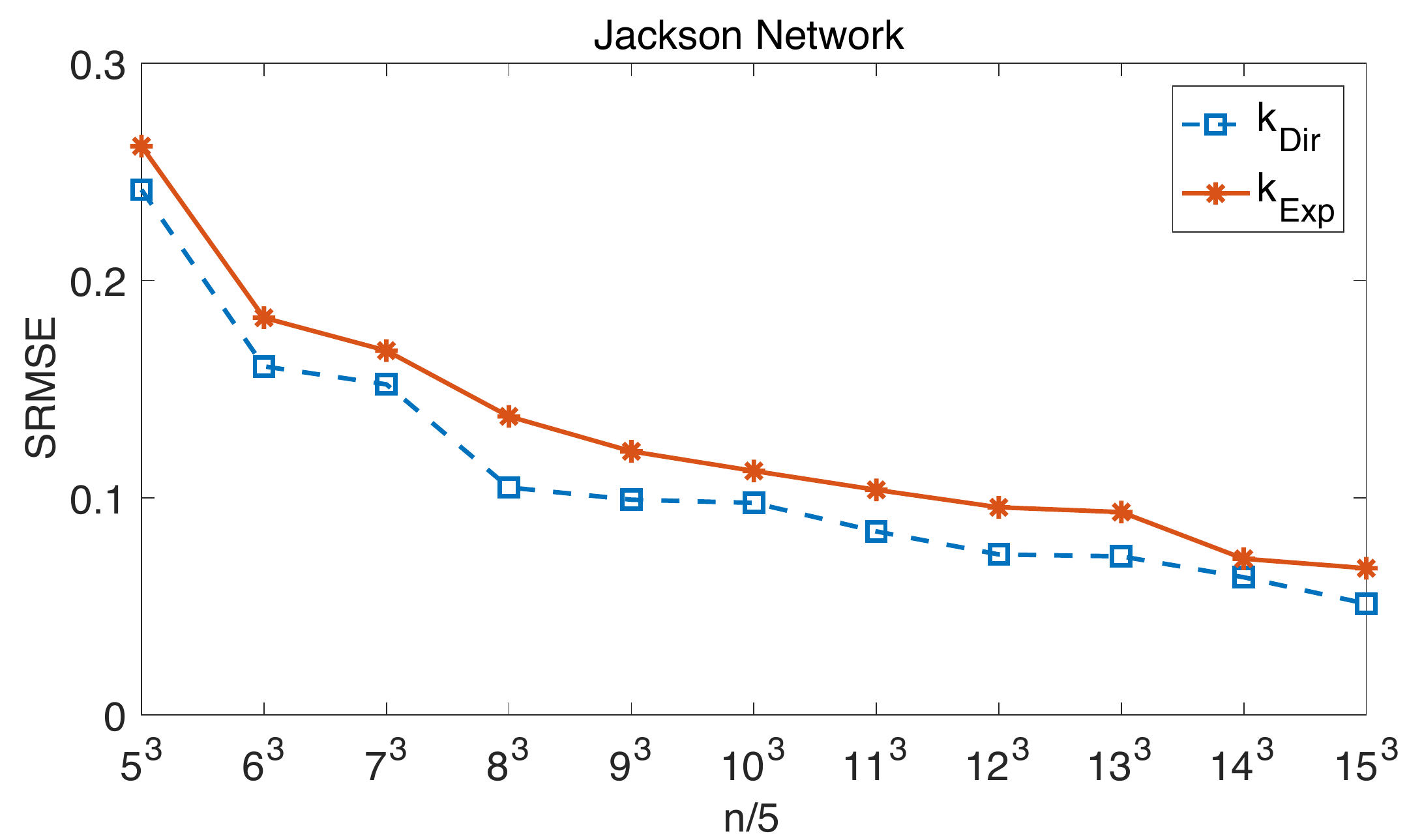} 
\end{array}
$
\end{center}
\end{figure}

We see from Figure \ref{fig:scalability} that SK can scale up dramatically with the use of MCFs. It can easily handle large-scale problems in a computationally efficient and numerically stable fashion. For example, even with $10^4\times 10^4$ covariance matrices, the MLE can be solved within a minute on an average desktop computer and does not encounter any numerical instability issue. This is a consequence of the analytical invertibility and the sparsity structure induced by MCFs. 

Moreover, between the two MCFs tested here, $k_{\mathrm{Dir}}$ outperforms $k_{\mathrm{Exp}}$ substantially in terms of computational efficiency, due to the enhanced MLE scheme in \S\ref{sec:enhanced_MLE}. However, we stress that such enhancement comes at the cost of flexibility in the allocation of simulation budget across design points, because $\BFSigma_\varepsilon$ needs to be in the form of $\sigma^2\BFI$. In terms of prediction accuracy, $k_{\mathrm{Dir}}$ is also noticeably better than $k_{\mathrm{Exp}}$.

\section{Concluding Remarks} \label{sec:conclusions}
The present paper addresses the poor scalability of the popular SK metamodel using a novel approach. By imposing a Markovian structure on the Gaussian random field, we identify the form of the covariance function that leads to analytically invertible covariance matrices with sparsity in the inverse. We further develop a connection between such MCFs and the Green's functions of S-L equations, which effectively provides a flexible, principled approach to constructing MCFs. With the use of MCFs, the computational complexity related to matrix inversion is reduced from $\CalO(n^3)$  to $\CalO(n^2)$ in general without any matrix approximations, to $\CalO(n)$ in the absence of simulation errors, and even to $\CalO(1)$ for some specific MCFs with carefully chosen design points. 

Extensive numerical experiments  demonstrate that for small-scale problems, MCFs have comparable performance as the squared exponential covariance function, a standard choice for SK, in terms of the prediction accuracy; however, the true advantage of MCFs resides in large-scale problems, which can be handled in a timely and stable manner without suffering from the numerical instability issue that SK normally exhibits under general covariance functions. 

Several follow-up problems should be investigated to realize the full potential of the methodology. For example, the condition number of the covariance matrix is examined numerically in the present paper. The observation that MCFs yield a small condition number ought to be addressed theoretically to further strengthen the foundation of the methodology. For another example, using gradient information to enhance the prediction accuracy of SK is a technique that receives much attention; see \cite{ChenAnkenmanNelson13} and \cite{QuFu14}. However, in the presence of the gradient, the size of the covariance matrix that needs to be inverted becomes $(D+1)n\times (D+1)n$, since there is a distinct derivative surface for the partial derivative along each dimension in addition to the response surface itself. Hence, the big $n$ problem is even more severe in this context and our methodology can potentially be of great help.

\section*{Acknowledgment}
The first author is supported by the Hong Kong PhD Fellow Scheme (Ref. No. PF14-13781). The second author is supported by the Hong Kong Research Grant Council (Project No. 16211417).

\appendix
\section{Proof of Proposition \ref{prop:minor}}\label{app:A}
We prove that for each $i=2,\ldots,n$,
\begin{equation}\label{eq:offdiag_minor}
|\BFK(\CalX\setminus\{x_{i-1}\}, \CalX\setminus\{x_i\})| =
p_1q_n\prod_{j=2,j\neq i}^n (p_j q_{j-1}-p_{j-1}q_j),
\end{equation}
by induction on $n$, the size of $\CalX$. The result is trivial for $n=2$. For $n=3$, 
\[|\BFK(\CalX\setminus\{x_1\}, \CalX\setminus\{x_2\})|= 
\begin{vmatrix}
p_1q_2 & p_2q_3 \\
p_1q_3 & p_3q_3
\end{vmatrix}
= p_1q_2p_3q_3 - p_1p_2q_3^2 = p_1q_3(p_3q_2-p_2q_3),
\]
and 
\[
|\BFK(\CalX\setminus\{x_2\}, \CalX\setminus\{x_3\})|=     
\begin{vmatrix}
p_1q_1  & p_1q_2 \\
p_1q_3  & p_2q_3 
\end{vmatrix}    
= p_1p_2q_1q_3 - p_1^2q_2q_3 = p_1q_3(p_2q_1-p_1q_2).
\]

We now suppose that the result holds for all $n\leq N-1$. For $n=N$, we first consider $i=2$. 
\begin{equation*}\label{eq:bigmatrix2}
\BFK(\CalX\setminus\{x_1\}, \CalX\setminus\{x_2\}) =
\kbordermatrix{
 & x_1  & x_3 & \cdots & x_N \\
x_2 & p_1q_2  & p_2q_3 & \cdots & p_2q_N \\
x_3 & p_1q_3 & p_3q_3 & \cdots & p_3q_N \\
\vdots &    &   & \vdots & \\
x_N & p_1q_N & p_3q_N & \cdots & p_Nq_N
}.
\end{equation*}
Applying the Laplace expansion along the first row, 
\begin{align}
&|\BFK(\CalX\setminus\{x_1\}, \CalX\setminus\{x_2\})| \nonumber \\
=& p_1q_2|\BFK(\CalX\setminus\{x_1,x_2\}, \CalX\setminus\{x_2,x_1\})| + \sum_{\ell=3}^N(-1)^{1+(\ell-1)} p_2q_\ell |\BFK(\CalX\setminus\{x_1,x_2\}, \CalX\setminus\{x_2,x_\ell\})|
 \nonumber  \\
= & 
p_1q_2|\BFK(\CalX\setminus\{x_1,x_2\}, \CalX\setminus\{x_2,x_1\})| - p_2q_3|\BFK(\CalX\setminus\{x_1,x_2\}, \CalX\setminus\{x_2,x_3\})|,\label{eq:expansion6}
\end{align}
where the second equality holds because the first two columns (corresponding to $x_1$ and $x_3$) in the minors are linearly dependent for all $\ell\geq 4$ (so that only the first term in the summation is nonzero). We apply Proposition \ref{prop:determinant} with $n=N-2$ to the first summand of  \eqref{eq:expansion6}. For the second summand, we let $\CalX'=\CalX\setminus\{x_2\}$ and relabel its points as $\{x_1,x_3,\ldots,x_N\} = \{x_1',x_2',\ldots,x'_{N-1}\}$. Then, we can apply  the induction assumption to obtain 
\begin{align*}
|\BFK(\CalX\setminus\{x_1,x_2\}, \CalX\setminus\{x_2,x_3\})| =& |\BFK(\CalX'\setminus\{x_1'\}, \CalX\setminus\{x_2'\})| \\ 
=& p'_1q'_{N-1}\prod_{j=3}^{N-1}(p'_jq'_{j-1}-p'_{j-1}q'_j) = p_1q_N\prod_{j=4}^N(p_jq_{j-1}-p_{j-1}q_j),
\end{align*}
where $p'_i=p(x'_i)$ and $q'_i=q(x'_i)$.  Hence, \eqref{eq:expansion6} becomes
\begin{align*}
|\BFK(\CalX\setminus\{x_1\}, \CalX\setminus\{x_2\})| =&  p_1q_2p_3q_N\prod_{j=4}^N(p_jq_{j-1}-p_{j-1}q_j)
- p_2q_3p_1q_N\prod_{j=4}^N(p_jq_{j-1}-p_{j-1}q_j) \\
= &p_1q_N\prod_{j=3}^N(p_jq_{j-1}-p_{j-1}q_j),
\end{align*}
proving \eqref{eq:offdiag_minor} for $n=N$ and $i=2$. 

For $i=3,\ldots,N$, 
\begin{equation}\label{eq:bigmatrix3}
\BFK(\CalX\setminus\{x_{i-1}\}, \CalX\setminus\{x_i\}) =
\kbordermatrix{
 & x_1  & x_2 & \cdots &  x_{i-1} & x_{i+1} & \cdots  & x_N \\
x_1 & p_1q_1  & p_1q_2 & \cdots &  p_1q_{i-1} & p_1q_{i+1} & \cdots & p_1q_N \\
x_2 & p_1q_2 & p_2q_2 & \cdots & p_2q_{i-1} & p_2q_{i+1} & \cdots & p_2q_N \\
\vdots& &  & \vdots &  &   & \vdots &\\
x_{i-2} &p_1q_{i-2} & p_2q_{i-2} & \cdots & p_{i-2}q_{i-1} & p_{i-2}q_{i+1} & \cdots & p_{i-2}q_N \\
x_{i} &p_1q_{i} & p_2q_{i} & \cdots & p_{i-1}q_{i} & p_{i}q_{i+1} & \cdots & p_{i}q_N \\
\vdots& &  & \vdots &  &   & \vdots &\\
x_N & p_1q_N & p_2q_N & \cdots & p_{i-1}q_N & p_{i+1}q_N & \cdots & p_Nq_N
}.
\end{equation}
Applying the Laplace expansion along the first row and using the same argument as the one leading to \eqref{eq:expansion6}, we obtain 
\begin{align}
&|\BFK(\CalX\setminus\{x_{i-1}\}, \CalX\setminus\{x_i\})| \nonumber \\
= & 
p_1q_1|\BFK(\CalX\setminus\{x_{i-1},x_1\}, \CalX\setminus\{x_i,x_1\})| - p_1q_2|\BFK(\CalX\setminus\{x_{i-1},x_1\}, \CalX\setminus\{x_i,x_2\})|.\label{eq:expansion7}
\end{align}
Notice that the first column of $\BFK(\CalX\setminus\{x_{i-1},x_1\}, \CalX\setminus\{x_i,x_2\})$ is a multiple of the first column of $\BFK(\CalX\setminus\{x_{i-1},x_1\}, \CalX\setminus\{x_i,x_1\})$. Hence, if $p_2\neq 0$, then 
\begin{equation}\label{eq:first_col_multiple}
|\BFK(\CalX\setminus\{x_{i-1},x_1\}, \CalX\setminus\{x_i,x_2\})| = \frac{p_1}{p_2}|\BFK(\CalX\setminus\{x_{i-1},x_1\}, \CalX\setminus\{x_i,x_1\})|.
\end{equation}
Moreover, by the induction assumption, i.e., applying \eqref{eq:offdiag_minor} to $\CalX\setminus\{x_1\}=\{x_2,\ldots,x_n\}$,
\begin{equation}\label{eq:induction_result}
|\BFK(\CalX\setminus\{x_{i-1},x_1\}, \CalX\setminus\{x_i,x_1\})| = p_2q_N\prod_{j=3,j\neq i}^N(p_j q_{j-1}-p_{j-1}q_j).
\end{equation}
Combining \eqref{eq:expansion7}, \eqref{eq:first_col_multiple}, and \eqref{eq:induction_result} yields 
\begin{align*}
|\BFK(\CalX\setminus\{x_{i-1}\}, \CalX\setminus\{x_i\})| =& \left(p_1q_1-\frac{p_1^2q_2}{p_2}\right) p_2q_n\prod_{j=3,j\neq i}^N(p_j q_{j-1}-p_{j-1}q_j)\\ 
=& p_1q_N\prod_{j=2,j\neq i}^N(p_j q_{j-1}-p_{j-1}q_j).
\end{align*}

On the other hand, if $p_2=0$, then we apply the Laplace expansion to $\BFK(\CalX\setminus\{x_{i-1},x_1\}, \CalX\setminus\{x_i,x_2\})$ along its first row, i.e., the row corresponding to $x_2$. Then, 
\begin{align}
|\BFK(\CalX\setminus\{x_{i-1},x_1\}, \CalX\setminus\{x_i,x_2\})| 
=& p_1q_2 |\BFK(\CalX\setminus\{x_{i-1},x_1,x_2\}, \CalX\setminus\{x_i,x_2,x_1\})| \nonumber \\
=& p_1q_2 p_3q_N\prod_{j=4,j\neq i}^N (p_j q_{j-1}-p_{j-1}q_j), \label{eq:induction_result2}
\end{align}
where the second equality follows from the induction assumption. Moreover, since $p_2=0$, $|\BFK(\CalX\setminus\{x_{i-1},x_1\}, \CalX\setminus\{x_i,x_1\})| = 0$ by \eqref{eq:induction_result}. It then follows from \eqref{eq:expansion7} and \eqref{eq:induction_result2} that, since $p_2=0$,
\begin{align*}
|\BFK(\CalX\setminus\{x_{i-1}\}, \CalX\setminus\{x_i\})| 
= &- p_1^2q_2^2p_3q_N \prod_{j=4,j\neq i}^N(p_j q_{j-1}-p_{j-1}q_j) = p_1q_N \prod_{j=2,j\neq i}^N(p_j q_{j-1}-p_{j-1}q_j).
\end{align*}

\section{Proof of Corollary \ref{cor:Dirichlet}}\label{app:B}
Without loss of generality, we assume that $\eta^2=1$. The proof is based on Theorem \ref{theo:inverse} and explicit calculations of the nonzero entries of $\BFG^{-1}$. In particular, it can be shown that regardless of the sign of $\nu$, the value of $p_iq_{i-1}-p_{i-1}q_i$ is a constant, independent of $i$. It then follows from Theorem \ref{theo:inverse} that we can take 
\[a = \frac{1}{p_2q_1-p_1q_2}, \; b=\frac{p_2}{p_1},\; c=\frac{p_{i+1}q_{i-1}-p_{i-1}q_{i+1}}{p_{i+1}q_i-p_iq_{i+1}},\;\mbox{and }d = \frac{q_{n-1}}{q_n},\]
to obtain the desired expression of $\BFG^{-1}$. 

There are six cases in total, depending on the sign of $\nu$ and the BC. We demonstrate the calculation only for the case corresponding to  the Dirichlet BC and  $\nu<0$. The calculation involved in the other cases is either simpler or highly similar, so we omit the details. 

Specifically, let $p(x)= \sin(\gamma x)$ and $q(x) = \sin(\gamma(1-x))$. 
Then, 
\[b = \frac{p_2}{p_1} = \frac{\sin(\gamma (x_1+h))}{\sin(\gamma x_1)},\quad d =\frac{q_{n-1}}{q_n} = \frac{\sin(\gamma(1-x_n+h))}{\sin(\gamma(1-x_n))},\]
and 
\begin{align*}
p_iq_{i-1}-p_iq_{i-1}&= \sin(\gamma ih)\sin(\gamma -\gamma(i-1)h)-\sin(\gamma(i-1)h)\sin(\gamma -\gamma ih)\\
&=\sin(\gamma ih)[\sin(\gamma)\cos(\gamma(i-1)h)-\cos(\gamma)\sin(\gamma(i-1)h)] \\
& -\sin(\gamma(i-1)h)[\sin(\gamma)\cos(\gamma ih)-\cos(\gamma)\sin(\gamma ih)]\\
&=\sin(\gamma) [\sin(\gamma ih)\cos(\gamma(i-1)h)-\cos(\gamma ih)\sin(\gamma(i-1)h)] = \sin(\gamma) \sin(\gamma h),
\end{align*}
for $i=2,\ldots,n$. Likewise, $p_{i+1}q_{i-1}-p_{i+1}q_{i-1} =\sin(\gamma) \sin(2\gamma h)$, for $i=2,\ldots,n-1$. Hence, 
\begin{align*}
a =& \frac{1}{p_2q_1-p_1q_2} = \frac{1}{\sin(\gamma) \sin(\gamma h)}, \\[1ex]
c=& \frac{p_{i+1}q_{i-1}-p_{i-1}q_{i+1}}{p_{i+1}q_i-p_iq_{i+1}} = \frac{\sin(2\gamma h)}{\sin(\gamma h)} = 2\cos(\gamma h).
\end{align*}

\bibliographystyle{chicago}
\bibliography{ScalableSK}

\begin{thebibliography}{}

\bibitem[\protect\citeauthoryear{Ababou, Bagtzoglou, and Wood}{Ababou
  et~al.}{1994}]{AbabouBagtzoglouWood94}
Ababou, R., A.~C. Bagtzoglou, and E.~F. Wood (1994).
\newblock On the condition number of covariance matrices in kriging,
  estimation, and simulation of random fields.
\newblock {\em Math. Geol.\/}~{\em 26\/}(1), 99--133.

\bibitem[\protect\citeauthoryear{Ankenman, Nelson, and Staum}{Ankenman
  et~al.}{2010}]{AnkenmanNelsonStaum10}
Ankenman, B., B.~L. Nelson, and J.~Staum (2010).
\newblock Stochastic kriging for simulation metamodeling.
\newblock {\em Oper. Res.\/}~{\em 58\/}(2), 371--382.

\bibitem[\protect\citeauthoryear{Arfken, Weber, and Harris}{Arfken
  et~al.}{2012}]{ArfkenWeberHarris12}
Arfken, G.~B., H.~J. Weber, and F.~E. Harris (2012).
\newblock {\em Mathematical Methods for Physicists: A Comprehensive Guide.\/}
  (7th ed.).
\newblock Academic Press.

\bibitem[\protect\citeauthoryear{Banerjee, Carlin, and Gelfand}{Banerjee
  et~al.}{2014}]{BanerjeeCarlinGelfand14}
Banerjee, S., B.~P. Carlin, and A.~E. Gelfand (2014).
\newblock {\em Hierarchical Modeling and Analysis for Spatial Data\/} (2nd
  ed.).
\newblock CRC Press.

\bibitem[\protect\citeauthoryear{Barton and Meckesheimer}{Barton and
  Meckesheimer}{2006}]{BartonMeckesheimer06}
Barton, R.~R. and M.~Meckesheimer (2006).
\newblock Metamodel-based simulation optimization.
\newblock In S.~Henderson and B.~Nelson (Eds.), {\em Handbooks in Operations
  Research and Management Science}, Volume~13, Chapter~18, pp.\  535--574.
  Elsevier.

\bibitem[\protect\citeauthoryear{Barton, Nelson, and Xie}{Barton
  et~al.}{2014}]{BartonNelsonXie14}
Barton, R.~R., B.~L. Nelson, and W.~Xie (2014).
\newblock Quantifying input uncertainty via simulation confidence intervals.
\newblock {\em INFORMS J. Comput.\/}~{\em 26\/}(1), 74--87.

\bibitem[\protect\citeauthoryear{Chen, Ankenman, and Nelson}{Chen
  et~al.}{2012}]{ChenAnkenmanNelson12}
Chen, X., B.~Ankenman, and B.~L. Nelson (2012).
\newblock The effects of common random numbers on stochastic kriging
  metamodels.
\newblock {\em ACM Trans. Model. Comput. Simul.\/}~{\em 22\/}(2), Article 7.

\bibitem[\protect\citeauthoryear{Chen, Ankenman, and Nelson}{Chen
  et~al.}{2013}]{ChenAnkenmanNelson13}
Chen, X., B.~Ankenman, and B.~L. Nelson (2013).
\newblock Enhancing stochastic kriging metamodels with gradient estimators.
\newblock {\em Oper. Res.\/}~{\em 61\/}(2), 512--528.

\bibitem[\protect\citeauthoryear{Dolph and Woodbury}{Dolph and
  Woodbury}{1952}]{DolphWoodbury52}
Dolph, C.~L. and M.~A. Woodbury (1952).
\newblock On the relation between {Green's} functions and covariances of
  certain stochastic processes and its application to unbiased linear
  prediction.
\newblock {\em Trans. Amer. Math. Soc.\/}~{\em 72}, 519--550.

\bibitem[\protect\citeauthoryear{Fang, Li, and Sudjianto}{Fang
  et~al.}{2006}]{fang2006}
Fang, K.-T., R.~Li, and A.~Sudjianto (2006).
\newblock {\em Design and Modeling for Computer Experiments}.
\newblock CRC Press.

\bibitem[\protect\citeauthoryear{Horn and Johnson}{Horn and
  Johnson}{2012}]{HornJohnson12}
Horn, R.~A. and C.~R. Johnson (2012).
\newblock {\em Matrix Analysis\/} (2nd ed.).
\newblock Cambridge University Press.

\bibitem[\protect\citeauthoryear{Huang, Allen, Notz, and Zeng}{Huang
  et~al.}{2006}]{huang2006}
Huang, D., T.~T. Allen, W.~I. Notz, and N.~Zeng (2006).
\newblock Global optimization of stochastic black-box systems via sequential
  kriging meta-models.
\newblock {\em J. Glob. Optim.\/}~{\em 34\/}(3), 441--466.

\bibitem[\protect\citeauthoryear{Laub}{Laub}{2005}]{Laub05}
Laub, A.~J. (2005).
\newblock {\em Matrix Analysis for Scientists and Engineers}.
\newblock SIAM.

\bibitem[\protect\citeauthoryear{Noschese, Pasquini, and Reichel}{Noschese
  et~al.}{2013}]{NoschesePasquiniReichel13}
Noschese, S., L.~Pasquini, and L.~Reichel (2013).
\newblock Tridiagonal {Toeplitz} matrices: Properties and novel applications.
\newblock {\em Numer. Lin. Algebra Appl.\/}~{\em 20}, 302--236.

\bibitem[\protect\citeauthoryear{Qu and Fu}{Qu and Fu}{2014}]{QuFu14}
Qu, H. and M.~C. Fu (2014).
\newblock Gradient extrapolated stochastic kriging.
\newblock {\em ACM Trans. Model. Comput. Simul.\/}~{\em 24\/}(4), 23:1--23:25.

\bibitem[\protect\citeauthoryear{{Qui\~{n}onero-Candela} and
  Rasmussen}{{Qui\~{n}onero-Candela} and
  Rasmussen}{2005}]{Quinonero-CandelaRasmussen05}
{Qui\~{n}onero-Candela}, J. and C.~E. Rasmussen (2005).
\newblock A unifying view of sparse approximate {Gaussian} process regression.
\newblock {\em J. Mach. Learn. Res.\/}~{\em 6}, 1939--1959.

\bibitem[\protect\citeauthoryear{Rasmussen and Williams}{Rasmussen and
  Williams}{2006}]{RasmussenWilliams06}
Rasmussen, C.~E. and K.~I. Williams (2006).
\newblock {\em Gaussian Processes for Machine Learning}.
\newblock MIT Press.

\bibitem[\protect\citeauthoryear{Rue and Held}{Rue and Held}{2005}]{RueHeld05}
Rue, H. and L.~Held (2005).
\newblock {\em Gaussian Markov Random Fields: Theory and Applications}.
\newblock CRC Press.

\bibitem[\protect\citeauthoryear{Salemi, Song, Nelson, and Staum}{Salemi
  et~al.}{2017}]{SalemiSongNelsonStaum17}
Salemi, P., E.~Song, B.~L. Nelson, and J.~Staum (2017).
\newblock Gaussian {Markov} random fields for discrete optimization via
  simulation: Framework and algorithms.
\newblock \textit{Oper. Res.}, forthcoming.

\bibitem[\protect\citeauthoryear{Sampson}{Sampson}{2010}]{Sampson10}
Sampson, P.~D. (2010).
\newblock Constructions for nonstationary spatial processes.
\newblock In A.~E. Gelfand, P.~J. Diggle, M.~Fuentes, and P.~Guttorp (Eds.),
  {\em Handbook of Spatial Statistics}, Chapter~9, pp.\  119--130. CRC Press.

\bibitem[\protect\citeauthoryear{Sang and Huang}{Sang and
  Huang}{2012}]{SangHuang12}
Sang, H. and J.~Z. Huang (2012).
\newblock A full scale approximation of covariance functions for large spatial
  data sets.
\newblock {\em J. R. Statist. Soc. B\/}~{\em 74\/}(1), 111--132.

\bibitem[\protect\citeauthoryear{Shahriari, Swersky, Wang, Adams, and {de
  Freitas}}{Shahriari et~al.}{2016}]{ShahriariSwerskyWangAdamsdeFreitas16}
Shahriari, B., K.~Swersky, Z.~Wang, R.~P. Adams, and N.~{de Freitas} (2016).
\newblock Taking the human out of the loop: A review of {Bayesian}
  optimization.
\newblock {\em Proc. IEEE\/}~{\em 104\/}(1), 148 -- 175.

\bibitem[\protect\citeauthoryear{Sun, Hong, and Hu}{Sun et~al.}{2014}]{sun2014}
Sun, L., L.~J. Hong, and Z.~Hu (2014).
\newblock Balancing exploitation and exploration in discrete optimization via
  simulation through a {G}aussian process-based search.
\newblock {\em Oper. Res.\/}~{\em 62\/}(6), 1416--1438.

\bibitem[\protect\citeauthoryear{Teschl}{Teschl}{2012}]{Teschl12}
Teschl, G. (2012).
\newblock {\em Ordinary Differential Equations and Dynamical Systems}.
\newblock American Mathematical Society.

\bibitem[\protect\citeauthoryear{Xie, Nelson, and Barton}{Xie
  et~al.}{2014}]{XieNelsonBarton14}
Xie, W., B.~L. Nelson, and R.~R. Barton (2014).
\newblock A {Bayesian} framework for quantifying uncertainty in stochastic
  simulation.
\newblock {\em Oper. Res.\/}~{\em 62\/}(6), 1439--1452.

\bibitem[\protect\citeauthoryear{Yang, Liu, Nelson, Ankenman, and
  Tongarlak}{Yang et~al.}{2011}]{YangLiuNelsonAnkenmanTongarlak11}
Yang, F., J.~Liu, B.~L. Nelson, B.~E. Ankenman, and M.~Tongarlak (2011).
\newblock Metamodelling for cycle time-throughput-product mix surfaces using
  progressive model fitting.
\newblock {\em Prod. Plan. Control\/}~{\em 22\/}(1), 50--68.

\bibitem[\protect\citeauthoryear{Zaitsev and Polyanin}{Zaitsev and
  Polyanin}{2002}]{ODEHandbook}
Zaitsev, V.~F. and A.~D. Polyanin (2002).
\newblock {\em Handbook of Exact Solutions for Ordinary Differential
  Equations\/} (2 ed.).
\newblock CRC Press.

\end{thebibliography}
\end{document}